\documentclass{article}

\usepackage[round]{natbib}

\usepackage[utf8]{inputenc} % allow utf-8 input
\usepackage[T1]{fontenc}    % use 8-bit T1 fonts
\usepackage{hyperref}       % hyperlinks
\usepackage{url}            % simple URL typesetting
\usepackage{booktabs}       % professional-quality tables
\usepackage{amsfonts}       % blackboard math symbols
\usepackage{nicefrac}       % compact symbols for 1/2, etc.
\usepackage{microtype}      % microtypography

\usepackage{amsmath}
\usepackage{amssymb}
\usepackage{amsthm}
\usepackage{mathrsfs}
\usepackage{graphicx}
\usepackage{ulem}
\usepackage{algorithmic}
\usepackage{algorithm}
\usepackage{bbm}

\usepackage{multirow}

\usepackage{xcolor}

\DeclareMathOperator{\E}{\mathbb{E}}

\title{Disjunct Support Spike and Slab Priors for Variable Selection in Regression under Quasi-sparseness}

\author{
  Daniel Andrade \\ % \footnote{The first author is also affiliated with NEC Corporation.} 
  \footnotesize Graduate University of Advanced Studies (SOKENDAI) \\
     \footnotesize  10-3 Midoricho, Tachikawa, Tokyo, 190-8562, Japan \\
    \footnotesize  daniel.andrade.silva@gmail.com
  \and 
  Kenji Fukumizu  \\
  \footnotesize    The Institute of Statistical Mathematics \\
  \footnotesize  10-3 Midoricho, Tachikawa, Tokyo, 190-8562, Japan \\
}

\DeclareMathOperator*{\argmax}{arg\,max}

\newtheorem{lemma}{Lemma}
\newtheorem{theorem}{Theorem}

\begin{document}

\maketitle

\begin{abstract}
Sparseness of the regression coefficient vector is often a desirable property, since, among other benefits, sparseness improves interpretability. 
In practice, many true regression coefficients might be negligibly small, but non-zero, which we refer to as quasi-sparseness.
Spike-and-slab priors as introduced in \citep{chipman2001practical} can be tuned to ignore very small regression coefficients, and, as a consequence provide a trade-off between prediction accuracy and interpretability. 
However, spike-and-slab priors with full support lead to inconsistent Bayes factors, in the sense that the Bayes factors of any two models are bounded in probability.
This is clearly an undesirable property for Bayesian hypotheses testing, where we wish that increasing sample sizes lead to increasing Bayes factors favoring the true model.
The moment matching priors as in \citep{johnson2012bayesian} can resolve this issue, but are unsuitable for the quasi-sparse setting due to their full support outside the exact value 0.
As a remedy, we suggest disjunct support spike and slab priors, for which we prove consistent Bayes factors in the quasi-sparse setting, and show experimentally fast growing Bayes factors favoring the true model.
Several experiments on simulated and real data confirm the usefulness of our proposed method to identify models with high effect size, while leading to better control over false positives than hard-thresholding.
\end{abstract}

\section{Introduction}
Sparseness of the regression coefficient vector is often a desirable property, since it (1) helps to improve interpretability, and (2) reduces the cost\footnote{In case where acquiring the value of a covariate incurs a cost.} of prediction.
In particular, we are interested here in the setting when the sparsity assumptions regarding exact zero regression coefficients are violated, and many regression coefficients might be non-zero but have negligible magnitude. 
This is also sometimes referred to as quasi-sparseness \citep{datta2016bayesian}.
In such situations, we may have to trade in a small reduction in prediction accuracy for an increase in sparseness. 

Spike-and-slab priors, as proposed by \citep{chipman2001practical}, can potentially handle such a trade-off between prediction accuracy and sparseness. Though, manual setting of these priors is difficult, since they are either too restrictive, or depend on the unknown noise variance of the response variable. % rovckova2014emvs
The limitations of these previous approaches are basically due to the desire for conjugate priors which results in closed-form solutions for the marginal likelihood. 

Here, in this work, we propose a hierarchical spike-and-slab prior for the linear regression model that allows the user to explicitly specify the minimal magnitude $\delta$ of the regression coefficients that is considered practically significant. 
The proposed model decouples the response noise prior variance from the regression coefficients' prior variance, and thus making the threshold parameter $\delta$ more meaningful than previous work \citep{chipman2001practical}.
For example, $\delta$ can be set such that the Mean-Squared Error (MSE) of the prediction is only little influenced by ignoring covariates with coefficients' magnitude smaller than $\delta$.
% In case where the specification of $\delta$ is difficult, we show that automatic selection of $\delta$ via empirical Bayes can be a viable choice. 

Our proposed method also resolves another subtle issue with the spike-and-slab priors from \citep{chipman2001practical}, namely
inconsistent Bayes factors (BF). Due to the fact that the spike-and-slab priors of \citep{chipman2001practical} (and related work like \citep{ishwaran2005spike}) have full support, 
the Bayes factors of any two models is bounded in probability, for which we give a formal proof in Section \ref{sec:proofAsymptoticallyCorrectBF_regression}.
This is an undesirable property for Bayesian hypothesis testing, since we would like that the BF between the true and the wrong model grows with increasing sample size.
In order to resolve this issue, our proposed method uses disjunct support priors, which allows us to guarantee consistent Bayes factors that grow exponentially fast favoring the true model.
% consistent Bayes factors in the sense that the ratio of the true model's marginal likelihood to any other models' marginal likelihood converges 
% to infinity for large sample sizes. 

Though our choice of the priors allows us to prove consistent Bayes factors, our choice complicates the calculation of the marginal likelihood.
% However, our choice of the prior does not enable the calculation of the marginal likelihood in closed-form anymore. 
As a solution, we propose to estimate all Bayes factors by introducing a latent variable indicator vector $\mathbf{z}$, with an efficient Gibbs sampler to sample from its posterior distribution.

We note that \citep{johnson2012bayesian} also proposed the use of disjunct support priors for variable selection in linear regression. However, the difference in support of their spike and slab priors is only at $\{0\}$, which makes them unsuitable for the quasi-sparse setting. % , which we explain in more detail in Section \ref{sec:relatedWork}.

%  This allows us to estimate all model probabilities $p(S  | \mathbf{y}, X, \delta)$, where $S$ is a set of variables.
 %  and $\delta$ is a threshold parameter specifying practical relevance (effect size). 
% Therefore, we propose to estimate the marginal likelihood through partly analytic integration combined with a 
% truncated normal approximation. A big advantage of our approximation is that it allows us to calculate the marginal likelihood even for regression problems with large number of variables $d$, which is infeasible for MCMC methods. 

The rest of this article is organized as follows.
In the next section, we summarize the properties of spike-and-slab priors from previous work. 
In Section \ref{sec:proposedMethod_regression}, we introduce our model for variable selection based on disjunct support spike-and-slab priors.
In Section \ref{sec:proofAsymptoticallyCorrectBF_regression}, we prove that the disjunct support priors of our proposed method allows us to guarantee consistent Bayes factors, and compare this to the asymptotic results for other spike-and-slab priors.
In Section \ref{sec:proposedMethod_estimation_regression}, we explain our MCMC sampling strategy for estimating model probabilities.
Since the elicitation of $\delta$ can be difficult, we discuss in Section \ref{sec:deltaSpecification_regression} a strategy for determining $\delta$ by estimating the increase in mean squared error (MSE) for prediction. 
We evaluate our proposed method on several simulated data sets in Section \ref{sec:evaluation_synthetic_regression}, and real data sets in Section \ref{sec:evaluation_real_regression}.
Finally, we summarize our findings in Section \ref{sec:conclusions_regression}.

\section{Related work} \label{sec:relatedWork}

\begin{figure*}[h]
  \centering
  \includegraphics[scale=0.5]{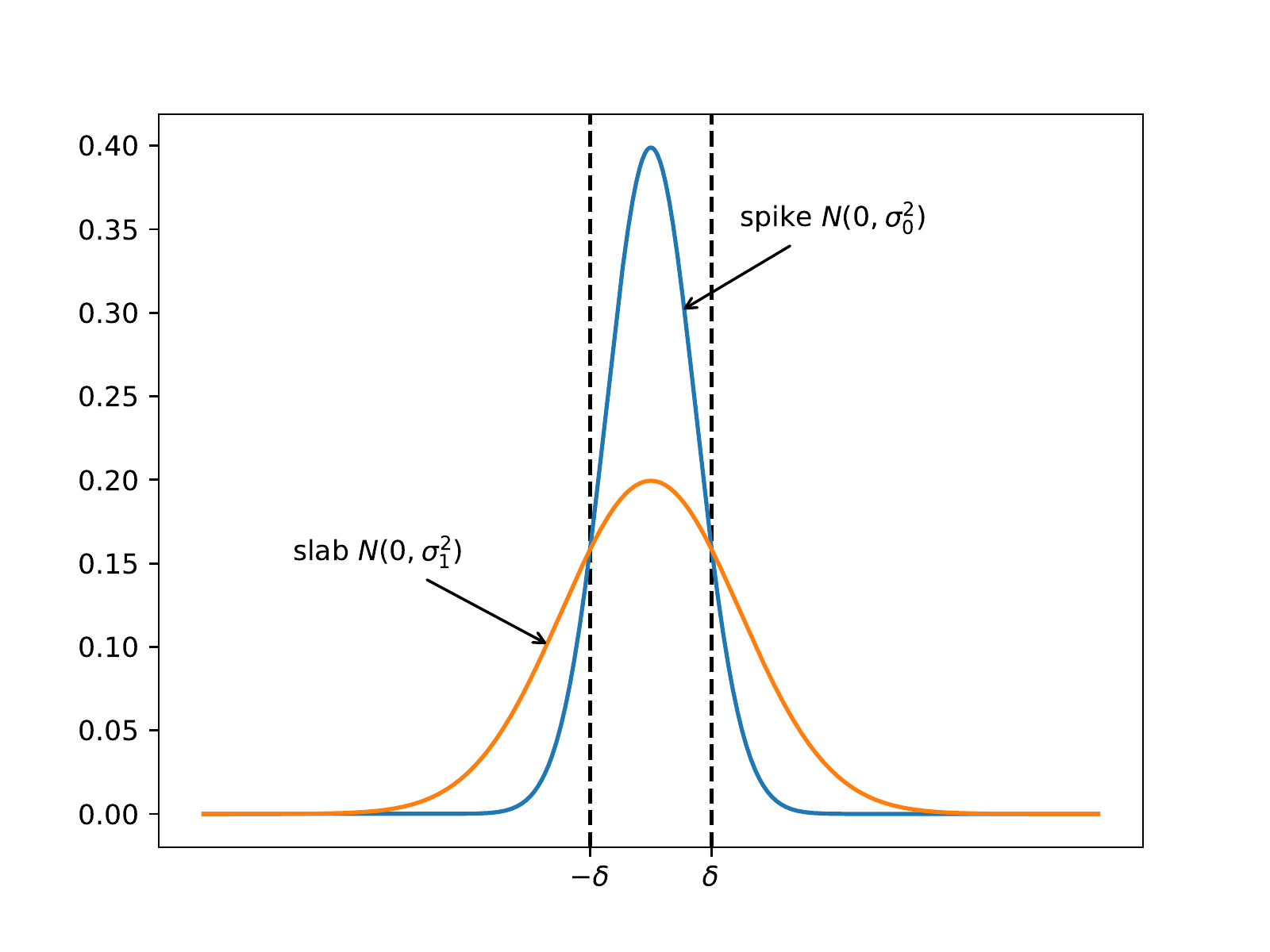}  % <left> <lower> <right> <upper> [scale=0.4, trim=10cm 0cm 10cm 0cm]
    \caption{Example of spike and slab prior as proposed in \citep{chipman2001practical}: both spike and slab priors are normal distributions but with different variances $\sigma_0^2$ and $\sigma_1^2$. 
    The variances $\sigma_0^2$ and $\sigma_1^2$ are specified such that the probability density of the spike prior dominates the one of the slab prior in the interval $[-\delta, \delta]$, otherwise the slab prior dominates.}
  \label{fig:spikeAndSlab_continuous}
\end{figure*} 

\begin{figure*}[h]
  \centering
  \includegraphics[scale=0.5]{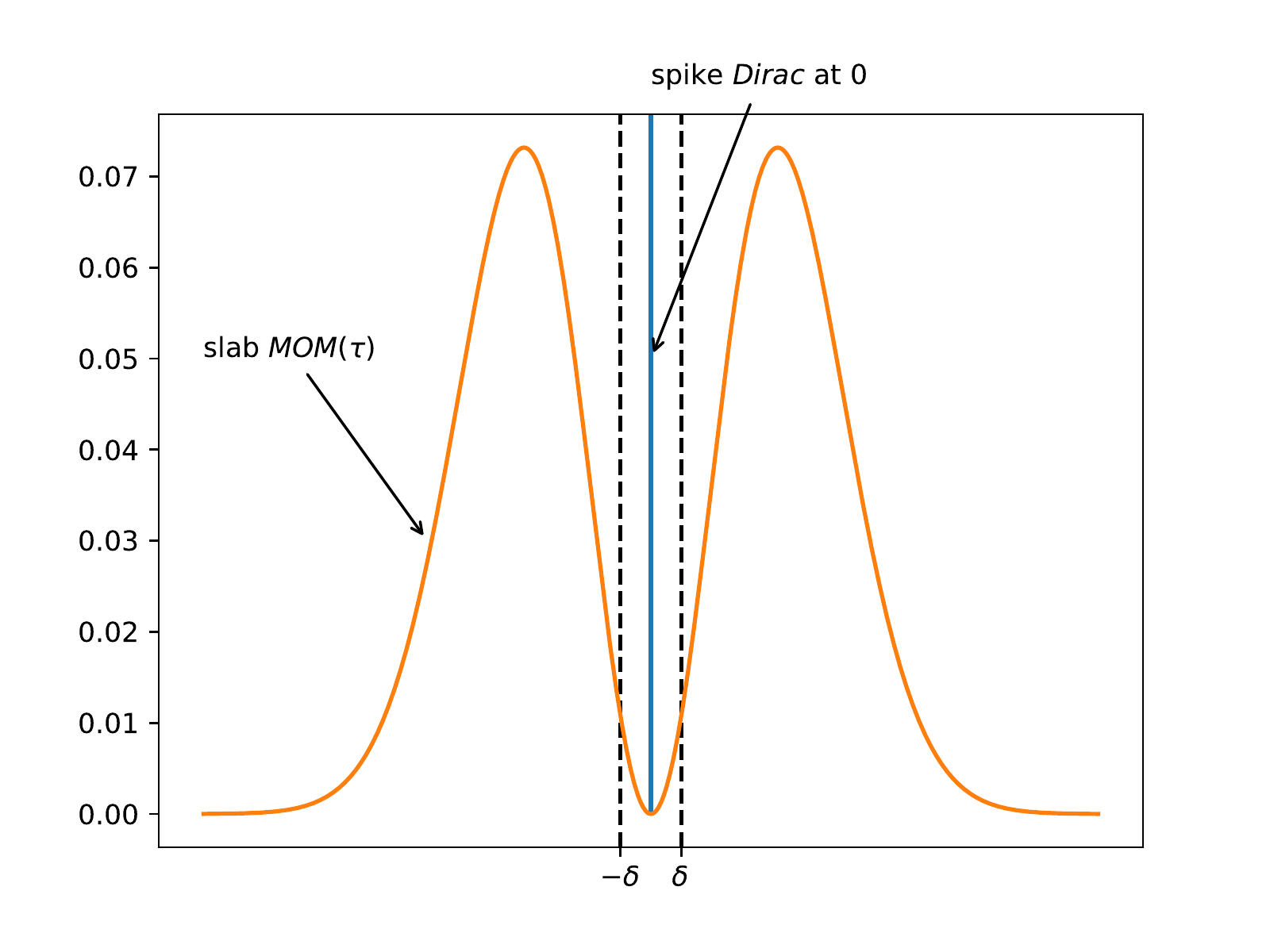}  % <left> <lower> <right> <upper> [scale=0.4, trim=10cm 0cm 10cm 0cm]
    \caption{Example of spike and slab prior as proposed in \citep{johnson2012bayesian}: spike distribution is a Dirac measure at point 0, and slab distribution is the moment matching prior with shape parameter $\tau$. 
    Parameter $\tau$ is specified such that the probability mass of $[-\delta, \delta]$ is 0.01 as suggested in \citep{johnson2010use}.}
  \label{fig:spikeAndSlab_MOM}
\end{figure*} 

A natural way to handle noise for variable selection are the spike-and-slab priors as proposed in \citep{chipman2001practical}.
The basic idea is to model the coefficients of the relevant and non-relevant variables by a normal distribution with variances $\sigma_1^2$ and $\sigma_0^2$, respectively, 
and $\sigma_1^2 \gg \sigma_0^2$. An example is shown in Figure \ref{fig:spikeAndSlab_continuous}. 

The variance parameters $\sigma_1^2$ and $\sigma_0^2$ must be set manually.
A difficulty of spike-and-slab priors is the correct setting of these parameters, and therefore \cite{ishwaran2005spike} proposed to place hyper-priors over these parameters in such a way 
that the resulting marginal prior $p(\boldsymbol{\beta})$ is little sensitive to the hyper-parameter choice. 
However, their prior choice does not allow for a closed-form marginal likelihood. Furthermore, their prior choice is only suitable for the situation where there is no noise, i.e. a variable $j$ is considered to be relevant if and only if the true coefficient $\beta_j$ is not zero. 

In contrast, the spike and slab priors proposed in \citep{chipman2001practical} allow to specify practical significance (what we call here ``relevance") by setting $\sigma_1^2$ to some large enough value (for example 100)  and then set $\sigma_0^2$ such that the intersection points of the two priors occur at a pre-specified value $\delta$ (and $-\delta$), see Figure \ref{fig:spikeAndSlab_continuous}. 
However, their method has some drawbacks:
\begin{itemize}
\item Their conjugate prior formulation is sensitive to the prior for the response variance, whereas their non-conjugate formulation is not sensitive to the response variance, but has no closed-form solution anymore. 
\item For any $\delta > 0$,  the Bayes factors are not consistent in the following sense.
Let $S$ be the true set of relevant variables and $S'$ any other set, then we have
\begin{align*}
\frac{p(\mathbf{y}_n | X_n, S)} {p(\mathbf{y}_n | X_n, S')} \stackrel{P}{\rightarrow} O_p(1)  \, ,
\end{align*}
where $\mathbf{y}_n := (y_1, \ldots y_n)$ and $X_n := (\mathbf{x}_1, \ldots, \mathbf{x}_n)$, are the observed responses and covariates of $n$ samples. 
This is due to the fact that the model dimension of spike-and-slab priors is the \emph{same} for model $S$ and $S'$. 
%  Furthermore, the priors in \citep{chipman2001practical} fullfil the regularity conditions for the Bayesian central limit.
As a consequence, the influence of the prior can be ignored, in the sense that the influence of the prior is asymptotically the \emph{same} for model $S$ and $S'$. For both models, the posterior distribution of $\boldsymbol{\beta}$ will concentrate around the true regression coefficient vector, and thus, the marginal likelihood cannot be distinguished any more. A formal proof will be given in Section \ref{sec:proofAsymptoticallyCorrectBF_regression}.
\item It might be difficult to specify $\delta$ a-priori.
\end{itemize}

Another method to handle noise on regression coefficients is to use nonlocal priors, as proposed in \citep{johnson2012bayesian,rossell2017nonlocal},  for the slab which places very small probability mass on the interval $[-\delta, \delta] \setminus \{0\}$ and its density is exact zero on $\{0\}$. 
For the spike distribution they suggest to use the Dirac measure at 0. The resulting spike and slab prior is illustrated in Figure \ref{fig:spikeAndSlab_MOM}. Therefore, their proposed spike and slab priors also have disjunct support, and as such enjoy exponentially fast growing Bayes factors \citep{johnson2010use}.
However, since the spike distribution has zero mass on $[-\delta, \delta] \setminus \{0\}$, it is unsuitable for the quasi-sparse setting. We analyze the asymptotic behavior of nonlocal priors in the quasi-sparse setting in Theorem \ref{prop:proofBF_nonlocalPriors_regression}, in Section \ref{sec:proofAsymptoticallyCorrectBF_regression}, and the finite sample behavior in our experiments, in Section \ref{sec:expAnalysisBF}. 
Interestingly, the nonlocal prior can be considered as a mixture of a truncated normal distribution with threshold $\delta$ and a uniform prior for $\delta$ \citep{rossell2017nonlocal}.
The prior on $\delta$ is uniform prior on an interval around $0$, where the length of the interval is controlled by the critical parameter $\tau$.
Therefore, the difficulty of specifying $\delta$ is shifted to the problem of specifying $\tau$. 
Recently, \cite{cao2018high} proposed to place a prior on $\tau$, instead of specifying a fixed value. However, their implementation relies on a Laplace approximation which does not enjoy any theoretic guarantees.

Finally, we note that recently \cite{miller2018robust} proposed a new framework, named $c$-posteriors, which can be applied to handle slight violations from the sparsity assumption.
However, their method introduces a hyper-parameter $c$ which might be difficult to interpret. Furthermore, their approach does not allow for the calculation of Bayes factors anymore. 

\section{Proposed method} \label{sec:proposedMethod_regression}

\begin{figure*}[h]
  \centering
  \includegraphics[scale=0.5]{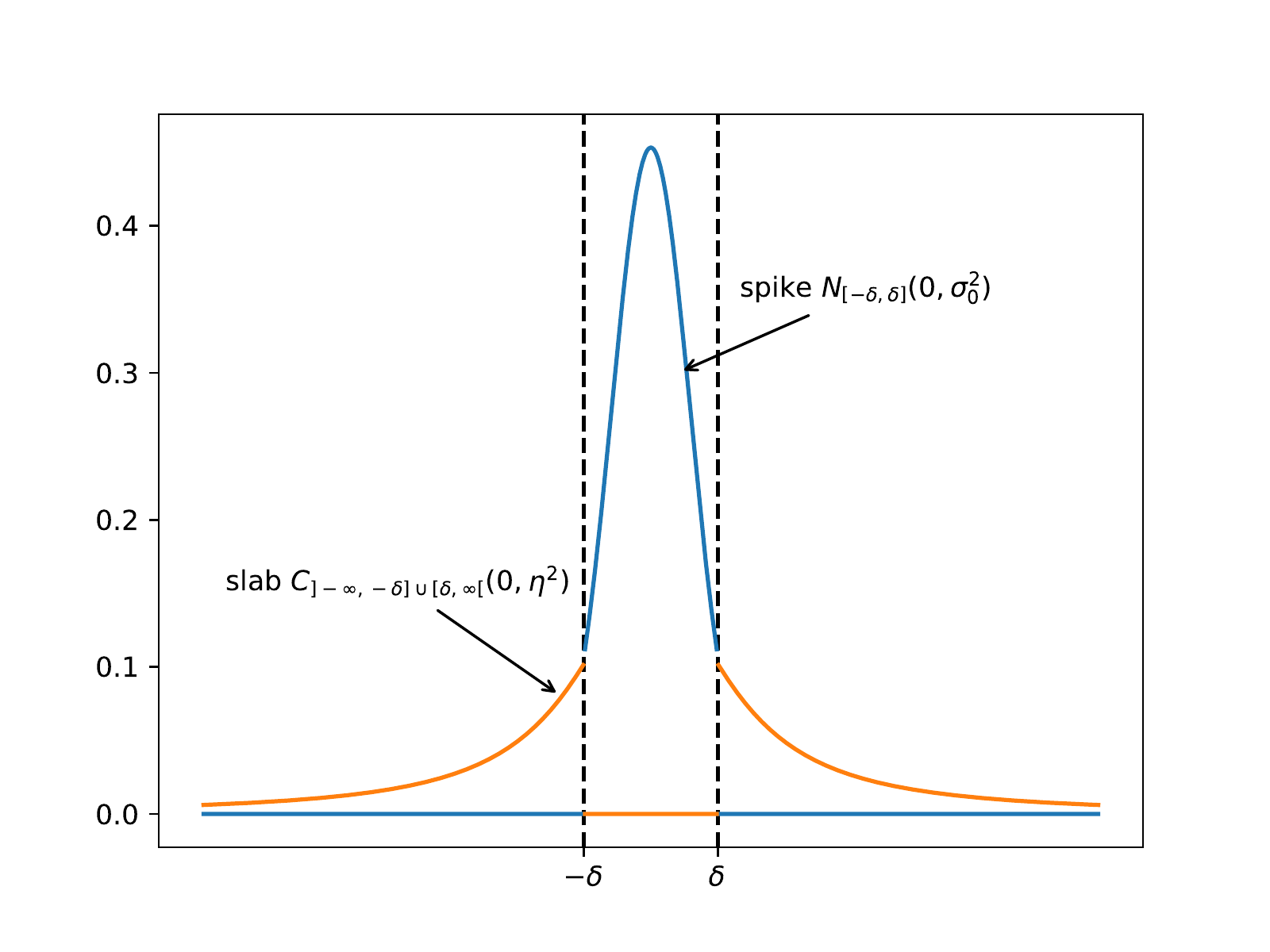}  % <left> <lower> <right> <upper> [scale=0.4, trim=10cm 0cm 10cm 0cm]
    \caption{Illustration of the proposed spike and slab prior. $C_{]-\infty, -\delta] \cup [\delta, \infty[}(0, \eta^2)$ denotes the Cauchy distribution with mean 0 and scale $\eta^2$. }
  \label{fig:spikeAndSlab_proposed}
\end{figure*} 

Let $S$ be the indices of the selected covariates (i.e. the covariate that are considered to be relevant), and $\mathcal{C}:= \{1,\ldots, d\} \setminus \mathcal{S}$ the set of irrelevant covariates.
Furthermore, let $s := |\mathcal{S}|$ be the number of selected covariates. 
We consider the following linear model for $y \in \mathbb{R}$ regressed on $\mathbf{x} \in \mathbb{R}^d$:
\begin{align*}
y = \mathbf{x}^T \boldsymbol{\beta} + \epsilon \, ,
\end{align*}
where 
\begin{equation*}
 \begin{aligned}
 & \left.
\begin{aligned}
& \epsilon \sim N(0, \sigma_r^2) \, , \\
& \sigma_r^2 \sim \text{Inv-}\chi^2(\nu_r, \eta_r^2) \, , 
 \end{aligned}
 \quad \qquad  \qquad \qquad  \qquad \right\} 
 \text{Prior for noise $\epsilon$} \\
 & \left.
\begin{aligned}
& s \sim \text{Multinomial}(p, \pi_{\text{rel}}) \, , \\
& \pi_{\text{rel}} \sim \text{Beta}(1,1) \, , \\
 \end{aligned}
 \quad \qquad \qquad \qquad \right\} 
 \text{Prior for number of relevant covariates $s$} \\
  & \left.
\begin{aligned}
& \sigma_1^2 \sim \text{Inv-}\chi^2(\nu_1, \eta_1^2) \, ,  \\
& \text{for $j \in \{1 , \ldots, d \}$:} \\
& \quad \text{if $j \in S$, then} \\
& \quad \quad \beta_j \sim N_{]-\infty, -\delta] \cup [\delta, \infty[} (0,  \sigma_1^2) \\
& \quad \text{else} \\
& \quad \quad \beta_j \sim N_{[-\delta, \delta]} (0,  \sigma_0^2) \, . \\
\end{aligned}
 \qquad \quad  \quad \; \right\} 
 \text{Prior for regression coefficients $\boldsymbol{\beta}$}  \\
 \end{aligned}
 \end{equation*}
 $\nu_r, \eta_r^2$ are set such that $\text{Inv-}\chi^2(\nu_r, \eta_r^2)$ is a weakly informative prior. $\text{Inv-}\chi^2$ denotes the scaled inverse chi-square distribution (see details below), where
 $\nu_r$ can be interpreted as the number of a-priori observations. 
For our experiments, we set $\nu_r$ and the prior variance $\sigma_r^2$ to 1.

$N_{[-\delta, \delta]}$  and $N_{]-\infty, -\delta] \cup [\delta, \infty[}$ denote the truncated normal distribution with support $[-\delta, \delta]$ and $]-\infty, -\delta] \cup [\delta, \infty[$ for the spike and slab prior, respectively.
The specification of $\sigma_0^2$, and $\sigma_1^2$ determines the shape of the spike and slab prior, respectively. 
For the slab prior, in order to allow for possibly large values of $\beta_j$, we place a diffuse hyper-prior on $\sigma_1^2$.
In particular, we set $\nu_1 = 1$, and $\eta_1^2 = 100$ which corresponds to a truncated Cauchy distribution with mean zero and scale $\eta_1^2$ for $p(\beta_j | j \in S, \nu_1, \eta_1^2, \delta)$.

At the boundary $\beta_j = \delta$ (and, due to symmetry $\beta_j = -\delta$) we want to be indifferent about whether $\beta_j$ was sampled from the spike or slab prior. Therefore, we set $\sigma_0^2$ such that
\begin{equation} \label{eq:priorEqualityConstraint}
p(\beta_j = \delta | j \in S, \nu_1, \eta_1^2, \delta) = p(\beta_j = \delta | j \notin S, \sigma_0^2, \delta) \,.
\end{equation}
The left hand side of Equation \eqref {eq:priorEqualityConstraint} does not have a closed-form solution. However, note that 
\begin{equation*}
p(\beta_j = \delta | j \in S, \nu_1, \eta_1^2, \delta) = \int N_{]-\infty, -\delta] \cup [\delta, \infty[} (\beta_j = \delta | 0,  \sigma_1^2) \cdot \text{Inv-}\chi^2(\sigma_1^2  | \nu_1, \eta_1^2) d \sigma_1^2 \, ,
\end{equation*}
which we solve using numerical integration.  Our proposed spike and slab prior is illustrated in Figure \ref{fig:spikeAndSlab_proposed}.

Therefore, the remaining critical hyper-parameter is only the specification of the threshold parameter $\delta$.
In Section \ref{sec:deltaSpecification_regression}, we discuss the specification of $\delta$.

Note that the prior on the number of relevant variables $s$ ensures multiplicity control and has been extensively studied in \citep{scott2010bayes,scott2006exploration}.
The probability of a variable being relevant $\pi_{\text{rel}}$ can be integrated out leading to
\begin{align*} % \label{eq:variableSizePrior}
p(s) = \frac{1}{d + 1} \left(
\begin{array}{c}
d \\
s
\end{array} \right)^{-1} \, .
\end{align*}

Note that the scaled inverse chi-square distribution is defined as follows (see e.g. \cite{gelman2013bayesian}): 
\begin{align*}
\text{Inv-}\chi^2(\sigma^2 | \nu, \eta^2) &= (\eta^2)^{\nu / 2} \frac{(\nu / 2)^{\nu / 2}}{\Gamma(\nu / 2)} (\sigma^2)^{-(\frac{\nu}{2} + 1)} e^{-\frac{1}{2 \sigma^2} \nu \eta^2} \, .
% p(\sigma_r^2 | \nu_r, \eta_r^2) &= (\eta_r^2)^{\nu_r / 2} \frac{(\nu_r / 2)^{\nu_r / 2}}{\Gamma(\nu_r / 2)} (\sigma_r^2)^{-(\frac{\nu_r}{2} + 1)} e^{-\frac{1}{2 \sigma_r^2} \nu_r \eta_r^2} \, .
\end{align*}
Therefore, the joint probability density function is given by:
\begin{align*}
p(\boldsymbol{\beta}, \sigma_r^2, \sigma_1^2, \mathbf{y}, S, |  X)  
&= p(s) \cdot (2\pi)^{-\frac{n}{2}} \cdot   (\sigma_r^2)^{-\frac{n}{2}} e^{-\frac{1}{2\sigma_r^2} ||\mathbf{y} - X \boldsymbol{\beta}||_2^2}  \\
&\quad \cdot (\eta_r^2)^{\nu_r / 2} \frac{(\nu_r / 2)^{\nu_r / 2}}{\Gamma(\nu_r / 2)} \cdot (\sigma_r^2)^{-(\frac{\nu_r}{2} + 1)} e^{-\frac{1}{2 \sigma_r^2} \nu_r \eta_r^2}  \\
&\quad \cdot (\eta_1^2)^{\nu_1 / 2} \frac{(\nu_1 / 2)^{\nu_1 / 2}}{\Gamma(\nu_1 / 2)} \cdot (\sigma_1^2)^{-(\frac{\nu_1}{2} + 1)} e^{-\frac{1}{2 \sigma_1^2} \nu_1 \eta_1^2}  \\
& \quad  \cdot  \Big( \prod_{j \in \mathcal{C}}  \mathbbm{1}_{\mathcal{N}}(\beta_j) \cdot \frac{1}{\iota(\mathcal{N}, \sigma_0^2)}  \cdot e^{-\frac{1}{2 \sigma_0^2} \beta_j^2} \Big)  \\
&\quad \cdot \Big( \prod_{j \in \mathcal{S}}  \mathbbm{1}_{\mathcal{R}}(\beta_j) \cdot \frac{1}{\iota(\mathcal{R}, \sigma_1^2)}  e^{-\frac{1}{2 \sigma_1^2} \beta_j^2}    \Big)  \\
&= C_0 \cdot p(s) \cdot   (\sigma_r^2)^{-\frac{n}{2}} e^{-\frac{1}{2\sigma_r^2} ||\mathbf{y} - X \boldsymbol{\beta}||_2^2}  \\
&\quad \cdot (\sigma_r^2)^{-(\frac{\nu_r}{2} + 1)} e^{-\frac{1}{2 \sigma_r^2} \nu_r \eta_r^2}  \\
&\quad \cdot (\sigma_1^2)^{-(\frac{\nu_1}{2} + 1)} e^{-\frac{1}{2 \sigma_1^2} \nu_1 \eta_1^2}  \\
& \quad  \cdot  \Big( \prod_{j \in \mathcal{C}}  \mathbbm{1}_{\mathcal{N}}(\beta_j) \cdot \frac{1}{\iota(\mathcal{N}, \sigma_0^2)}  \cdot e^{-\frac{1}{2 \sigma_0^2} \beta_j^2} \Big)  \\
&\quad \cdot \Big( \prod_{j \in \mathcal{S}}  \mathbbm{1}_{\mathcal{R}}(\beta_j) \cdot \frac{1}{\iota(\mathcal{R}, \sigma_1^2)}  e^{-\frac{1}{2 \sigma_1^2} \beta_j^2}    \Big)  \, ,
\end{align*}
where we defined $\mathcal{N} := [-\delta, \delta]$, and $\mathcal{R} := ]-\infty, -\delta] \cup [\delta, \infty[$, and 
\begin{align*}
\iota(\mathcal{A}, \sigma^2) := \int  \mathbbm{1}_{\mathcal{A}}(x) e^{-\frac{1}{2 \sigma^2} x^2} dx
\end{align*}
and 
\begin{align*}
C_0 &:= (2\pi)^{-\frac{n}{2}}
 \cdot (\eta_r^2)^{\nu_r / 2} \frac{(\nu_r / 2)^{\nu_r / 2}}{\Gamma(\nu_r / 2)} 
  \cdot (\eta_1^2)^{\nu_1 / 2} \frac{(\nu_1 / 2)^{\nu_1 / 2}}{\Gamma(\nu_1 / 2)} \, .
\end{align*}

\section{Asymptotic Bayes factors} \label{sec:proofAsymptoticallyCorrectBF_regression}

In this section, we formally prove the asymptotic behavior of the Bayes factors between the true model and any other model,
first for our proposed method, in Theorem \ref{prop:proofBFproposed_regression}, and then for previously proposed spike and slab priors, in Theorem \ref{prop:proofBFprevious_regression} and Theorem \ref{prop:proofBF_nonlocalPriors_regression}.

In the following, we define the true set of relevant variables $S$ as 
\begin{align} \label{eq:definitionS}
S := \Big\{j \in \{1,\ldots, d\} \Big|  \; | \; \beta_{j, t} | > \delta \Big\} \, .
\end{align}
Furthermore, we denote convergence in probability by $\stackrel{P}{\rightarrow}$. 

\begin{theorem}  \label{prop:proofBFproposed_regression}
Let $S$ be the true set of relevant variables and $S'$ any other set of variables.
For the proposed method with disjunct support priors (as defined in Section \ref{sec:proposedMethod_regression}), it holds that
\begin{align*}
\frac{1}{n} \log \frac{p(\mathbf{y}_n | X_n, S)} {p(\mathbf{y}_n | X_n, S')} \stackrel{P}{\rightarrow} c  \, ,
\end{align*}
for some $c > 0$, and where $X_n := (\mathbf{x}_1, \ldots, \mathbf{x}_n)$, are $n$ samples drawn from a non-degenerated probability distribution $p(\mathbf{x})$ with finite covariance matrix,  
and $\mathbf{y}_n := (y_1, \ldots, y_n)$, where $y_i \sim p(y | \mathbf{x}_i, \sigma^2_{r,t}, \boldsymbol{\beta}_t)$, for some true parameters $\sigma^2_{r,t}$ and $\boldsymbol{\beta}_t$.
We assume that $\boldsymbol{\beta}_t$ is not on the boundary of the support of the prior $p(\boldsymbol{\beta} | S)$.
\end{theorem}

Theorem \ref{prop:proofBFproposed_regression} means that the Bayes factors favoring the true model grows exponentially fast in the number of samples $n$. 

\begin{proof}
Results from hypothesis testing with disjunct support for the null and alternative hypothesis as in \citep{johnson2010use,walker1969asymptotic} can be applied here, though several assumptions must be checked.
Instead, we give here a direct proof. 

First, in order to approximate the marginal likelihood $p(\mathbf{y}_n | X_n, S)$, we use the Laplace approximation from Theorem 1 in \citep{kass1990validity}.
The likelihood function of the normal linear model is Laplace regular (see proof in \cite{kass1990validity}), which means that the conditions on the likelihood function in Theorem 1 \citep{kass1990validity} hold. 
Let us denote by $\Theta_S$ the Cartesian product of the support of the priors $p(\boldsymbol{\beta} | S)$ and $p(\sigma_r^2)$ (for technical reasons we may exclude the points at $\delta$ and $-\delta$ to make $\Theta_S$ an open subset of $\mathbb{R}^{d+1}$). Since the densities of the Cauchy distribution, the normal distribution, and the scaled inverse chi-square distribution, are four times continuously differentiable, we have that the priors $p(\boldsymbol{\beta} | S)$ and $p(\sigma_r^2)$ are four times continuously differentiable on its support. 

Let $\hat{\boldsymbol{\theta}}_n$ be the maximum likelihood estimate (MLE) for $\log p(\mathbf{y}_n | X_n, \boldsymbol{\theta})$. 
Note that by the consistency of the MLE, we have that $\hat{\boldsymbol{\theta}}_n \stackrel{p}{\rightarrow}  \boldsymbol{\theta}_t$ (see for example Theorem 4.17. in \cite{MathematicalStatisticsJunShao2003}),
therefore for any open ball around $\boldsymbol{\theta}_t$, denoted by $\mathcal{B}(\boldsymbol{\theta}_t)$, we have 
$P(\hat{\boldsymbol{\theta}}_n \in \mathcal{B}(\boldsymbol{\theta}_t)) \rightarrow 1$, and therefore $P(\hat{\boldsymbol{\theta}}_n \in \Theta_S) \rightarrow 1$.

Therefore, all conditions of Theorem 1 in \citep{kass1990validity} are met.
Let us define $p(\boldsymbol{\theta} | S) := p(\boldsymbol{\beta} | S) \cdot p(\sigma_r^2)$, and $h(\boldsymbol{\theta}) := - \frac{1}{n} \log p(\mathbf{y}_n | X_n, \boldsymbol{\theta})$.
Next, applying Theorem 7 and 1 from \citep{kass1990validity}, we have almost surely that \footnote{We use here the notation $\text{det}(X)$ for the determinant of a matrix $X$ in order to avoid confusion with the absolute value function.}
\begin{align*} 
p(\mathbf{y}_n | X_n, S) 
& = \int_{\Theta_S} p(\mathbf{y}_n | X_n, \boldsymbol{\theta}) p(\boldsymbol{\theta} | S) d \boldsymbol{\theta}  \\
& = (2 \pi)^{d+1} \cdot \Big[ \text{det}(n \cdot \frac{\partial^2}{\partial^2 \boldsymbol{\theta}} h(\hat{\boldsymbol{\theta}}_n) ) \Big]^{-\frac{1}{2}} \cdot p(\mathbf{y}_n | X_n,\hat{\boldsymbol{\theta}}_n) 
\cdot ( p(\hat{\boldsymbol{\theta}}_n | S) + O(n^{-1})) \, .
\end{align*}
Furthermore, we have that  
\begin{align*} 
\frac{\partial^2}{\partial^2 \boldsymbol{\theta}} h(\hat{\boldsymbol{\theta}}_n) 
&= - \frac{1}{n} \sum_{i=1}^n  \frac{\partial^2}{\partial^2 \boldsymbol{\theta}} \log p(y_i | \mathbf{x}_i, \hat{\boldsymbol{\theta}}_n)  \\
&\stackrel{a.s.}{\rightarrow}   - \E_{\mathbf{x}, y} \Big[ \frac{\partial^2}{\partial^2  \boldsymbol{\theta}} \log p(y | \mathbf{x}, \hat{\boldsymbol{\theta}}_n) \Big] \, .
\end{align*}
Since $\boldsymbol{\theta} \mapsto \E_{\mathbf{x}, y} \Big[ \frac{\partial^2}{\partial^2  \boldsymbol{\theta}} \log p(y | \mathbf{x}, \boldsymbol{\theta}) \Big]$ is a continuous function, 
and $\hat{\boldsymbol{\theta}}_n \stackrel{p}{\rightarrow}  \boldsymbol{\theta}_t$, we have by the continuous mapping theorem that 
\begin{align*} 
\E_{\mathbf{x}, y} \Big[ \frac{\partial^2}{\partial^2  \boldsymbol{\theta}} \log p(y | \mathbf{x}, \hat{\boldsymbol{\theta}}_n) \Big] \stackrel{p}{\rightarrow}   \E_{\mathbf{x}, y} \Big[ \frac{\partial^2}{\partial^2  \boldsymbol{\theta}} \log p(y | \mathbf{x}, \boldsymbol{\theta}_t) \Big] \, ,
\end{align*}
and since the matrix $- \E_{\mathbf{x}, y} \Big[ \frac{\partial^2}{\partial^2  \boldsymbol{\theta}} \log p(y | \mathbf{x}, \boldsymbol{\theta}_t ) \Big]$ is positive definite with every entry in $O(1)$, 
we have that 
% $\text{det}(\E_{\mathbf{x}, y} \Big[ \frac{\partial^2}{\partial^2  \boldsymbol{\theta}} \log p(y | \mathbf{x}, \boldsymbol{\theta}_t \Big]) > 0$, we have that 
$\log \text{det}(\frac{\partial^2}{\partial^2 \boldsymbol{\theta}} h(\hat{\boldsymbol{\theta}}_n) )  \in O_p(1)$. 
In summary, we have
\begin{align}
\log p(\mathbf{y}_n | X_n, S) 
& = (d + 1) \log (2 \pi) - \frac{d+1}{2} \log n \nonumber \\
& \quad - \frac{1}{2} \log \text{det}(\frac{\partial^2}{\partial^2 \boldsymbol{\theta}} h(\hat{\boldsymbol{\theta}}_n) ) + \log p(\mathbf{y}_n | X_n,\hat{\boldsymbol{\theta}}_n) 
+ \log ( p(\hat{\boldsymbol{\theta}}_n | S) + O(n^{-1})) \nonumber \\
& =  - \frac{d+1}{2} \log n + \log p(\mathbf{y}_n | X_n,\hat{\boldsymbol{\theta}}_n) + O_p(1) \, . \label{eq:nominatorApprox}
\end{align}

\paragraph{Upper bound for $p(\mathbf{y}_n | X_n, S')$}
First, note that the true parameter $\boldsymbol{\beta}_t$ is not not contained in the support of the prior $p(\boldsymbol{\beta} | S')$, since $S' \neq S$. Therefore, the regularity conditions for the Laplace approximation are not fulfilled. However, we can easily derive an upper bound as follows.
Let us define 
\begin{align*}
\hat{\boldsymbol{\theta}}_{S', n} := \argmax_{\boldsymbol{\theta} : p(\boldsymbol{\theta} | S') > 0}  p(\mathbf{y}_n | X_n, \boldsymbol{\theta}) \, 
\end{align*}
then we have
\begin{align} \label{eq:denominatorBound}
p(\mathbf{y}_n | X_n, S') = \int p(\mathbf{y}_n | X_n, \boldsymbol{\theta}) p(\boldsymbol{\theta} | S') d \boldsymbol{\theta} \leq  p(\mathbf{y}_n | X_n, \hat{\boldsymbol{\theta}}_{S', n})  \, .
\end{align}

\paragraph{Lower bound on $\log \frac{p(\mathbf{y}_n | X_n, S)}{p(\mathbf{y}_n | X_n, S')}$}

Putting together the results from Equations \eqref{eq:nominatorApprox} and \eqref{eq:denominatorBound}, we get
\begin{align*}
\log \frac{p(\mathbf{y}_n | X_n, S)}{p(\mathbf{y}_n | X_n, S')} 
&\geq \log p(\mathbf{y}_n | X_n, \hat{\boldsymbol{\theta}}_{n}) - \frac{d+1}{2} \log n + O_p(1)  - \log p(\mathbf{y}_n | X_n, \hat{\boldsymbol{\theta}}_{S', n}) \\
&= n \Big( \frac{1}{n} \sum_{i=1}^n \log p(y_i | \mathbf{x}_i, \hat{\boldsymbol{\theta}}_{n}) - \frac{1}{n} \sum_{i=1}^n  \log p(y_i | \mathbf{x}_i, \hat{\boldsymbol{\theta}}_{S', n}) \Big) - \frac{d + 1}{2} \log n + O_p(1) \\
&\stackrel{P}{\rightarrow} n \Big( \E_{\mathbf{x}} \Big[ g(\boldsymbol{\theta}_t) \Big]   -  \E_{\mathbf{x}} \Big[ g(\boldsymbol{\theta}_{S'}) \Big]   \Big) - \frac{d + 1}{2} \log n + O_p(1)  \, ,
\end{align*}
where 
\begin{align*}
g(\boldsymbol{\theta}) := \E_{y \sim p(y | \boldsymbol{\theta}, \mathbf{x})} \Big[  \log p(y | \boldsymbol{\theta}, \mathbf{x}) \Big] \, ,
\end{align*}
and $\boldsymbol{\theta}_{S'} := \argmax_{\boldsymbol{\theta} : p(\boldsymbol{\theta} | S') > 0}   \E_{\mathbf{x}} \Big[ g(\boldsymbol{\theta}) \Big]$.
Since $\boldsymbol{\theta}_t$ is the unique global maximum of $\E_{\mathbf{x}} \Big[ g(\boldsymbol{\theta}) \Big]$ (see Lemma 1 in Appendix A) and $p(\boldsymbol{\theta}_t | S') = 0$, we have that 
\begin{align*}
c_{\Delta} := \E_{\mathbf{x}} \Big[ g(\boldsymbol{\theta}_t) \Big]   -  \E_{\mathbf{x}} \Big[ g(\boldsymbol{\theta}_{S'}) \Big]  > 0
 \end{align*}
 and therefore
 \begin{align*}
\log \frac{p(\mathbf{y}_n | X_n, S)}{p(\mathbf{y}_n | X_n, S')} \geq n \cdot c_{\Delta} - \frac{d + 1}{2} \log n + O_p(1)  \stackrel{P}{\rightarrow} \infty  \,. 
\end{align*}
From the above line, we also see that the convergence of the Bayes factor $\frac{p(\mathbf{y}_n | X_n, S)}{p(\mathbf{y}_n | X_n, S')}$ is exponential in $n$. 
\end{proof}

Next, let us investigate the Bayes factors for full support spike and slab priors, as for example in \citep{chipman2001practical,ishwaran2005spike}.

\begin{theorem} \label{prop:proofBFprevious_regression}
Under the same assumptions as in Theorem \ref{prop:proofBFproposed_regression}, but assuming full support spike and slab priors for the evaluation of the marginal likelihoods $p(\mathbf{y}_n | X_n, S)$ and 
$p(\mathbf{y}_n | X_n, S')$, we have the following result:
\begin{align*}
\frac{p(\mathbf{y}_n | X_n, S)} {p(\mathbf{y}_n | X_n, S')} \stackrel{P}{\rightarrow} O_p(1)  \, .
\end{align*}
\end{theorem}

\begin{proof}
 Since the priors have full support, the posterior distribution also has full support. 
 Both posterior distributions contain the true regression coefficient vector $\boldsymbol{\beta}_t$, i.e. 
\begin{align*}
 p(\boldsymbol{\beta}_t | \mathbf{y}_n, X_n, S') > 0, \forall S' 
 \end{align*}
 Furthermore, since the likelihood function is the same as before in Theorem \ref{prop:proofBFproposed_regression}, we have, that 
%  as was proven before, that $\E_{\mathbf{x}} \Big[ g(\boldsymbol{\theta}) \Big]$ has the unique maximum $(\theta_t, \boldsymbol{\beta}_t)$ and is locally concave around this maximum.
 the regularity conditions for the Laplace approximation are fulfilled for \emph{all} models $S'$, and we have:
 \begin{align*}
 \log p(\mathbf{y}_n | X_n, S') 
 &\stackrel{P}{\rightarrow} \log p(\mathbf{y}_n | X_n, \hat{\boldsymbol{\theta}}_{S', n}) - \frac{d  + 1}{2} \log n + O_p(1)  \\
 &\stackrel{P}{\rightarrow} \log p(\mathbf{y}_n | X_n, \boldsymbol{\theta}_t) - \frac{d + 1}{2} \log n + O_p(1)  \, .
  \end{align*}
And therefore 
  \begin{align*}
 \log \frac{p(\mathbf{y}_n | X_n, S)}{p(\mathbf{y}_n | X_n, S')} \stackrel{P}{\rightarrow} O_p(1)   \, .
  \end{align*}
\end{proof}

%An immediate consequence of Theorem \ref{prop:proofBFprevious_regression} is that for full support spike and slab priors, as for example in \citep{chipman2001practical,ishwaran2005spike},
%the Bayes factors are inconsistent, in the sense that the Bayes factor, favoring the model with only the relevant variables, is bounded in probability. 
%Furthermore, in the setting where all irrelevant regression coefficients are negligible small but non zero (i.e. $\beta_{j,t} \neq 0$, for all $j \in \{1,\ldots, d\}$, and $|S| < d$), the 
%Bayes factors with nonlocal priors as in \citep{johnson2010use,johnson2012bayesian,rossell2017nonlocal} are also inconsistent is this sense.\footnote{This is because in \cite{johnson2010use,johnson2012bayesian,rossell2017nonlocal}, the nonlocal prior density of $p(\beta)$ is only exact zero for $\beta = 0$.}

The next theorem emphasizes that the disjunct support priors as in \citep{johnson2010use,johnson2012bayesian,rossell2017nonlocal}
are unsuitable for the quasi-sparse setting. For simplicity, we focus here on the product moment matching priors \citep{johnson2012bayesian}, but similar results hold for other disjunct support priors with difference in support only at $\{0\}$.
% proposed in \citep{johnson2010use,johnson2012bayesian}.

\begin{theorem} \label{prop:proofBF_nonlocalPriors_regression}
Let us assume that the true regression coefficient vector is quasi-sparse, i.e.  $\beta_{j,t} \neq 0$, for all $j \in \{1,\ldots, d\}$, and $|S| < d$, where $S$ is set of true relevant variables, as defined in Equation \eqref{eq:definitionS},
% As before, let us denote by $\boldsymbol{\beta}_t$ the true regression coefficient vector, and 
% the set of relevant regression coefficients as % Let $\mathcal{C}:= \{1,\ldots, d\} \setminus \mathcal{S}$, the set of irrlevant coefficient regrission 
% \begin{align*}
% S := \Big\{j \in \{1,\ldots, d\} \Big|  \; | \; \beta_{j, t} | > \delta \Big\} \, ,
% C := \Big\{j \in \{1,\ldots, d\} \Big|  \; | \beta_{j, t} | > \delta \Big\} \, .
% \end{align*}
% Let $S'$ be any other set of variables.
and let $A$ be the model with all variables, i.e. $A := \{1, \ldots, d\}$. %  be any other set of variables.
Assume the product moment matching priors as proposed in \citep{johnson2012bayesian}, we have
\begin{align*}
\frac{p(\mathbf{y}_n | X_n, S)} {p(\mathbf{y}_n | X_n, A)} \stackrel{P}{\rightarrow} 0  \, .
\end{align*}
\end{theorem}

\begin{proof}
The proof is an immediate consequence of the consistency property of the product moment matching priors as proven in Theorem 1 in \citep{johnson2012bayesian},
since $A$ is the true model in the sense that $A = \{j  \, | \, \beta_{t,j} \neq 0 \, , j \in \{1,\ldots, d\}  \}$.
\end{proof}

\section{Estimation of model probabilities} \label{sec:proposedMethod_estimation_regression}

Calculating the marginal likelihood for each model explicitly is computationally challenging, due to the disjunct support priors on $\beta$: 
\begin{itemize}
\item A Laplace approximation is not valid anymore, since the true parameter might not be contained in the support of the prior distribution.
\item Chib's method \citep{chib1995marginal,chib2001marginal} is computationally very expensive since, though we can sample, the normalization constants of each conditional probability is not available.
\end{itemize}

Instead, we estimate $p(S | \mathbf{y}, X)$, by introducing a model indicator vector $\mathbf{z} \in \{0,1\}^d$, where $z_j$ indicates whether variable $j$ should be included in $S$ or not.
We sample $M$ samples from the posterior distribution of $\mathbf{z}$ using Algorithm \ref{alg:propGibsForZ}.

\begin{algorithm}[h]
\caption{Gibbs sampler for sampling from $p(\boldsymbol{z} | \sigma_0, \mathbf{y}, X)$. \label{alg:propGibsForZ}}
\begin{algorithmic}
\FOR {$t$ from 1 to $M$}
\FOR  {$j$ from $1$ to $d$}
	\STATE $p(z_j)$  := sample from $p(z_j | \boldsymbol{\beta}_{-j}, \mathbf{z}_{-j}, \sigma_r,  \sigma_1, \sigma_0, \mathbf{y},  X)$
	\STATE $p(\beta_j)$ := sample from $p(\beta_j | \boldsymbol{\beta}_{-j}, \mathbf{z}, \sigma_r , \sigma_1, \sigma_0, \mathbf{y},  X)$
\ENDFOR
\STATE $\sigma_r^2$ := sample from $p(\sigma_r^2 | \boldsymbol{\beta},  \mathbf{z}, \mathbf{y}, X)$
\STATE $\sigma_1^2$ := sample from $p(\sigma_1^2 | \boldsymbol{\beta}, \mathbf{z}, \mathbf{y}, X)$
\ENDFOR
\end{algorithmic}
\end{algorithm}
Sampling from each of the conditional distributions in Algorithm \ref{alg:propGibsForZ} is explained in the following.
We note that all of the conditional distributions, except $p(\sigma_1^2 | \boldsymbol{\beta}, \sigma_r^2, \mathbf{z}, \mathbf{y},  X)$, have an analytic solution that can be expressed by 
standard distributions. Therefore, we find that even for high-dimensional spaces, using Algorithm \ref{alg:propGibsForZ} is computationally feasible.  

\subsection{Analytic solution for $p(z_j | \boldsymbol{\beta}_{-j}, \mathbf{z}_{-j}, \sigma_r,  \sigma_1, \sigma_0, \mathbf{y},  X)$}

Let $\mathbf{x}_j$ denote the $j$-th column of $X$, and $X_{-j}$ the matrix X where column $j$ is removed. Then we have 
\begin{align*}
||\mathbf{y} - X \boldsymbol{\beta}||_2^2 = ||\mathbf{y} - (\mathbf{x}_j \beta_j + X_{-j} \boldsymbol{\beta}_{-j})||_2^2  = ||\tilde{\mathbf{y}} - \mathbf{x}_j \beta_j ||_2^2 \, ,
\end{align*}
with $\tilde{\mathbf{y}} := \mathbf{y} - X_{-j} \boldsymbol{\beta}_{-j}$.

\begin{align*}
p(z_j | \boldsymbol{\beta}_{-j}, \mathbf{z}_{-j}, \sigma_r,  \sigma_1, \sigma_0, \mathbf{y},  X)
&\propto \int p(\boldsymbol{\beta}, \mathbf{z}, \sigma_r, \sigma_1, \mathbf{y} |  X, \sigma_0)  d \beta_j \\
&= \int p(\mathbf{z}) \cdot C_0 \cdot   (\sigma_r^2)^{-\frac{n}{2}} e^{-\frac{1}{2\sigma_r^2} ||\mathbf{y} - X \boldsymbol{\beta}||_2^2}  \\
&\quad \cdot (\sigma_r^2)^{-(\frac{\nu_r}{2} + 1)} e^{-\frac{1}{2 \sigma_r^2} \nu_r \eta_r^2}  \\
& \quad  \cdot  \Big( \prod_{j \in \mathcal{C}}  \mathbbm{1}_{\mathcal{N}}(\beta_j) \cdot \frac{1}{\iota(\mathcal{N}, \sigma_0^2)}  \cdot e^{-\frac{1}{2 \sigma_0^2} \beta_j^2} \Big)  \\
&\quad \cdot \Big( \prod_{j \in \mathcal{S}}  \mathbbm{1}_{\mathcal{R}}(\beta_j) \cdot \frac{1}{\iota(\mathcal{R}, \sigma_1^2)}  e^{-\frac{1}{2 \sigma_1^2} \beta_j^2}    \Big) d \beta_j  \\
&\propto \frac{p(\mathbf{z})}{\iota(\mathcal{A}_{z_j}, \sigma_{z_j}^2)}  \cdot \int e^{-\frac{1}{2\sigma_r^2} ||\mathbf{y} - X \boldsymbol{\beta}||_2^2} \cdot  \mathbbm{1}_{\mathcal{A}_{z_j}}(\beta_j)  \cdot e^{-\frac{1}{2 \sigma_{z_j}^2} \beta_j^2}  d \beta_j \\
&=  \frac{p(\mathbf{z})}{\iota(\mathcal{A}_{z_j}, \sigma_{z_j}^2)}  \cdot \int e^{-\frac{1}{2\sigma_r^2} ( ||\tilde{\mathbf{y}}||^2_2 - 2 \tilde{\mathbf{y}}^T \mathbf{x}_j \beta_j  + || \mathbf{x}_j ||_2^2 \beta_j^2)} 
\cdot \mathbbm{1}_{\mathcal{A}_{z_j}}(\beta_j) \cdot e^{-\frac{1}{2 \sigma_{z_j}^2} \beta_j^2} d \beta_j \\
&\propto  \frac{p(\mathbf{z})}{\iota(\mathcal{A}_{z_j}, \sigma_{z_j}^2)}  \cdot \int e^{-\frac{1}{2\sigma_r^2} (- 2 \tilde{\mathbf{y}}^T \mathbf{x}_j \beta_j  + (|| \mathbf{x}_j ||_2^2  + \frac{\sigma_r^2}{\sigma_{z_j}^2}) \beta_j^2)} 
\cdot \mathbbm{1}_{\mathcal{A}_{z_j}}(\beta_j) d \beta_j \\
 &=  \frac{p(\mathbf{z})}{\iota(\mathcal{A}_{z_j}, \sigma_{z_j}^2)}  \cdot \int e^{-\frac{1}{2 \tilde{\sigma}^2} (\beta_j  - \tilde{\mu})^2} e^{ \frac{\tilde{\mu}}{2 \sigma_r^2} \tilde{\mathbf{y}}^T \mathbf{x}_j  }
\cdot \mathbbm{1}_{\mathcal{A}_{z_j}}(\beta_j) d \beta_j \\
 &= p(\mathbf{z}) \cdot e^{ \frac{\tilde{\mu}}{2 \sigma_r^2} \tilde{\mathbf{y}}^T \mathbf{x}_j  } \cdot  \frac{\iota(\mathcal{A}_{z_j},  \tilde{\mu}, \tilde{\sigma}^2)}{\iota(\mathcal{A}_{z_j}, \sigma_{z_j}^2)}  \, ,
\end{align*}
where $\tilde{\mu} := \frac{\tilde{\mathbf{y}}^T \mathbf{x}_j}{ || \mathbf{x}_j ||^2_2 + \frac{\sigma_r^2}{\sigma_{z_j}^2}}$, and $\tilde{\sigma}^2 := ( \frac{1}{\sigma_r^2}  || \mathbf{x}_j ||^2_2+ \frac{1}{\sigma_{z_j}^2} )^{-1}$,
and $\iota(\mathcal{A}_{z_j},  \tilde{\mu}, \tilde{\sigma}^2)$ is the normalization constant of a truncated normal distribution given by 
\begin{align*}
\iota(\mathcal{A}_{z_j},  \tilde{\mu}, \tilde{\sigma}^2) :=\int e^{-\frac{1}{2 \tilde{\sigma}^2} (\beta_j  - \tilde{\mu})^2} \cdot \mathbbm{1}_{\mathcal{A}_{z_j}}(\beta_j) d \beta_j  \, .
\end{align*}

\subsubsection{Case $\delta = 0$.}
In the case, where $\delta = 0$, some care is needed. First, consider $z_j  = 1$, then we can proceed as before 
\begin{align*}
p(z_j = 1 | \beta_{-j}, \sigma_r,  \sigma_0, \sigma_1, \mathbf{y}, X, \mathbf{z}_{-j}) = c \cdot  p(\mathbf{z}) \cdot e^{ \frac{\tilde{\mu}}{2 \sigma_r^2} \tilde{\mathbf{y}}^T \mathbf{x}_j  } \cdot  \frac{\iota(\mathcal{\mathbb{R}},   \tilde{\mu}, \tilde{\sigma}^2)}{\iota(\mathcal{\mathbb{R}}, \sigma_{1}^2)} \, ,
\end{align*}
where $c$ is a normalization constant.
Second, for $z_j  = 0$, the prior $p(\beta_j)$ is a Dirac measure with $1$ at position $0$, and otherwise 0. Therefore, we can use the same calculation as before, but replacing $\beta_j$ by $0$. This way, we get
\begin{align*}
p(z_j = 0 | \beta_{-j}, \sigma_r,  \sigma_0, \sigma_1, \mathbf{y}, X, \mathbf{z}_{-j}) = c \cdot  p(\mathbf{z}) \, .
\end{align*}
Note that in both cases, we can integrate over $\beta_j$, and therefore the reversible jump MCMC methodology \citep{green1995reversible,green2009reversible} is not necessary here. 

\subsection{Analytic solution for $p(\beta_j | \boldsymbol{\beta}_{-j}, \mathbf{z}, \sigma_r , \sigma_1, \sigma_0, \mathbf{y},  X)$}

For $\delta > 0$, we have
\begin{align*}
p(\beta_j | \boldsymbol{\beta}_{-j}, \mathbf{z}, \sigma_r , \sigma_1, \sigma_0, \mathbf{y},  X)
&\propto p(\boldsymbol{\beta}, \mathbf{z}, \sigma_r, \sigma_1, \mathbf{y} |  X, \sigma_0) \\
&= p(\mathbf{z}) \cdot C_0 \cdot   (\sigma_r^2)^{-\frac{n}{2}} e^{-\frac{1}{2\sigma_r^2} ||\mathbf{y} - X \boldsymbol{\beta}||_2^2}  \\
&\quad \cdot (\sigma_r^2)^{-(\frac{\nu_r}{2} + 1)} e^{-\frac{1}{2 \sigma_r^2} \nu_r \eta_r^2}  \\
& \quad  \cdot  \Big( \prod_{j \in \mathcal{C}}  \mathbbm{1}_{\mathcal{N}}(\beta_j) \cdot \frac{1}{\iota(\mathcal{N}, \sigma_0^2)}  \cdot e^{-\frac{1}{2 \sigma_0^2} \beta_j^2} \Big)  \\
&\quad \cdot \Big( \prod_{j \in \mathcal{S}}  \mathbbm{1}_{\mathcal{R}}(\beta_j) \cdot \frac{1}{\iota(\mathcal{R}, \sigma_1^2)}  e^{-\frac{1}{2 \sigma_1^2} \beta_j^2}    \Big)  \\
&\propto  e^{-\frac{1}{2\sigma_r^2} ||\mathbf{y} - X \boldsymbol{\beta}||_2^2} \cdot  \mathbbm{1}_{\mathcal{A}_{z_j}}(\beta_j)  \cdot e^{-\frac{1}{2 \sigma_{z_j}^2} \beta_j^2}  \\
&=  e^{-\frac{1}{2\sigma_r^2} ||\tilde{\mathbf{y}} - \mathbf{x}_j \beta_j ||_2^2} \cdot  \mathbbm{1}_{\mathcal{A}_{z_j}}(\beta_j)  \cdot e^{-\frac{1}{2 \sigma_{z_j}^2} \beta_j^2}  \\
&=  e^{-\frac{1}{2\sigma_r^2} ( ||\tilde{\mathbf{y}}||^2_2 - 2 \tilde{\mathbf{y}}^T \mathbf{x}_j \beta_j  + || \mathbf{x}_j ||_2^2 \beta_j^2 )} \cdot  \mathbbm{1}_{\mathcal{A}_{z_j}}(\beta_j)  \cdot e^{-\frac{1}{2 \sigma_{z_j}^2} \beta_j^2}  \\
&\propto e^{-\frac{1}{2\sigma_r^2} (- 2 \tilde{\mathbf{y}}^T \mathbf{x}_j \beta_j  + (|| \mathbf{x}_j ||_2^2  + \frac{\sigma_r^2}{\sigma_{z_j}^2}) \beta_j^2)}  \cdot \mathbbm{1}_{\mathcal{A}_{z_j}}(\beta_j) \\
 &=  e^{-\frac{1}{2 \tilde{\sigma}^2} (\beta_j  - \tilde{\mu})^2} e^{ \frac{\tilde{\mu}}{2 \sigma_r^2} \tilde{\mathbf{y}}^T \mathbf{x}_j  } \cdot \mathbbm{1}_{\mathcal{A}_{z_j}}(\beta_j) \\
&\propto  N_{\mathcal{A}_{z_j}}(\beta_j | \tilde{\mu},  \tilde{\sigma}^2) \, .
\end{align*}

Note that if $\delta = 0$, then 
\begin{align*}
p(\beta_j | \boldsymbol{\beta}_{-j}, \mathbf{z}, \sigma_r , \sigma_1, \sigma_0, \mathbf{y},  X) = 
\begin{cases} 
   N(\beta_j | \tilde{\mu},  \tilde{\sigma}^2)  & \text{if } z_j = 1 \, ,\\
   1_{\{0\}}(\beta_j)       & \text{if } z_j = 0 \, .
  \end{cases}
\end{align*}

\subsection{Analytic solution for $p(\sigma_r^2 | \boldsymbol{\beta},  \mathbf{z}, \mathbf{y}, X)$} \label{sec:sigma_r_conditional_prob}

For the conditional posterior $p(\sigma_r^2 | \boldsymbol{\beta},  \mathbf{z}, \mathbf{y}, X)$, we have a closed form solution given by 
\begin{align*}
p(\sigma_r^2 | \boldsymbol{\beta},  \mathbf{y}, \mathbf{z}, X) & \propto p(\boldsymbol{\beta}, \sigma_r, \mathbf{y}, \mathbf{z} |  X)  \\
&= p(\mathbf{z}) \cdot C_0 \cdot   (\sigma_r^2)^{-\frac{n}{2}} e^{-\frac{1}{2\sigma_r^2} ||\mathbf{y} - X \boldsymbol{\beta}||_2^2}  \\
&\quad \cdot (\sigma_r^2)^{-(\frac{\nu_r}{2} + 1)} e^{-\frac{1}{2 \sigma_r^2} \nu_r \eta_r^2}  \\
& \quad  \cdot  \Big( \prod_{j \in \mathcal{C}}  \mathbbm{1}_{\mathcal{N}}(\beta_j) \cdot \frac{1}{\iota(\mathcal{N}, \sigma_0^2)}  \cdot e^{-\frac{1}{2 \sigma_0^2} \beta_j^2} \Big)  \\
&\quad \cdot \Big( \prod_{j \in \mathcal{S}}  \mathbbm{1}_{\mathcal{R}}(\beta_j) \cdot \frac{1}{\iota(\mathcal{R}, \sigma_1^2)}  e^{-\frac{1}{2 \sigma_1^2} \beta_j^2}    \Big)  \\
&\propto  (\sigma_r^2)^{-\frac{n}{2}} e^{-\frac{1}{2\sigma_r^2} ||\mathbf{y} - X \boldsymbol{\beta}||_2^2} 
\cdot (\sigma_r^2)^{-(\frac{\nu_r}{2} + 1)} e^{-\frac{1}{2 \sigma_r^2} \nu_r \eta_r^2} \\
&\propto  (\sigma_r^2)^{-(\frac{\nu_r + n}{2} + 1)} e^{-\frac{1}{2 \sigma_r^2} ( ||\mathbf{y} - X \boldsymbol{\beta}||_2^2  + \nu_r \eta_r^2)} \\
&\propto  (\sigma_r^2)^{-(\frac{\nu_r + n}{2} + 1)} e^{-\frac{1}{2 \sigma_r^2}(\nu_r + n) \frac{||\mathbf{y} - X \boldsymbol{\beta}||_2^2  + \nu_r \eta_r^2}{\nu_r + n}} \\
&\propto  \text{Inv-}\chi^2(\sigma^2_r \, | \, \nu_r + n , \frac{||\mathbf{y} - X \boldsymbol{\beta}||_2^2  + \nu_r \eta_r^2}{\nu_r + n}) \, .
\end{align*}

\subsection{Sampling from $p(\sigma_1^2 | \boldsymbol{\beta}, \sigma_r^2, \mathbf{z}, \mathbf{y},  X)$}

For sampling from $p(\sigma_1^2 | \boldsymbol{\beta}, \sigma_r^2, \mathbf{z}, \mathbf{y},  X)$, we employ a Slice sampler as described in the following.
First note that 
\begin{align*}
p(\sigma_1^2 | \boldsymbol{\beta}, \sigma_r^2, \mathbf{y}, \mathbf{z},  X) 
&\propto p(\sigma_1^2, \boldsymbol{\beta}, \sigma_r^2, \mathbf{y}, \mathbf{z} |  X) \\
&= p(\mathbf{z}) \cdot C_0 \cdot   (\sigma_r^2)^{-\frac{n}{2}} e^{-\frac{1}{2\sigma_r^2} ||\mathbf{y} - X \boldsymbol{\beta}||_2^2}  \\
&\quad \cdot (\sigma_r^2)^{-(\frac{\nu_r}{2} + 1)} e^{-\frac{1}{2 \sigma_r^2} \nu_r \eta_r^2}  \\
& \quad  \cdot  \Big( \prod_{j \in \mathcal{C}}  \mathbbm{1}_{\mathcal{N}}(\beta_j) \cdot \frac{1}{\iota(\mathcal{N}, \sigma_0^2)}  \cdot e^{-\frac{1}{2 \sigma_0^2} \beta_j^2} \Big)  \\
&\quad \cdot \Big( \prod_{j \in \mathcal{S}}  \mathbbm{1}_{\mathcal{R}}(\beta_j) \cdot \frac{1}{\iota(\mathcal{R}, \sigma_1^2)}  e^{-\frac{1}{2 \sigma_1^2} \beta_j^2}    \Big)  \\
&\quad \cdot (\sigma_1^2)^{-(\frac{\nu_1}{2} + 1)} e^{-\frac{1}{2 \sigma_1^2} \nu_1 \eta_1^2}  \\
&\propto  \Big( \prod_{j \in \mathcal{S}}  \mathbbm{1}_{\mathcal{R}}(\beta_j) \cdot \frac{1}{\iota(\mathcal{R}, \sigma_1^2)}  \cdot e^{-\frac{1}{2 \sigma_1^2} \beta_j^2} \Big)  \\
&\quad \cdot (\sigma_1^2)^{-(\frac{\nu_1}{2} + 1)} e^{-\frac{1}{2 \sigma_1^2} \nu_1 \eta_1^2}  \\
&\propto \frac{1}{\iota(\mathcal{R}, \sigma_1^2)^s} \cdot (\sigma_1^2)^{-(\frac{\nu_1}{2} + 1)} e^{-\frac{1}{2 \sigma_1^2} (\nu_1 \eta_1^2 + \sum_{j \in \mathcal{S}} \beta_j^2 ) }  \, .
\end{align*}
If $\sigma_1^2 \gg 1$, and $\delta \ll 1$, we have \emph{approximately} that 
\begin{align} \label{eq:approximationForSamplingSigma1}
\iota(\mathcal{R}, \sigma_1^2) 
\propto \iota(\mathbb{R}, \sigma_1^2) = (2 \pi \sigma_1^2)^{\frac{1}{2}} \, ,
 \end{align}
 and we have \emph{exactly} (not approximately) that
\begin{align*}
p(\sigma_1^2 | \boldsymbol{\beta}, \sigma_r^2, \mathbf{y}, \mathbf{z},  X) 
&\propto  \Big( \frac{(2 \pi \sigma_1^2)^{\frac{s}{2}} }{\iota(\mathcal{R}, \sigma_1^2)^s} \Big) 
\cdot (2 \pi \sigma_1^2)^{-\frac{s}{2}}
\cdot  \Big( (\sigma_1^2)^{-(\frac{\nu_1}{2} + 1)} e^{-\frac{1}{2 \sigma_1^2} (\nu_1 \eta_1^2 + \sum_{j \in \mathcal{S}} \beta_j^2 ) }  \Big) \\
&\propto  \Big( \frac{(2 \pi \sigma_1^2)^{\frac{s}{2}} }{\iota(\mathcal{R}, \sigma_1^2)^s} \Big) 
\cdot  \Big( (\sigma_1^2)^{-(\frac{\nu_1 + s}{2} + 1)} e^{-\frac{1}{2 \sigma_1^2} (\nu_1 \eta_1^2 + \sum_{j \in \mathcal{S}} \beta_j^2 ) }  \Big) \\
&\propto  \Big( \frac{(2 \pi \sigma_1^2)^{\frac{s}{2}} }{\iota(\mathcal{R}, \sigma_1^2)^s} \Big) 
\cdot  \Big( (\sigma_1^2)^{-(\frac{\nu_1 + s}{2} + 1)} e^{-\frac{1}{2 \sigma_1^2} (\nu_1 + s) \frac{(\nu_1 \eta_1^2 + \sum_{j \in \mathcal{S}} \beta_j^2 )}{\nu_1 + s} }  \Big) \\
&\propto  \Big( \frac{(2 \pi \sigma_1^2)^{\frac{s}{2}} }{\iota(\mathcal{R}, \sigma_1^2)^s} \Big) 
\cdot  \text{Inv-}\chi^2(\nu_1 + s, \frac{(\nu_1 \eta_1^2 + \sum_{j \in \mathcal{S}} \beta_j^2 )}{\nu_1 + s} \Big)  \, .
\end{align*}
That means we have that 
\begin{align*}
p(\sigma^2_1 | \boldsymbol{\beta}, \sigma_r^2, \mathbf{y}, X, \mathcal{S})  
&\propto h(\sigma^2_1) \cdot  \text{Inv-}\chi^2(\sigma^2_1 \, | \, \tilde{\nu} , \tilde{\eta}^2) \, ,
\end{align*}
for $\tilde{\nu} := \nu_1 + s$, $\tilde{\eta}^2 := \frac{(\nu_1 \eta_1^2 + \sum_{j \in \mathcal{S}} \beta_j^2 )}{\nu_1 + s}$, and the function $h(\sigma_1^2) := \frac{(2 \pi \sigma_1^2)^{\frac{s}{2}} }{\iota(\mathcal{R}, \sigma_1^2)^s} $ is changing slowly with $\sigma_1^2$.  
Therefore, we use a slice sampler (see e.g. \cite{carlin2008bayesian}) as follows. We start from the (approximate) mode given by $\sigma^2_1 := \frac{\tilde{\nu} \tilde{\eta}^2}{\tilde{\nu} + 2}$, and then run the following two steps, until we retain a sample in the second step:\footnote{We assume that we started in a high probability region, and therefore use a burn-in of only 10.}
\begin{enumerate}
\item Sample $U \sim \text{Uniform}( [0, h(\sigma^2_1)] )$.
\item Sample $\sigma^2_1 \sim \text{Inv-}\chi^2( \tilde{\nu} , \tilde{\eta}^2)$, and retain the sample if $U < h(\sigma^2_1)$.
\end{enumerate}

Note that the sampling scheme is guaranteed to sample exactly from \\
$p(\sigma^2_1 | \boldsymbol{\beta}, \sigma_r^2, \mathbf{y}, X, \mathcal{S})$, independently of how well the approximation $h(\sigma^2_1) \propto 1$ holds. The correctness of the sampling scheme is shown in Appendix B. % \ref{app:sliceSampler}.
However, of course, the efficiency (whether we accept the sample in step 2) will depend on the closeness of the approximation in Equation \eqref{eq:approximationForSamplingSigma1}.
In practice, we observe that the sampling method is efficient if $s$ is small. 
In detail, for several settings, for $s = 1$, and $s = 10$, the lowest acceptance rates were around 97\% and 67\%, respectively, where
we tested $\sum_{j \in \mathcal{S}} \beta_j^2 \in \{0.1, 1.0, 10.0, 100.0\}$, and $\delta = \{0.8, 0.05, 0.001\}$.

\section{Specification of $\delta$} \label{sec:deltaSpecification_regression}

In some situations, where prior knowledge is given in the form of similar regression tasks from the past, it is possible to directly elicit a suitable threshold value $\delta$. 

% \subsection{Estimating expected increase of Mean Squared Error} \label{sec:estimatingExpectedIncreaseInMSE}
As an alternative, several plausible values for $\delta$ might be evaluated in terms of the expected increase of mean squared error (MSE).
% Furthermore, it only makes a statement in absolute terms of increase in MSE. However, often we are interested in statements like 
% "the selected model (with few variables) increases the mean squared error by no more than 5\% when compared to the best model that can use all variables." \cite{piironen2017comparison,hahn2015decoupling}.
As the final model, we can then select the model that is the sparsest and does not increase MSE by more than, for example, 5\% when compared to the best model
(see \cite{piironen2017comparison,hahn2015decoupling} for similar ideas).

For the ``best model" we use the Bayesian model averaged (BMA) regression model, since it is often considered the gold standard due to its good theoretic and practical performance \citep{fernandez2001benchmark,piironen2017comparison}.
The BMA model for the prediction of a new datapoint $(\tilde{y}, \tilde{\mathbf{x}})$ is defined as 
\begin{align*}
p(\tilde{y} | \tilde{\mathbf{x}}) = \sum_{\mathbf{z}} \int p(\tilde{y} | \tilde{\mathbf{x}}, \mathbf{z} , \boldsymbol{\theta}) p(\mathbf{z}, \boldsymbol{\theta} | \mathbf{y}, X) d\boldsymbol{\theta} \, ,
\end{align*}
where $\boldsymbol{\theta}$ denotes all parameters. The BMA model is a meta-model since it still requires the specification of the model for $p(\mathbf{z}, \boldsymbol{\theta}, y | X)$.
Here, we use for $p(\mathbf{z}, \boldsymbol{\theta}, y | X)$, our proposed model with $\delta = 0$.

The expected mean squared error of BMA is therefore given by
\begin{align*}
\text{MSE}_{\text{bma}} :=  \E_{\mathbf{z}} [ \E_{\sigma_r^2} [\sigma_r^2 | \mathbf{z}, \mathbf{y}, X]  | \mathbf{y}, X] \, ,
\end{align*}
which we estimate from the samples of our MCMC algorithm in Algorithm \ref{alg:propGibsForZ}.

Given a threshold $\delta^*$, and the best subset of variables specified by $\mathbf{z}^*$, we estimate the MSE as follows 
\begin{align*}
\text{MSE}_{\delta^*} := \E_{\sigma_r^2} [\sigma_r^2 | \mathbf{z}^*, \mathbf{y}, X_{| \mathbf{z}^*}, \delta^*] \, ,
\end{align*}
where $X_{| \mathbf{z}^*}$ means that only the covariates index by $\mathbf{z}^*$ are used, where 
\begin{align*} % \label{eq:specOfZ}
\mathbf{z}^* := \argmax_{\mathbf{z}} p(\mathbf{z} | \mathbf{y}, X, \delta^*) \, .
\end{align*}
We can now estimate for each threshold $\delta$ the expected increase in MSE when compared to $\text{MSE}_{\text{bma}}$, i.e.:
\begin{align} \label{eq:MSE_increaseEstimation}
\text{expected increase in MSE} = \frac{\text{MSE}_{\delta^*}}{\text{MSE}_{\text{bma}}} - 1.0 \, . 
\end{align}
We then select the most parsimonious model that has an expected increase in MSE of less than 5\%.
% We note that similar strategies for predictive model selection have been proposed in \cite{piironen2017comparison,hahn2015decoupling}, though, their models are different from ours,
% and they do not make use of $p(\mathbf{z} | \mathbf{y}, X, \delta^*)$ as in Equation \eqref{eq:specOfZ}.

\section{Evaluation on synthetic data} \label{sec:evaluation_synthetic_regression}

We study two settings, the low-dimensional setting with $d < n$ and the high-dimensional setting with $d \geq n$.

For the low-dimensional experiments,  we use the same regression setting as in \citep{tibshirani1996regression}, namely the regression coefficient vector is set to
\begin{align*}
\boldsymbol{\beta}^T = ({\bf  3}, {\bf  1.5}, 0.0, 0.0, {\bf  2.0}, 0.0, 0.0, 0.0)^T \, ,
\end{align*}
and the response noise is set to $\sigma_r = 3.0$.
For each sample, we draw a covariate vector $\mathbf{x} \sim N(\mathbf{0}, \Sigma)$, where $\Sigma_{ij} = 0.5^{|i-j|}$.
The number of samples is varied from $n = 10$ to $n = 100000$. 

For the high-dimensional experiments, we use the same setting as in \citep{rovckova2014emvs}, with $d = 1000$ and $n \in \{100, 1000\}$, where 
the first three covariate are set to $3$ ,$2$, and $1$, and all others are set to zero.
The covariate vector is drawn from $\mathbf{x} \sim N(\mathbf{0}, \Sigma)$, where $\Sigma_{ij} = 0.6^{|i-j|}$.

Furthermore, in the noise setting, we replace each zero entry of the original regression coefficient vector by a value sampled from $\text{Uniform}([-\eta, \eta])$, where $\eta \in \{0.2, 0.5\}$.
For example, when $\eta = 0.5$, the new regression coefficient vector for the low-dimensional experiment becomes
\begin{align*}
\boldsymbol{\beta}^T = ({\bf 3}, {\bf 1.5}, -0.12, -0.35, {\bf 2.0}, 0.16, 0.26, -0.01)^T \, ,
\end{align*}
where the relevant variables are marked by bold font. 
The expected increase in mean squared error (MSE) for choosing the parsimonious model without the noise coefficients is about $0.4\%$ and $2.8\%$, for $\eta = 0.2$, and $\eta = 0.5$, respectively.

In the high-dimensional noise setting, we replace only $1\%$ of the original zero entries (following the largest entries 3, 2, and 1).
For choosing the parsimonious model (i.e. only relevant variables), this leads to an expected increase in mean squared error of about $3.1\%$ for $\eta = 0.2$.\footnote{For the high-dimensional setting we do not consider $\eta = 0.5$, since this would correspond to an expected MSE increase of $19.0\%$.}

All methods are evaluated in terms of identifying the set of relevant variables.

\subsection{Analysis of Bayes factors} \label{sec:expAnalysisBF}
First, we investigate the advantage of disjunct support priors in contrast to full support priors, and the product moment matching priors from \citep{johnson2012bayesian}.
For the full support priors, we replace the truncated normal spike and slab priors by non-truncated ones, i.e.
the model specification from Section \ref{sec:proposedMethod_regression} stays the same except that, 
if $j \in S$, then $\beta_j \sim N (0,  \sigma_1^2)$, else $\beta_j \sim N (0,  \sigma_0^2)$.
Sampling from the posterior probabilities is analogously to Algorithm \ref{alg:propGibsForZ}. 
Furthermore, we compare also to the product moment matching priors (MOM) from \citep{johnson2012bayesian}.\footnote{Implemented in the R package 'mombf'.}
The critical hyper-parameter of MOM is $\tau$ which controls the definition of practical relevance. In particular, following \citep{johnson2010use}, we set $\tau$ such that $P(\beta \in [-\delta, \delta]) = 0.01$.
For all methods, we fix the threshold of practical relevance $\delta$ to $0.5$.
We evaluate the Bayes factors (BF) in favor for the true model, defined as 
\begin{align*}
\text{BF in favor for true model} = \frac{p(\mathbf{y} | X, S)} {p(\mathbf{y} | X, S')}  \, ,
\end{align*}
where $S$ is the true model and $S'$ is the most frequently selected model $S_1$, in case where $S \neq S_1$, otherwise we set $S'$ to the second most frequently selected model $S_2$.
This means, $\text{BF} <  1$, if the most frequent model was not the true model, and $\text{BF} >  1$, denotes the Bayes factor compared to the second best model. 
We hope to observe that $\text{BF}$ grows with increasing sample size $n$. 
As can be seen in the results in Table \ref{tab:simData_lowDimensional_BF}, this is indeed the case for the proposed model with disjunct support priors, but not always the case for full support priors and the MOM priors.
This confirms the asymptotic results from Theorem \ref{prop:proofBFproposed_regression}, \ref{prop:proofBFprevious_regression} and \ref{prop:proofBF_nonlocalPriors_regression}.
In particular, when $n \geq 1000$, we observe in both, the low and high-dimensional setting, higher Bayes factors for the proposed method with disjunct support priors than for full support and MOM priors.
% Importantly, even when $n$ is in the order of the number of dimensions $d$, or only slightly above, i.e. $n \geq 50$ (for low-dimensional setting) and $n = 1000$ (for high-dimensional setting), we observe higher Bayes factors when using the proposed method with disjunct support priors. 

  \begin{table*}[h]
  \footnotesize
  \caption{Shows the Bayes factors favoring the true model with disjunct support priors, full support priors and moment matching priors (MOM) from \citep{johnson2012bayesian}. Low-dimensional setting, $d = 8$ and $n \in \{10, 50, 100, 1000, 100000\}$, and high-dimensional setting, $d = 1000$ and $n \in \{100, 1000\}$.
  The threshold $\delta$ is fixed to $0.5$. In case where the true model was always selected (i.e. no second most frequent model), we set the Bayes factor to $\infty$. All experiments are repeated 10 times and the average and standard deviation (in brackets) are reported.}
  \label{tab:simData_lowDimensional_BF}
  
  \begin{flushleft}
  \begin{tabular}{llllll}
  \\
     \multicolumn{6}{c}{\footnotesize Low-dimensional setting ($d$ = 8)} \\
  \toprule 
      \multicolumn{6}{c}{\footnotesize No noise on regression coefficients} \\
      \midrule
 \footnotesize & \footnotesize 10 & \footnotesize 50 & \footnotesize 100 & \footnotesize 1000 & \footnotesize 100000 \\
\midrule
\footnotesize disjunct support   & \footnotesize 0.1 (0.12) & 2.13 (2.44) & 18.34 (18.88) & $\infty$ (-) & $\infty$ (-) \\
\footnotesize full support   & \footnotesize 0.26 (0.27) & 1.81 (1.83) & 6.68 (4.42) & 18.41 (2.67) & 30.23 (0.68) \\
\footnotesize MOM   & \footnotesize 0.02 (0.02) & 2.48 (3.49) & 93.98 (172.34) & 8755.93 (9011.52) & 8886733.9 (11906002.92) \\
 \midrule
\multicolumn{6}{c}{Noise on regression coefficients ($\eta = 0.5$)} \\
\midrule
disjunct support   & 0.12 (0.29) & 5.39 (6.87) & 12.67 (14.19) & 120.98 (151.28) & $\infty$ (-) \\
full support   & 0.37 (0.59) & 3.72 (4.65) & 5.68 (4.13) & 6.24 (3.79) & 4.24 (0.51) \\
MOM   & 0.02 (0.03) & 7.98 (17.19) & 45.4 (89.39) & 48.56 (76.55) & 0.0 (0.0) \\
\bottomrule
\end{tabular}
$ $ \\
\vspace{0.5cm}
    \begin{tabular}{lll}
   \multicolumn{3}{c}{High-dimensional setting ($d$ = 1000)} \\
  \toprule 
      \multicolumn{3}{c}{No noise on regression coefficients} \\
      \midrule
 & 100 & 1000  \\
\midrule
disjunct support   & 2.12 (6.3) & $\infty$ (-) \\
full support   & 7.02 (21.06) & 55924.83 (44821.0) \\
MOM   & 1444.77 (3444.03) & 252.36 (520.61) \\
 \midrule
\multicolumn{3}{c}{Noise on regression coefficients ($\eta = 0.2$)} \\
\midrule
disjunct support   & 0.0 (0.0) & 77308.22 (109895.97) \\
full support   & 0.01 (0.03) & 3085.49 (3492.8) \\
MOM   & 12.87 (11.64) & 8.32 (14.38) \\
\bottomrule
    \end{tabular}
    \end{flushleft}
 \end{table*}

\subsection{Comparison to other model selection methods} \label{sec:synExperimentsToOtherMethod}
We evaluate our proposed method for $\delta \in \{0.8, 0.5, 0.05, 0.01, 0.001, 0.0\}$, and select the most parsimonious model that is estimated to lead to an increase in MSE of not more than 5\% as was described in Section \ref{sec:deltaSpecification_regression}. For MCMC we use 10000 samples, out of which 10\% are used for burn in.

We compare to the Gaussian and Laplace spike-and-slab priors combined with the EM-algorithm as proposed in \citep{rovckova2014emvs,rovckova2018spike} which we denote as ``EMVS" and ``SSLASSO", respectively.\footnote{Implemented in the R package 'EMVS' and 'SSLASSO'.}
Note that EMVS and SSLASSO do not provide model or variable inclusion posterior probabilities. Here we show only the results for SSLASSO. The results for EMVS were always similar or worse than SSLASSO and are given in Appendix C. Comparison to the robust objective prior proposed in \citep{bayarri2012criteria} are also given in Appendix C.

The above methods cannot account for negligible noise on the coefficient vectors. Therefore, we introduce another baseline using the horseshoe prior \citep{carvalho2010horseshoe} as follows.\footnote{Implemented in the R package 'horseshoe'.}
First, using the horseshoe prior, we estimate the mean coefficient vector $\boldsymbol{\beta}$ and the mean response variance $\sigma_{r, full}^2$ for the full model. 
Then,  for each $\delta$, we hard threshold $\boldsymbol{\beta}$, and this way get a model candidate $\mathbf{z}_{\delta}$. 
Finally, using again the horseshoe prior for the linear regression model but reduced to the covariates $\mathbf{z}_{\delta}$,  we estimate the mean response variance $\sigma_{r, \mathbf{z}_{\delta}}^2$, and then select the most parsimonious model that has lower expected increase in MSE than 5\%.
To estimate the expected increase in MSE, we use Formula \eqref{eq:MSE_increaseEstimation},  where we replace $\text{MSE}_{\delta^*}$ and $\text{MSE}_{\text{bma}}$ by $\sigma_{r, \mathbf{z}_{\delta}}^2$ and 
$\sigma_{r, full}^2$, respectively.

Finally, we compared to three frequentist methods for model search.
We used the Least Angle Regression (LARS) method \citep{efron2004least} or Lasso \citep{tibshirani1996regression} to get a set of candidate models, and then ranked each model using either Akaike information criterion (AIC) \citep{akaike1973information}, the Bayesian information criterion (BIC) and its extensions \citep{schwarz1978estimating,chen2008extended,foygel2010extended}, or stability selection \citep{meinshausen2010stability}. We show here only the results for BIC. The other results are given in Appendix C.

\paragraph{Low-dimensional setting}
The results for the low-dimensional setting, with and without noise, are shown in Figure \ref{fig:syntheticResults_lowDim}.
Overall, we see that the proposed method and the horseshoe prior method perform best and can identify only the relevant set of variables in the noise setting, assuming sufficiently large $n$. 
Note that BIC and SSLASSO also perform good in the noise setting when the sample size is only small or moderate. 
This phenomena is likely to be due to that the sampling noise and the noise on the regression coefficients cannot be distinguished anymore, and as a consequence, BIC and SSLASSO tend to select the more parsimonious models. However, in the noise setting when $n \geq 1000$, sampling noise and the signal from the regression coefficients can be distinguished even for regression coefficients with very small magnitude, and as a consequence BIC and SSLASSO start to select also the irrelevant variables.  % 

\paragraph{High-dimensional setting}
The results for the high-dimensional setting, with and without noise, are shown in Figure \ref{fig:syntheticResults_highDim}.
Overall, the horseshoe prior method performs somehow unsatisfactory, tending to select too many variables. Inspecting the results for different $\delta$ confirmed this (details given in Appendix C). 
BIC performed very poorly in this setting, selecting too many variables. One reason seems to stem from the numerical instability of the maximum likelihood estimate for $d \leq n$.\footnote{As an ad-hoc remedy we tried to combine it with a ridge estimate, but this did not seem to help.} The proposed method and SSLASSO performed best in this setting. Interestingly, even in the noise setting, SSLASSO correctly selected mostly only the relevant variables, which is likely to be due to the same phenomena as described in the low-dimensional setting.

\paragraph{Comparison to moment matching priors}
Finally, we also compare to the product moment matching priors (MOM) from \citep{johnson2012bayesian}. 
The results are show in Tables \ref{tab:simDataNoise0_lowDimensional_deltaAnalysisVS_MOM} and \ref{tab:simDataNoise05_lowDimensional_deltaAnalysisVS_MOM}.
From these results, we can draw two important conclusions. 
First, as seen in Table \ref{tab:simDataNoise05_lowDimensional_deltaAnalysisVS_MOM}, when $n$ grows large enough, MOM starts to select also the noise variables. This is expected, since the MOM prior densities are only exactly zero at $\beta = 0$, but otherwise have strictly positive support.
The second observation is that for small values of the practical relevance threshold $\delta$ (values smaller or equal to $0.05$), MOM tends to select too few or too many variables.
The results for the high-dimensional setting are similar and show in Appendix C.

 \begin{figure*}[h]
  \centering
  \includegraphics[scale=0.35, trim=10cm 0cm 10cm 0cm]{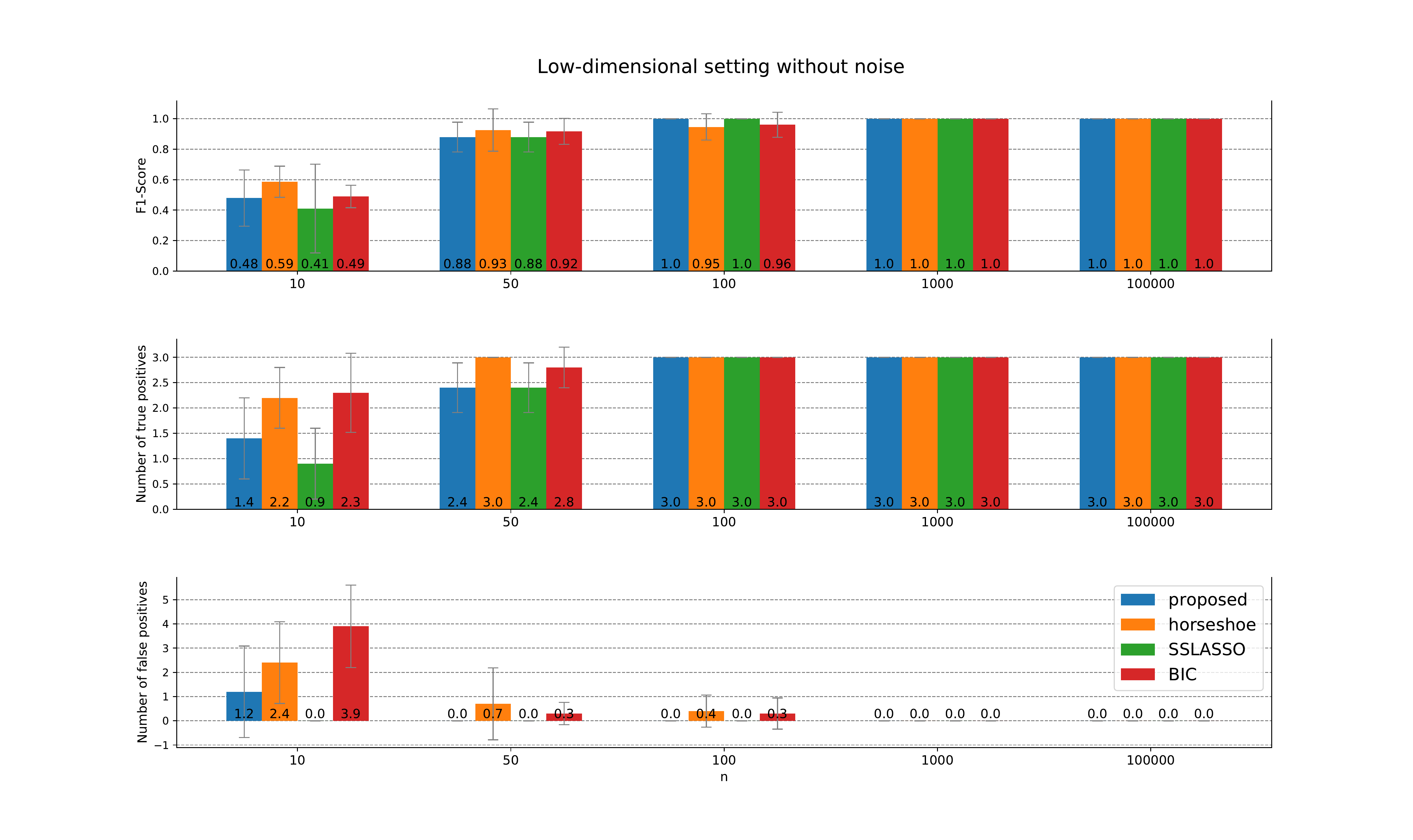}  % <left> <lower> <right> <upper> [scale=0.4, trim=10cm 0cm 10cm 0cm]
  \includegraphics[scale=0.35, trim=10cm 0cm 10cm 2cm]{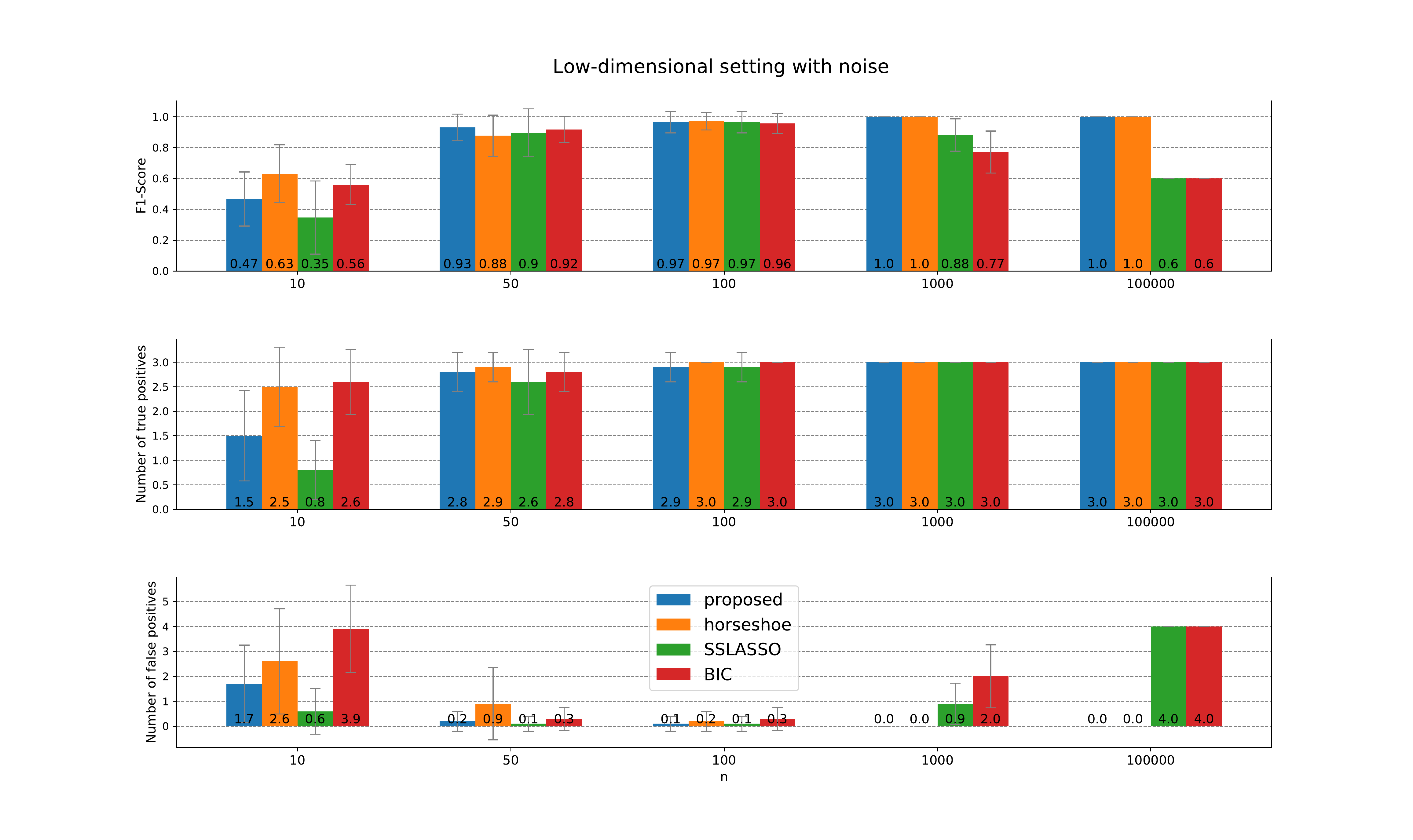}  % <left> <lower> <right> <upper> [scale=0.4, trim=10cm 0cm 10cm 0cm]
    \caption{Low-dimensional setting, $d = 8$ and $n \in \{10, 50, 100, 1000, 100000\}$. Evaluation results with no noise and noise ($\eta = 0.5$) on regression coefficients, shown in upper and lower half, respectively.}
  \label{fig:syntheticResults_lowDim}
\end{figure*} 

 \begin{figure*}[h]
  \centering
  \includegraphics[scale=0.35, trim=10cm 0cm 10cm 0cm]{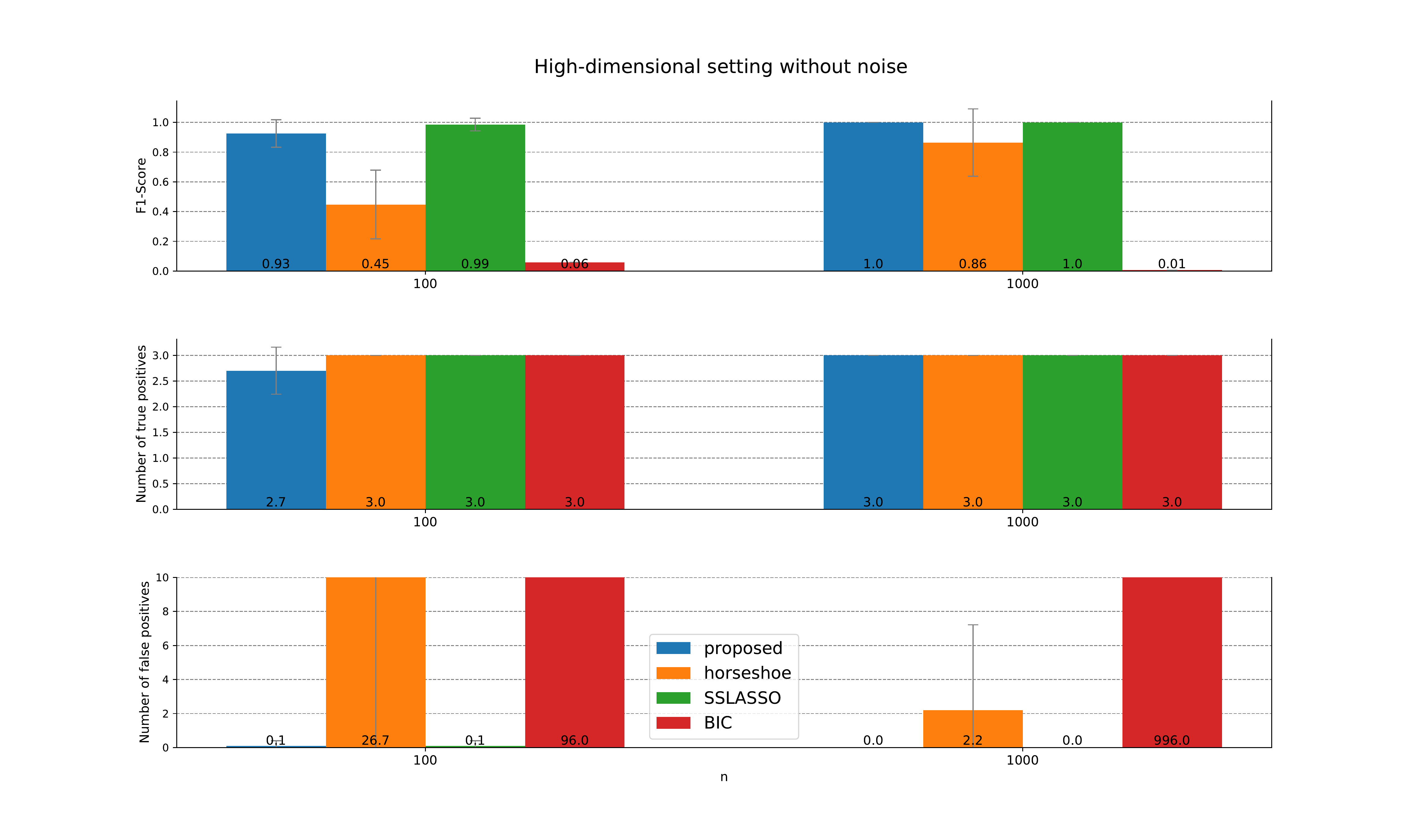}  % <left> <lower> <right> <upper> [scale=0.4, trim=10cm 0cm 10cm 0cm]
   \includegraphics[scale=0.35, trim=10cm 0cm 10cm 2cm]{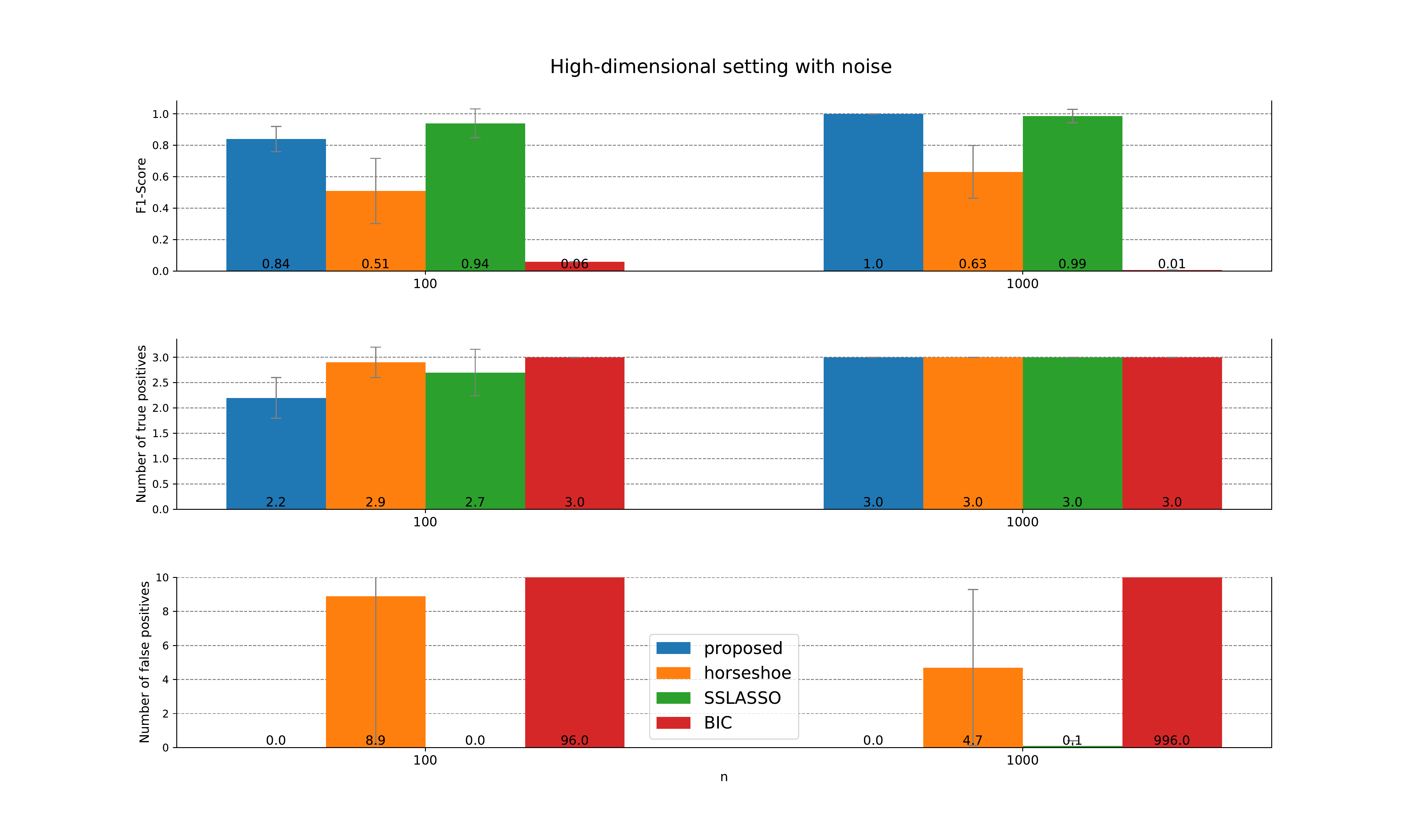}  % <left> <lower> <right> <upper> [scale=0.4, trim=10cm 0cm 10cm 0cm]
    \caption{High-dimensional setting, $d = 1000$ and $n \in \{100, 1000\}$. Evaluation results with no noise and noise ($\eta = 0.2$) on regression coefficients, shown in upper and lower half, respectively.}
  \label{fig:syntheticResults_highDim}
\end{figure*}

  % fully updated
  \begin{table*}[h]
 \center
  \caption{Low-dimensional setting, $d = 8$ and $n \in \{10, 50, 100, 1000, 100000\}$. Evaluation results with no noise on regression coefficients. For different $\delta$, comparison of the proposed method and product moment matching priors (MOM) from \citep{johnson2012bayesian}.}
  \label{tab:simDataNoise0_lowDimensional_deltaAnalysisVS_MOM}
  \begin{tabular}{llllll}
  \toprule 
      \multicolumn{6}{c}{F1-Scores} \\
      \midrule
 & 10 & 50 & 100 & 1000 & 100000 \\
\midrule
proposed ($\delta$ = 0.8)  & 0.47 (0.18) & 0.88 (0.1) & 1.0 (0.0) & 1.0 (0.0) & 1.0 (0.0) \\
proposed ($\delta$ = 0.5)  & 0.47 (0.18) & 0.88 (0.1) & 1.0 (0.0) & 1.0 (0.0) & 1.0 (0.0) \\
proposed ($\delta$ = 0.05)  & 0.47 (0.19) & 0.89 (0.09) & 1.0 (0.0) & 1.0 (0.0) & 1.0 (0.0) \\
proposed ($\delta$ = 0.01)  & 0.47 (0.19) & 0.89 (0.09) & 1.0 (0.0) & 1.0 (0.0) & 1.0 (0.0) \\
proposed ($\delta$ = 0.001)  & 0.48 (0.19) & 0.87 (0.09) & 0.98 (0.08) & 1.0 (0.0) & 1.0 (0.0) \\
\midrule
MOM ($\delta$ = 0.8)  & 0.52 (0.2) & 0.81 (0.13) & 0.98 (0.06) & 1.0 (0.0) & 1.0 (0.0) \\
MOM ($\delta$ = 0.5)  & 0.53 (0.2) & 0.88 (0.1) & 1.0 (0.0) & 1.0 (0.0) & 1.0 (0.0) \\
MOM ($\delta$ = 0.05)  & 0.38 (0.25) & 0.59 (0.14) & 0.77 (0.23) & 0.99 (0.04) & 1.0 (0.0) \\
MOM ($\delta$ = 0.01)  & 0.0 (0.0) & 0.55 (0.0) & 0.55 (0.0) & 0.55 (0.0) & 0.97 (0.06) \\
MOM ($\delta$ = 0.001)  & 0.0 (0.0) & 0.0 (0.0) & 0.0 (0.0) & 0.55 (0.0) & 0.55 (0.0) \\
\midrule
      \multicolumn{6}{c}{Average number of selected variables} \\
          \midrule
 & 10 & 50 & 100 & 1000 & 100000 \\
\midrule
proposed ($\delta$ = 0.8)  & 1.9 (1.87) & 2.4 (0.49) & 3.0 (0.0) & 3.0 (0.0) & 3.0 (0.0) \\
proposed ($\delta$ = 0.5)  & 1.9 (1.87) & 2.4 (0.49) & 3.0 (0.0) & 3.0 (0.0) & 3.0 (0.0) \\
proposed ($\delta$ = 0.05)  & 2.3 (2.65) & 2.6 (0.66) & 3.0 (0.0) & 3.0 (0.0) & 3.0 (0.0) \\
proposed ($\delta$ = 0.01)  & 2.3 (2.65) & 2.6 (0.66) & 3.0 (0.0) & 3.0 (0.0) & 3.0 (0.0) \\
proposed ($\delta$ = 0.001)  & 3.0 (3.1) & 2.5 (0.67) & 3.2 (0.6) & 3.0 (0.0) & 3.0 (0.0) \\
\midrule
MOM ($\delta$ = 0.8)  & 3.1 (2.77) & 2.1 (0.54) & 2.9 (0.3) & 3.0 (0.0) & 3.0 (0.0) \\
MOM ($\delta$ = 0.5)  & 3.0 (2.68) & 2.4 (0.49) & 3.0 (0.0) & 3.0 (0.0) & 3.0 (0.0) \\
MOM ($\delta$ = 0.05)  & 5.6 (3.67) & 7.5 (1.5) & 5.5 (2.5) & 3.1 (0.3) & 3.0 (0.0) \\
MOM ($\delta$ = 0.01)  & 0.0 (0.0) & 8.0 (0.0) & 8.0 (0.0) & 8.0 (0.0) & 3.2 (0.4) \\
MOM ($\delta$ = 0.001)  & 0.0 (0.0) & 0.0 (0.0) & 0.0 (0.0) & 8.0 (0.0) & 8.0 (0.0) \\
     \bottomrule
    \end{tabular}
 \end{table*}

 % fully updated
    \begin{table*}[h]
 \center
  \caption{Low-dimensional setting, $d = 8$ and $n \in \{10, 50, 100, 1000, 100000\}$. Evaluation results with noise on regression coefficients $\eta = 0.5$. For different $\delta$, comparison of the proposed method and product moment matching priors (MOM) from \citep{johnson2012bayesian}.}
  \label{tab:simDataNoise05_lowDimensional_deltaAnalysisVS_MOM}
  \begin{tabular}{llllll}
  \toprule 
      \multicolumn{6}{c}{F1-Scores} \\
      \midrule
 & 10 & 50 & 100 & 1000 & 100000 \\
\midrule
proposed ($\delta$ = 0.8)  & 0.36 (0.25) & 0.85 (0.19) & 0.97 (0.07) & 1.0 (0.0) & 1.0 (0.0) \\
proposed ($\delta$ = 0.5)  & 0.45 (0.18) & 0.93 (0.09) & 0.97 (0.07) & 1.0 (0.0) & 1.0 (0.0) \\
proposed ($\delta$ = 0.05)  & 0.34 (0.24) & 0.93 (0.09) & 0.95 (0.08) & 0.8 (0.11) & 0.6 (0.0) \\
proposed ($\delta$ = 0.01)  & 0.34 (0.24) & 0.93 (0.09) & 0.95 (0.08) & 0.8 (0.11) & 0.59 (0.02) \\
proposed ($\delta$ = 0.001)  & 0.34 (0.24) & 0.93 (0.09) & 0.95 (0.08) & 0.8 (0.11) & 0.59 (0.02) \\
\midrule
MOM ($\delta$ = 0.8)  & 0.42 (0.24) & 0.74 (0.21) & 0.95 (0.08) & 0.97 (0.06) & 0.6 (0.0) \\
MOM ($\delta$ = 0.5)  & 0.46 (0.2) & 0.81 (0.17) & 0.97 (0.07) & 0.96 (0.07) & 0.6 (0.0) \\
MOM ($\delta$ = 0.05)  & 0.27 (0.27) & 0.55 (0.0) & 0.77 (0.23) & 0.76 (0.1) & 0.6 (0.0) \\
MOM ($\delta$ = 0.01)  & 0.0 (0.0) & 0.55 (0.0) & 0.55 (0.0) & 0.55 (0.0) & 0.58 (0.02) \\
MOM ($\delta$ = 0.001)  & 0.0 (0.0) & 0.0 (0.0) & 0.0 (0.0) & 0.55 (0.0) & 0.55 (0.0) \\
\midrule
      \multicolumn{6}{c}{Average number of selected variables} \\
          \midrule
 & 10 & 50 & 100 & 1000 & 100000 \\
\midrule
proposed ($\delta$ = 0.8)  & 2.4 (2.11) & 2.5 (0.92) & 3.0 (0.45) & 3.0 (0.0) & 3.0 (0.0) \\
proposed ($\delta$ = 0.5)  & 3.2 (2.36) & 3.0 (0.63) & 3.0 (0.45) & 3.0 (0.0) & 3.0 (0.0) \\
proposed ($\delta$ = 0.05)  & 2.5 (2.42) & 3.0 (0.63) & 3.1 (0.54) & 4.6 (0.92) & 7.0 (0.0) \\
proposed ($\delta$ = 0.01)  & 2.5 (2.42) & 3.0 (0.63) & 3.1 (0.54) & 4.6 (0.92) & 7.1 (0.3) \\
proposed ($\delta$ = 0.001)  & 2.5 (2.42) & 3.0 (0.63) & 3.1 (0.54) & 4.6 (0.92) & 7.1 (0.3) \\
\midrule
MOM ($\delta$ = 0.8)  & 3.3 (2.57) & 1.9 (0.83) & 2.9 (0.54) & 3.2 (0.4) & 7.0 (0.0) \\
MOM ($\delta$ = 0.5)  & 3.4 (2.37) & 2.5 (1.02) & 3.0 (0.45) & 3.3 (0.46) & 7.0 (0.0) \\
MOM ($\delta$ = 0.05)  & 4.0 (4.0) & 8.0 (0.0) & 5.5 (2.5) & 5.0 (0.89) & 7.0 (0.0) \\
MOM ($\delta$ = 0.01)  & 0.0 (0.0) & 8.0 (0.0) & 8.0 (0.0) & 8.0 (0.0) & 7.3 (0.46) \\
MOM ($\delta$ = 0.001)  & 0.0 (0.0) & 0.0 (0.0) & 0.0 (0.0) & 8.0 (0.0) & 8.0 (0.0) \\
     \bottomrule
    \end{tabular}
 \end{table*}

 \section{Evaluation on real data} \label{sec:evaluation_real_regression}

In this section, we compare the results of our proposed and all baselines on two real data sets: crime data \citep{raftery1997bayesian,liang2008mixtures} and ozone data  \citep{garcia2013sampling}.
Details of all datasets, preprocessing, and additional results on GDP growth data (SDM) \citep{sala2004determinants} are given in Appendix D. 

We also show the results for AIC, the extended Bayesian information criterion (EBIC), stability selection, EMVS and the robust objective prior \citep{bayarri2012criteria}, which we denote as ``GibbsBvs".
Note that EBIC with $\gamma = 0$, is equal to ordinary BIC \citep{schwarz1978estimating}.
For the experiments with the real data we use 100000 MCMC-samples for the proposed method, GibbsBvs, and the horseshoe prior.\footnote{Out of which 10\% are used for burn in.}

%  whereas EMVS tends to select models with only very few variables.
% In particular, EMVS suggests that in the ozone and SDM data, none of the variables are relevant, which is quite a strong statement that is contradicting the results from all other methods.
% The stability selection method appears to select too few variables, independent of the setting of $q$ (hyper-parameter for controlling the e number of selected variables). 
% The results for EBIC highlight the sensitivity to the hyper-parameter $\gamma$. Note that EBIC with $\gamma = 0$, is equal to ordinary BIC \citep{schwarz1978estimating}.

The results for the ozone and crime data are shown in Tables \ref{tab:realData_ozone} and \ref{tab:realData_crime}, respectively.
We see that the horseshoe method performs similar as in the simulated data, tending to select models with relatively many variables.
For ozone, our model suggests that the model using x6.x7, x6.x8, x7.x7, and x6.x6, have relatively high regression coefficients, but not all of them are together in one model, possibly due to high correlation. 
For crime, our model suggests that all variables should be considered as relevant, whereas in particular M, Ed, Po1, Ineq have high regression coefficients.
% This is similar to the results of SSLASSO and GibbsBvs.

To further analyze the results of our proposed method, we inspect the top 10 model probabilities and variable inclusion probabilities calculated for $\delta =0$ and $\delta = 0.5$.
The model probabilities for ozone and crime  are shown in Tables \ref{tab:rankingModels_ozone} and \ref{tab:rankingModels_crime}, respectively.
Considering the low model probabilities, it is clear that there is no clearly winning model, and that care is needed when drawing conclusions from only the top model. 

In order to investigate the importance of each individual variable, we also show the variable inclusion probabilities for ozone and crime in Tables \ref{tab:rankingVars_ozone} and \ref{tab:rankingVars_crime}, respectively. In each of the Tables, we also show the results that were reported in previous studies. 
From the difference in the probabilities between previous studies, $\delta =0$ and $\delta = 0.5$,  we can draw some interesting conclusions.

\paragraph{Ozone data}
In Table \ref{tab:rankingVars_ozone}, we show the inclusion probabilities of the proposed method together with the results reported in \citep{garcia2013sampling}.
Comparing those results to the result of the proposed method, we find that the discrepancy between the results is not large, except in two cases.
First, the importance of the variable x9, including its interaction terms, is much higher in \citep{garcia2013sampling}.
Second, the squared term x7.x7 is considered as relevant by the proposed method, even when $\delta = 0.5$, which is in contrast to \citep{garcia2013sampling}, where an inclusion probability of only $45\%$ is reported.
Comparing the proposed method between $\delta = 0.0$ and $\delta = 0.5$, we see that the interaction variable x6.x8 is the most likely to be included for $\delta = 0.0$, with probability around 70\%. 
However, looking at the result with $\delta =0.5$, the effect size of x7.x7 is likely to be larger than x6.x8.

 \paragraph{Crime data}
In Table \ref{tab:rankingVars_crime}, we show the inclusion probabilities of the proposed method together with the results reported in \citep{liang2008mixtures}.
For the proposed method with $\delta =0$, we see good agreement with the results in \citep{liang2008mixtures}. 
This is in particular true with respect to the median probability model that includes all variables with probability larger than or equal to 0.5.
However, inspecting the inclusion probabilities for $\delta = 0.5$, there is not enough evidence that the effect size of Po2 and U2 is high.
 
 % updated
   \begin{table*}[h]
 \center
  \caption{Selected variables for the ozone data. For the proposed method and horseshoe method, we denote by "MSE inc" the expected increase in mean squared error compared to choosing the full model.}
  \label{tab:realData_ozone}
  \begin{tabular}{ll}
  \toprule 
 method & selected variables \\
\midrule
  \footnotesize proposed ($\delta$ =  0.8, MSE inc = 37.31\%) & \footnotesize x7.x7 \\
  \footnotesize proposed ($\delta$  =  0.5, MSE inc = 19.5\%) & \footnotesize x6.x6, x6.x7 \\
  \footnotesize proposed ($\delta$  =  0.05, MSE inc = 5.43\%) & \footnotesize x6.x7, x6.x8, x7.x7 \\
  \footnotesize proposed ($\delta$  =  0.01, MSE inc = 4.91\%) & \footnotesize x6.x6, x6.x7, x6.x8 \\
  \footnotesize proposed ($\delta$  =  0.001, MSE inc =  4.94\%) & \footnotesize  x6.x6, x6.x7, x6.x8 \\
  \footnotesize proposed ($\delta$  =  0.0, MSE inc = 5.44\% ) & \footnotesize x6.x7, x6.x8, x7.x7 \\
  \midrule
   \footnotesize horseshoe ($\delta$  =  0.8, MSE inc = 15.47\%) & \footnotesize  x6.x7, x7.x7, x7.x10 \\ 
  \footnotesize horseshoe ($\delta$  =  0.5, MSE inc = 5.11\%) & \footnotesize  x6.x7, x6.x8, x7.x7, x7.x8, x7.x10 \\ 
  \footnotesize horseshoe ($\delta$  =  0.05, MSE inc = 0.0\%) & \footnotesize  all except x5, x4.x5, x4.x8, x5.x8, x6.x9, x6.x10, x8.x8, x8.x9, x9.x10 \\
    \footnotesize horseshoe ($\delta$  =  0.01, MSE inc = 0.0\%) & \footnotesize  all except x5.x8, x6.x9 \\ 
  \footnotesize horseshoe ($\delta$  =  0.001, MSE inc = 0.0\%) & \footnotesize  all except x6.x9 \\
  \footnotesize horseshoe ($\delta$  =  0.0, MSE inc = 0.0\%) & \footnotesize  all \\
    \midrule
    \footnotesize GibbsBvs & \footnotesize x6.x6, x6.x7, x6.x8  \\
\footnotesize EMVS & \footnotesize none \\
\footnotesize SSLASSO & \footnotesize x6.x7, x6.x8, x7.x7 \\
\footnotesize MOM ($\delta = 0.8$) & \footnotesize x6.x6, x6.x7, x6.x8 \\
\footnotesize MOM ($\delta = 0.5$) & \footnotesize  x6.x6, x6.x7, x6.x8 \\
\footnotesize MOM ($\delta = 0.05$) & \footnotesize all \\
\footnotesize AIC & \footnotesize x9, x4.x4, x6.x7, x6.x8, x7.x7, x7.x8, x7.x10, x8.x10, x9.x9  \\
\footnotesize EBIC ($\gamma = 0$) & \footnotesize x6.x7, x6.x8, x7.x7, x7.x8, x7.x10, x8.x10, x9.x9 \\
\footnotesize EBIC ($\gamma = 0.5$) & \footnotesize x6.x7, x7.x7, x7.x8, x7.x10, x9.x9 \\
\footnotesize EBIC ($\gamma = 1.0$) & \footnotesize x4.x8, x6.x7, x7.x7 \\
\footnotesize stability ($q = 0.1\cdot d$) & \footnotesize x7.x7 \\
\footnotesize stability ($q = 0.5\cdot d$) & \footnotesize x7.x10 \\
\footnotesize stability ($q = 0.8\cdot d$) & \footnotesize none \\
       \bottomrule
    \end{tabular}
 \end{table*}
 
  % updated
   \begin{table*}[h]
 \center
  \caption{Selected variables for the crime data. For the proposed method and horseshoe method, we denote by "MSE inc" the expected increase in mean squared error compared to choosing the full model.}
  \label{tab:realData_crime}
  \begin{tabular}{ll}
  \toprule 
 method & selected variables \\
\midrule
  \footnotesize proposed ($\delta$  =  0.8, MSE inc = 65.62\%) & \footnotesize  Po1, Ineq \\
  \footnotesize proposed ($\delta$  =  0.5, MSE inc = 21.97\%) & \footnotesize  M, Ed, Po1, Ineq \\
  \footnotesize proposed ($\delta$  =  0.05, MSE inc = 0.0\%) & \footnotesize  all \\
  \footnotesize proposed ($\delta$  =  0.01, MSE inc = 0.0\%) & \footnotesize all \\
  \footnotesize proposed ($\delta$  =  0.001, MSE inc = 0.0\%) & \footnotesize all  \\
  \footnotesize proposed ($\delta$  =  0.0, MSE inc = 0.0\%) & \footnotesize all \\
  \midrule
  \footnotesize horseshoe ($\delta$  =  0.8, MSE inc = 17.23\%) & \footnotesize M, Ed, Po1, Po2, NW, Ineq, Prob \\ 
  \footnotesize horseshoe ($\delta$  =  0.5, MSE inc = 3.07\%) & \footnotesize all except So, LF, M.F, Pop, U1, Time \\
  \footnotesize horseshoe ($\delta$  =  0.05, MSE inc = 0.0\%) & \footnotesize all except M.F \\
  \footnotesize horseshoe ($\delta$  =  0.01, MSE inc = 0.0\%) & \footnotesize all except M.F \\
  \footnotesize horseshoe ($\delta$  =  0.001, MSE inc = 0.0\%) & \footnotesize  all \\
  \footnotesize horseshoe ($\delta$  =  0.0, MSE inc = 0.0\%) & \footnotesize all \\
  \midrule
  \footnotesize GibbsBvs & \footnotesize all  \\
\footnotesize EMVS & \footnotesize Po1, Ineq \\
\footnotesize SSLASSO & \footnotesize Ed, Po1, NW, Ineq \\
\footnotesize MOM ($\delta = 0.8$) & \footnotesize Po1, Ineq \\
\footnotesize MOM ($\delta = 0.5$) & \footnotesize  Po1, Ineq \\
\footnotesize MOM ($\delta = 0.05$) & \footnotesize all \\
\footnotesize AIC & \footnotesize all except So, Po2, M.F, U1 \\
\footnotesize EBIC ($\gamma = 0$) & \footnotesize all except So, Po2, LF, Pop, U1, GDP, Time \\
\footnotesize EBIC ($\gamma = 0.5$) & \footnotesize M, Ed, Po1, M.F, NW, Ineq, Prob  \\
\footnotesize EBIC ($\gamma = 1.0$) & \footnotesize Po1, NW \\
\footnotesize stability ($q = 0.1\cdot d$) & \footnotesize Po1 \\
\footnotesize stability ($q = 0.5\cdot d$) & \footnotesize NW \\
\footnotesize stability ($q = 0.8\cdot d$) & \footnotesize none \\
       \bottomrule
    \end{tabular}
 \end{table*}

   \begin{table*}[h]
 \center
  \caption{Top 10 selected models using the proposed method with $\delta = 0.0$ and $\delta = 0.5$ for the ozone data. Last column also shows the highest posterior probability model reported in \citep{garcia2013sampling} using a $g$-prior where inclusion probabilities are calculated exactly (i.e. no MCMC).}
  \label{tab:rankingModels_ozone}
  \begin{tabular}{ll}
  \toprule 
         model & probability \\
\midrule
\multicolumn{2}{c}{$\delta = 0.5$} \\
\midrule
x6.x6, x6.x7 & 0.066 \\
x6.x7, x7.x7 & 0.034 \\
x4.x10, x7.x7, x7.x10 & 0.027 \\
x7.x7 & 0.026 \\
x10, x4.x7, x7.x10 & 0.02 \\
x4.x7, x4.x10, x7.x10 & 0.019 \\
x10, x7.x7, x7.x10 & 0.019 \\
x7, x6.x7, x7.x7 & 0.016 \\
x7.x7, x7.x10 & 0.015 \\
x6.x6, x6.x7, x7.x8 & 0.013 \\
\midrule
      \multicolumn{2}{c}{$\delta = 0.0$} \\
\midrule
x6.x7, x6.x8, x7.x7 & 0.031 \\
x6.x6, x6.x7, x6.x8 & 0.029 \\
x10, x6.x7, x6.x8, x7.x7, x7.x10 & 0.018 \\
x4.x10, x6.x7, x6.x8, x7.x7, x7.x10 & 0.018 \\
x4.x6, x4.x10, x6.x8, x7.x7, x7.x10 & 0.016 \\
x10, x4.x6, x6.x8, x7.x7, x7.x10 & 0.013 \\
x6.x7, x7.x7, x7.x8 & 0.01 \\
x6, x4.x10, x6.x8, x7.x7, x7.x10 & 0.01 \\
x4.x6, x6.x8, x7.x7 & 0.009 \\
x6, x6.x8, x7.x7 & 0.009 \\
\midrule
      \multicolumn{2}{c}{Gracia-Donato} \\
\midrule
x10, x4.x6, x6.x8, x7.x7, x7.x10 & 0.0009 \\
       \bottomrule
    \end{tabular}
 \end{table*}

\begin{table*}[h]
 \center
  \caption{All variable inclusion probabilities using the proposed method with $\delta = 0.0$ and $\delta = 0.5$ for the ozone data. Last column also shows the results reported in \citep{garcia2013sampling} using a $g$-prior where inclusion probabilities are calculated exactly (i.e. no MCMC).}
  % \footnote{From the description in \citep{garcia2013sampling} it is unclear how $g$ is determined.}
  \label{tab:rankingVars_ozone}
  \begin{tabular}{llll}
  \toprule 
variable & $\delta = 0.5$  & $\delta = 0.0$ & Gracia-Donato \\
\midrule
 x7.x7 & 0.58 & 0.67 & 0.450 \\ % <- interesting
x6.x7 & 0.568 & 0.603 & 0.636 \\
x7.x10 & 0.5 & 0.649 & 0.743 \\
x6.x6 & 0.313 & 0.245 & 0.532 \\
x4.x10 & 0.233 & 0.334 &  0.361 \\
x6.x8 & 0.226 & 0.702 & 0.560 \\  % <- interesting
x10 & 0.226 & 0.291 & 0.368 \\
x4.x7 & 0.212 & 0.234 & 0.252 \\
x7.x8 & 0.179 & 0.279 &  0.349 \\
x4.x6 & 0.164 & 0.295 & 0.325 \\
x6 & 0.139 & 0.246 &  0.297 \\
x7 & 0.133 & 0.16 &  0.195 \\
x7.x9 & 0.09 & 0.072 &  0.431 \\
x8 & 0.076 & 0.139 & 0.200 \\
x4.x9 & 0.064 & 0.059 &  0.301 \\
x4.x8 & 0.064 & 0.132 &  0.208 \\
x9.x9 & 0.059 & 0.156 & 0.434 \\ % <- interesting
x9 & 0.053 & 0.056 &  0.291 \\
x8.x10 & 0.037 & 0.112 & 0.236 \\
x10.x10 & 0.028 & 0.07 & 0.117 \\
x8.x8 & 0.028 & 0.067 & 0.142 \\
x8.x9 & 0.019 & 0.034 & 0.263 \\
x5.x10 & 0.019 & 0.036 &  0.124 \\
x6.x10 & 0.017 & 0.052 &  0.115 \\
x6.x9 & 0.012 & 0.036 & 0.126 \\
x4.x4 & 0.011 & 0.032 & 0.164 \\
x5.x6 & 0.011 & 0.027 &  0.107 \\
x4 & 0.011 & 0.031 &  0.164 \\
x5.x8 & 0.009 & 0.031 &  0.098 \\
x5.x5 & 0.008 & 0.024 &  0.124 \\
x5.x7 & 0.008 & 0.025 & 0.094 \\
x9.x10 & 0.007 & 0.024 & 0.103 \\
x5 & 0.006 & 0.019 &  0.096 \\
x4.x5 & 0.006 & 0.02 & 0.095 \\
x5.x9 & 0.005 & 0.022 &  0.088 \\
       \bottomrule
    \end{tabular}
 \end{table*}

   \begin{table*}[h]
 \center
  \caption{Top 10 selected models using the proposed method with $\delta = 0.0$ and $\delta = 0.5$ for the crime data. }
  \label{tab:rankingModels_crime}
  \begin{tabular}{ll}
  \toprule 
         model & probability \\
\midrule
\multicolumn{2}{c}{$\delta = 0.5$} \\
\midrule
M, Ed, Po1, Ineq & 0.021 \\
M, Ed, Po1, NW, Ineq, Prob & 0.017 \\
Po1, Ineq & 0.017 \\
Ed, Po1, Ineq & 0.016 \\
M, Ed, Po1, NW, U2, Ineq, Prob & 0.015 \\
M, Ed, Po1, Ineq, Prob & 0.015 \\
Ed, Po1, NW, Ineq, Prob & 0.013 \\
M, Ed, Po1, U2, Ineq & 0.011 \\
M, Ed, Po1, U2, Ineq, Prob & 0.011 \\
M, Ed, Po1, NW, Ineq, Prob, Time & 0.011 \\
\midrule
      \multicolumn{2}{c}{$\delta = 0.0$} \\
\midrule
all  & 0.02 \\
M, Ed, Po1, Ineq & 0.01 \\
M, Ed, Po1, NW, U2, Ineq, Prob & 0.01 \\
Ed, Po1, Ineq & 0.008 \\
M, Ed, Po1, NW, Ineq, Prob & 0.008 \\
all except So, Po2, LF, M.F, Pop, U1, GDP & 0.008 \\
Po1, Ineq & 0.007 \\
all except So, LF, M.F, Pop, U1, GDP, Time & 0.007 \\
M, Ed, Po1, NW, Ineq, Prob, Time & 0.007 \\
M, Ed, Po1, U2, Ineq, Prob & 0.007 \\
       \bottomrule
    \end{tabular}
 \end{table*}

    \begin{table*}[h]
 \center
  \caption{All variable inclusion probabilities using the proposed method with $\delta = 0.5$ and $\delta = 0.0$ for the crime data. Last column also shows the results reported in \citep{liang2008mixtures} with the Zellner-Siow Prior with the null-model as the reference model.}
  \label{tab:rankingVars_crime}
  \begin{tabular}{llll}
  \toprule 
%          variable & probability \\
% \midrule
variable & $\delta = 0.5$  & $\delta = 0.0$ & Liang \\
\midrule
Ineq & 0.993 & 0.995 & 1.0 \\
Ed & 0.906 & 0.943 & 0.97 \\
Prob & 0.758 & 0.833 & 0.90 \\
Po1 & 0.742 & 0.792 & 0.67 \\
M & 0.731 & 0.808 & 0.85 \\
NW & 0.604 & 0.711 & 0.69 \\
Po2 & 0.52 & 0.591 & 0.45 \\
U2 & 0.425 & 0.557 & 0.61 \\
GDP & 0.381 & 0.481 & 0.36 \\
Time & 0.256 & 0.395 & 0.37 \\
Pop & 0.244 & 0.368 & 0.37 \\
So & 0.207 & 0.32 & 0.27 \\
U1 & 0.134 & 0.269 & 0.25 \\
M.F & 0.12 & 0.238 & 0.20 \\
LF & 0.115 & 0.233 & 0.20 \\
       \bottomrule
    \end{tabular}
 \end{table*}

 \section{Conclusions} \label{sec:conclusions_regression}
 
 We proposed a new type of spike-and-slab prior that is particularly well suited for the situation where there are small negligible, but non-zero regression coefficients (quasi-sparseness). 
 These small negligible regression coefficients are considered as noise, since they can lead to the selection of overly complex models (i.e. models with many variables), although, only few variables should be considered as practically relevant.
The proposed method uses disjunct support priors on the regression coefficients with a threshold  parameter $\delta > 0$ in order to ignore small coefficients. 
We showed that in the quasi-sparse setting, the proposed method leads to consistent Bayes factors, which is not the case for full support priors as originally proposed in \citep{chipman2001practical}, and the moment matching priors (MOM) \citep{johnson2010use,johnson2012bayesian,rossell2017nonlocal}.
  
 Due to the non-conjugacy of the priors proposed by our method, estimating the marginal likelihood explicitly is computationally infeasible. 
 We therefore introduced a latent variable indicator vector $\mathbf{z}$, and proposed an efficient Gibbs sampler to sample from its posterior distribution.
 % This allows us to estimate all model probabilities $p(S  | \mathbf{y}, X, \delta)$, where $S$ is a set of relevant variables and $\delta$ is a threshold parameter specifying practical relevance (effect size). 
 
For synthetic data with ground truth, we showed that the proposed method leads to good model selection performance in various settings: with/without noise and low/high dimensions.
For real data, we showed that by inspecting the model and variable inclusion probabilities for different threshold values $\delta$, we can draw interesting conclusions about the effect size (the absolute magnitude) of regression coefficients.
Together with an estimate of the mean squared error (MSE) of the final model, this allows for a trade-off between sparsity and prediction accuracy, similar to the practical advise given in \citep{hahn2015decoupling}.

\section*{Appendix A: Asymptotic Results}
\begin{lemma}  \label{lemma:uniqueMaximum}
The function
\begin{align*}
\boldsymbol{\theta} \rightarrow \E_{\mathbf{x}} \Big[ g(\boldsymbol{\theta}) \Big]  \, ,
\end{align*}
has a unique maximum, where 
\begin{align*}
g(\boldsymbol{\theta}) := \E_{y \sim p(y | \boldsymbol{\theta}, \mathbf{x})} \Big[  \log p(y | \boldsymbol{\theta}, \mathbf{x}) \Big] \, .
\end{align*}
and $\mathbf{x}$ is distributed according to some non-degenerated distribution with mean zero and positive definite covariance matrix $C$.
\end{lemma}

\begin{proof}
First of all, let us do a change of variable using the one-to-one mapping $\tau := \sigma_r^{-2}$.
For simplicity, let us denote $\boldsymbol{\theta} := (\tau,  \boldsymbol{\beta})$, and the true parameter vector as $\boldsymbol{\theta}_t$.
We have 
\begin{align*}
\log p(y | \boldsymbol{\theta}, \mathbf{x}) = \frac{1}{2} \log \tau - \frac{1}{2} \tau (y - \mathbf{x}^T \boldsymbol{\beta})^2 - \frac{1}{2} \log 2\pi \,, 
\end{align*}
and 
\begin{align*}
\E_{\mathbf{x}} \Big[ g(\boldsymbol{\theta}) \Big] = \E_{\mathbf{x}, y} \Big[ \log p(y | \boldsymbol{\theta}, \mathbf{x})  \Big] 
= \frac{1}{2} \log \tau - \frac{1}{2} \tau \E_{\mathbf{x}, y} \Big[  (y - \mathbf{x}^T \boldsymbol{\beta})^2 \Big] - \frac{1}{2} \log 2\pi \,. 
\end{align*}
Since $C$ is positive definite, we have that $\E_{\mathbf{x}, y} \Big[  (y - \mathbf{x}^T \boldsymbol{\beta})^2 \Big]$ has a unique minimum at $\boldsymbol{\beta} = \boldsymbol{\beta}_t$.
To see this note that 
\begin{align*}
\E_{\mathbf{x}, y} \Big[  (y - \mathbf{x}^T \boldsymbol{\beta})^2 \Big] 
&= \E_{\mathbf{x}, \epsilon} \Big[  (\mathbf{x}^T \boldsymbol{\beta}_t + \epsilon - \mathbf{x}^T \boldsymbol{\beta})^2  \Big] \\
&= \E_{\mathbf{x}, \epsilon} \Big[ \big( \mathbf{x}^T (\boldsymbol{\beta}_t - \boldsymbol{\beta}) + \epsilon \big)^2  \Big] \\
&= \E_{\mathbf{x}} \Big[ \big( \mathbf{x}^T (\boldsymbol{\beta}_t - \boldsymbol{\beta}) \big)^2 \Big] + 2 \E_{\mathbf{x}, \epsilon} \Big[\epsilon \cdot  \mathbf{x}^T (\boldsymbol{\beta}_t - \boldsymbol{\beta})  \Big] + E_{\epsilon} \Big[  \epsilon ^2  \Big] \\
&= (\boldsymbol{\beta}_t - \boldsymbol{\beta})^T \E_{\mathbf{x}} \Big[ \mathbf{x} \mathbf{x}^T \Big]  (\boldsymbol{\beta}_t - \boldsymbol{\beta}) + 2 \E_{\epsilon} \Big[\epsilon \Big] \cdot  \E_{\mathbf{x}} \Big[ \mathbf{x}^T \Big](\boldsymbol{\beta}_t - \boldsymbol{\beta})   + E_{\epsilon} \Big[  \epsilon ^2  \Big] \\
&= (\boldsymbol{\beta}_t - \boldsymbol{\beta})^T C (\boldsymbol{\beta}_t - \boldsymbol{\beta}) + \sigma_{r, t}^2 \, .
\end{align*}
where we used that $\E_{\mathbf{x}} \Big[ \mathbf{x} \Big] = \mathbf{0}$, and $C = \E_{\mathbf{x}} \Big[ \mathbf{x} \mathbf{x}^T \Big]$. 
For $\boldsymbol{\beta} = \boldsymbol{\beta}_t$, we have $\E_{\mathbf{x}, y} \Big[  (y - \mathbf{x}^T \boldsymbol{\beta})^2 \Big] = \frac{1}{\tau_t}$.
Furthermore, since 
\begin{align*}
\E_{\mathbf{x}} \Big[ g(\tau, \boldsymbol{\beta}_t) \Big] = \frac{1}{2} \log \tau - \frac{1}{2} \tau \frac{1}{\tau_t} - \frac{1}{2} \log 2\pi 
\end{align*} 
is strictly concave with respect to $\tau$, with unique maximum $\tau_r$, 
we have that the unique maximum of $\E_{\mathbf{x}} \Big[ g(\boldsymbol{\theta}) \Big]$ is given by $(\tau_t, \boldsymbol{\beta}_t)$. 
\end{proof}

\section*{Appendix B: Slice Sampler}

First let us introduce the auxiliary random variable $U$, and the following joint distribution:
\begin{align*}
p(U, \sigma^2_1) = 
\begin{cases}
\frac{1}{L} \cdot  \text{Inv-}\chi^2(\sigma^2_1 |  \tilde{\nu} , \tilde{\eta}^2) & \text{if } 0 <  U < h(\sigma^2_1),\\
0 & \text{else. }
\end{cases} \, ,
 \end{align*} 
 where $L$ is an appropriate normalization constant.
 We then have that 
  \begin{align*}
p(\sigma^2_1) &= \int_0^{h(\sigma^2_1)} p(u, \sigma^2_1) du  \\ 
&= \frac{1}{L} \cdot  \text{Inv-}\chi^2(\sigma^2_1 |   \tilde{\nu} , \tilde{\eta}^2) \int_0^{h(\sigma^2_1)} 1 du \\
&= \frac{1}{L} \cdot  \text{Inv-}\chi^2(\sigma^2_1 |   \tilde{\nu} , \tilde{\eta}^2) \big[ u \big]_0^{h(\sigma^2_1)} \\
&= \frac{1}{L} \cdot  h(\sigma^2_1) \cdot \text{Inv-}\chi^2(\sigma^2_1 |   \tilde{\nu} , \tilde{\eta}^2)  \\
&\propto h(\sigma^2_1) \cdot \text{Inv-}\chi^2(\sigma^2_1 |   \tilde{\nu} , \tilde{\eta}^2)  \\
 \end{align*}
 
 In order to sample from the joint distribution $p(U, \sigma^2_1)$, we employ a Gibbs sampler, where
 \begin{align*}
p(U | \sigma^2_1) = \text{Uniform}( [0, h(\sigma^2_1)] )
 \end{align*}
 
 and 
 
 \begin{align*}
p(\sigma^2_1 | u) = 
\begin{cases}
\frac{1}{\tilde{L}} \cdot  \text{Inv-}\chi^2(\sigma^2_1 |  \tilde{\nu} , \tilde{\eta}^2) & \text{if }  h(\sigma^2_1) > u,\\
0 & \text{else. }
\end{cases} \, ,
 \end{align*} 
 for an appropriate normalization constant $\tilde{L}$.

\section*{Appendix C: Additional results synthetic data}

We show the results for $\delta \in \{0.8, 0.5, 0.05, 0.01, 0.001, 0.0\}$, and the results of the most parsimonious model that is estimated to lead to an increase in MSE of not more than 5\%.
For all methods based on MCMC we use 10000 samples, out of which 10\% are used for burn in.

As our first baseline, we use the robust objective prior proposed in \citep{bayarri2012criteria} together with a Gibbs sampler to explore the space of models, which we denote as ``GibbsBvs".\footnote{Implemented in the R package 'BayesVarSel'. As suggested by the authors, we use the g-Zellner prior \citep{zellner1986assessing} in cases where the robust prior from \citep{bayarri2012criteria} fails.}
Furthermore, we use the Gaussian and Laplace spike-and-slab priors combined with EM-algorithm as proposed in \citep{rovckova2014emvs,rovckova2018spike} which we denote as "EMVS" and "SSLASSO", respectively.\footnote{Implemented in the R package 'EMVS' and 'SSLASSO'.}
Note that EMVS and SSLASSO do not provide model or variable inclusion posterior probabilities.

Finally, we include also three frequentist methods for model search.
As a first frequentist method, we use the popular Least Angle Regression (LARS) method \citep{efron2004least} to get a set of candidate models.
We then select the model using the Extended Bayesian information criterion (EBIC) with $\gamma \in \{0, 0.5, 1\}$ \citep{chen2008extended,foygel2010extended}, or the Akaike information criterion (AIC) \citep{akaike1973information}.
Note that EBIC with $\gamma = 0$, is equal to the Bayesian information criterion (BIC) \citep{schwarz1978estimating}.
As a third frequentist method, we use linear regression with Lasso \citep{tibshirani1996regression} combined with stability selection \citep{meinshausen2010stability}.
Stability selection has two hyper-parameters that need to be specified: the "upper bound for the per-family error rate" (PFER) and "the number of (unique) selected variables" (denoted by $q$) 
as in the R package 'stabs'. For PFER we set always 1. However, we found that stability selection can be sensitive to the choice of $q$, and therefore show all results for three different values. 

We evaluate all methods in terms of F1-Score. All experiments are repeated 10 times and we report average and standard deviations (shown in brackets).
For large $n$, GibbsBvs did not execute correctly, which we mark as "-". For the high-dimensional setting GibbsBvs did not finish computation due to high memory requirements.
% For the proposed method and the horseshoe prior method as described in the previous paragraph, we select the threshold value $\delta$, as described in Section \ref{sec:estimatingExpectedIncreaseInMSE}, and for the horseshoe prior method as described in the previous paragraph. 
When we selected the threshold value $\delta$ automatically by using the estimated increase in MSE, we mark this in all Tables by "$^*$".
If not reported otherwise, we use for all baselines the default settings.

\paragraph{Low-dimensional setting}
The results for the low-dimensional setting, with and without noise, are shown in Tables \ref{tab:simDataNoise0_lowDimensional}, \ref{tab:simDataNoise02_lowDimensional} and \ref{tab:simDataNoise05_lowDimensional}.
Overall, we see that the proposed method and the horseshoe prior method perform best. 

GibbsBvs, SSLASSO, EBIC and Stability selection (with $q \geq 4$) perform good for no noise or small noise. However, for $\eta = 0.5$, GibbsBvs, SSLASSO, EBIC and Stability Selection start to select more irrelevant variables with increasing sample size $n$.
Asymptotically, all four methods are expected to select all variables with coefficient regressions $\beta_j \neq 0$, no matter how small $\beta_j$ is. 
However, if the sample size is small ($n \leq 100$), then all three methods are not influenced by the noise, i.e. they ignore the negligible small regression coefficients.

AIC performs similar to EBIC for $n \leq 100$, but for larger sample sizes it tends to select too many variables, even in the no-noise setting.
This is not too surprising, since it is well known that AIC is not model selection consistent (see e.g. \citep{yang2005can}).

Interestingly, in the noise setting ($\eta = 0.2$ and $\eta = 0.5$), even for large $n$, EMVS find the correct relevant variables. 
However, for small sample sizes EMVS tends to select too few variables. This suggests that EMVS has a strong inductive bias for sparse models, which can be helpful in the noise setting,
but is deteriorating performance for small to medium-sized $n$.

\paragraph{High-dimensional setting}
The results for the high-dimensional setting, with and without noise, are shown in Tables \ref{tab:simDataNoise0_highDimensional}, and \ref{tab:simDataNoise02_highDimensional}.
Overall, we see that the proposed method, SSLASSO, Stability selection (with $q \geq 50$)  and EMVS perform best. In this setting the EMVS seems to profit from its inductive bias for sparse models.
On the other hand, the horseshoe prior method performs somehow unsatisfactory, tending to select too many variables.
AIC and EBIC performed very poorly in this setting, selecting too many variables. One reason seems to stem from the numerical instability of the maximum likelihood estimate for $d \leq n$. As an ad-hoc remedy we tried to combine it with a ridge estimate, but this did not seem to help.

\paragraph{Analysis of different $\delta$}
In Tables \ref{tab:simDataNoise0_lowDimensional_deltaAnalysis}, \ref{tab:simDataNoise02_lowDimensional_deltaAnalysis}, and \ref{tab:simDataNoise05_lowDimensional_deltaAnalysis}, we show the results for different fixed $\delta$ in the low-dimensional setting, and in Tables \ref{tab:simDataNoise0_highDimensional_deltaAnalysis} and \ref{tab:simDataNoise02_highDimensional_deltaAnalysis} for the high-dimensional setting.
The proposed method is less sensitive to the choice of $\delta$ and tends to select sparse models even in the high-dimensional setting. However, as expected, the horseshoe prior method is highly sensitive to the choice of $\delta$.
 
 \paragraph{Comparison to moment matching priors}
In Tables \ref{tab:simDataNoise0_highDimensional_deltaAnalysisVS_MOM} and \ref{tab:simDataNoise02_highDimensional_deltaAnalysisVS_MOM}, we show the comparison of the proposed method and the product moment matching priors (MOM) from \citep{johnson2012bayesian} for the high-dimensional settings. Similar to the low-dimensional setting, we observe that for small values of the practical relevance threshold $\delta$, MOM tends to select too few or too many variables.

% fully updated
  \begin{table*}[h]
 \center
  \caption{Low-dimensional setting, $d = 8$ and $n \in \{10, 50, 100, 1000, 100000\}$. Evaluation results with no noise on regression coefficients.}
  \label{tab:simDataNoise0_lowDimensional}
  \begin{tabular}{llllll}
  \toprule 
      \multicolumn{6}{c}{F1-Scores} \\
      \midrule
 & 10 & 50 & 100 & 1000 & 100000 \\
\midrule
proposed$^*$  & 0.48 (0.18) & 0.88 (0.1) & 1.0 (0.0) & 1.0 (0.0) & 1.0 (0.0) \\
horseshoe$^*$  & 0.59 (0.1) & 0.93 (0.14) & 0.95 (0.09) & 1.0 (0.0) & 1.0 (0.0) \\
GibbsBvs & 0.48 (0.16) & 0.89 (0.09) & 0.95 (0.09) & 1.0 (0.0) & - \\
EMVS & 0.25 (0.25) & 0.13 (0.27) & 0.0 (0.0) & 0.96 (0.08) & 1.0 (0.0) \\
SSLASSO & 0.41 (0.29) & 0.88 (0.1) & 1.0 (0.0) & 1.0 (0.0) & 1.0 (0.0) \\
AIC & 0.49 (0.07) & 0.89 (0.1) & 0.88 (0.13) & 0.88 (0.09) & 0.85 (0.15) \\
EBIC ($\gamma = 0.0$)  & 0.49 (0.07) & 0.92 (0.09) & 0.96 (0.08) & 1.0 (0.0) & 1.0 (0.0) \\
EBIC ($\gamma = 0.5$)  & 0.53 (0.08) & 0.91 (0.1) & 1.0 (0.0) & 1.0 (0.0) & 1.0 (0.0) \\
EBIC ($\gamma = 1.0$)  & 0.47 (0.16) & 0.9 (0.1) & 1.0 (0.0) & 1.0 (0.0) & 1.0 (0.0) \\
stability ($q = 1$)  & 0.2 (0.24) & 0.3 (0.24) & 0.35 (0.23) & 0.5 (0.0) & 0.5 (0.0) \\
stability ($q = 4$)  & 0.1 (0.2) & 0.9 (0.1) & 0.98 (0.06) & 1.0 (0.0) & 1.0 (0.0) \\
stability ($q = 6$)  & 0.05 (0.15) & 0.91 (0.1) & 0.99 (0.04) & 1.0 (0.0) & 1.0 (0.0) \\
\midrule
      \multicolumn{6}{c}{Average number of selected variables} \\
          \midrule
 & 10 & 50 & 100 & 1000 & 100000 \\
\midrule
proposed$^*$  & 2.6 (2.58) & 2.4 (0.49) & 3.0 (0.0) & 3.0 (0.0) & 3.0 (0.0) \\
horseshoe$^*$  & 4.6 (2.2) & 3.7 (1.49) & 3.4 (0.66) & 3.0 (0.0) & 3.0 (0.0) \\
GibbsBvs & 5.1 (3.56) & 2.6 (0.66) & 3.4 (0.66) & 3.0 (0.0) & - \\
EMVS & 0.5 (0.5) & 0.3 (0.64) & 0.0 (0.0) & 2.8 (0.4) & 3.0 (0.0) \\
SSLASSO & 0.9 (0.7) & 2.4 (0.49) & 3.0 (0.0) & 3.0 (0.0) & 3.0 (0.0) \\
AIC & 6.2 (2.4) & 3.8 (0.75) & 4.0 (1.18) & 3.9 (0.7) & 4.3 (1.49) \\
EBIC ($\gamma = 0.0$)  & 6.2 (2.4) & 3.1 (0.7) & 3.3 (0.64) & 3.0 (0.0) & 3.0 (0.0) \\
EBIC ($\gamma = 0.5$)  & 5.2 (2.89) & 2.7 (0.64) & 3.0 (0.0) & 3.0 (0.0) & 3.0 (0.0) \\
EBIC ($\gamma = 1.0$)  & 3.8 (3.19) & 2.5 (0.5) & 3.0 (0.0) & 3.0 (0.0) & 3.0 (0.0) \\
stability ($q = 1$)  & 0.4 (0.49) & 0.6 (0.49) & 0.7 (0.46) & 1.0 (0.0) & 1.0 (0.0) \\
stability ($q = 4$)  & 0.2 (0.4) & 2.5 (0.5) & 2.9 (0.3) & 3.0 (0.0) & 3.0 (0.0) \\
stability ($q = 6$)  & 0.1 (0.3) & 2.7 (0.64) & 3.1 (0.3) & 3.0 (0.0) & 3.0 (0.0) \\
       \bottomrule
    \end{tabular}
 \end{table*}

 % fully updated
   \begin{table*}[h]
 \center
  \caption{Low-dimensional setting, $d = 8$ and $n \in \{10, 50, 100, 1000, 100000\}$. Evaluation results with noise on regression coefficients $\eta = 0.2$.}
  \label{tab:simDataNoise02_lowDimensional}
  \begin{tabular}{llllll}
  \toprule 
      \multicolumn{6}{c}{F1-Scores} \\
      \midrule
 & 10 & 50 & 100 & 1000 & 100000 \\
\midrule
proposed$^*$  & 0.41 (0.16) & 0.95 (0.08) & 0.97 (0.07) & 1.0 (0.0) & 1.0 (0.0) \\
horseshoe$^*$  & 0.6 (0.21) & 0.88 (0.13) & 0.99 (0.04) & 1.0 (0.0) & 1.0 (0.0) \\
GibbsBvs & 0.49 (0.09) & 0.93 (0.09) & 0.97 (0.06) & 0.99 (0.04) & - \\
EMVS & 0.21 (0.26) & 0.05 (0.15) & 0.0 (0.0) & 1.0 (0.0) & 1.0 (0.0) \\
SSLASSO & 0.28 (0.23) & 0.86 (0.21) & 0.97 (0.07) & 1.0 (0.0) & 0.63 (0.03) \\
AIC & 0.56 (0.14) & 0.87 (0.09) & 0.87 (0.11) & 0.81 (0.15) & 0.58 (0.02) \\
EBIC ($\gamma = 0.0$)  & 0.56 (0.13) & 0.92 (0.09) & 0.97 (0.06) & 1.0 (0.0) & 0.63 (0.03) \\
EBIC ($\gamma = 0.5$)  & 0.55 (0.14) & 0.91 (0.12) & 0.99 (0.04) & 1.0 (0.0) & 0.63 (0.03) \\
EBIC ($\gamma = 1.0$)  & 0.54 (0.14) & 0.91 (0.12) & 0.97 (0.07) & 1.0 (0.0) & 0.64 (0.03) \\
stability ($q = 1$)  & 0.1 (0.2) & 0.35 (0.23) & 0.5 (0.0) & 0.5 (0.0) & 0.5 (0.0) \\
stability ($q = 4$)  & 0.1 (0.2) & 0.89 (0.16) & 0.98 (0.06) & 1.0 (0.0) & 0.93 (0.07) \\
stability ($q = 6$)  & 0.0 (0.0) & 0.89 (0.2) & 1.0 (0.0) & 0.99 (0.04) & 0.67 (0.0) \\
\midrule
      \multicolumn{6}{c}{Average number of selected variables} \\
          \midrule
 & 10 & 50 & 100 & 1000 & 100000 \\
\midrule
proposed$^*$  & 2.5 (2.06) & 2.9 (0.54) & 3.0 (0.45) & 3.0 (0.0) & 3.0 (0.0) \\
horseshoe$^*$  & 4.8 (2.23) & 3.8 (1.54) & 3.1 (0.3) & 3.0 (0.0) & 3.0 (0.0) \\
GibbsBvs & 5.7 (2.9) & 3.0 (0.63) & 3.2 (0.4) & 3.1 (0.3) & - \\
EMVS & 1.1 (2.02) & 0.1 (0.3) & 0.0 (0.0) & 3.0 (0.0) & 3.0 (0.0) \\
SSLASSO & 1.1 (0.94) & 2.7 (0.78) & 3.0 (0.45) & 3.0 (0.0) & 6.6 (0.49) \\
AIC & 6.6 (2.15) & 3.7 (0.9) & 4.0 (1.1) & 4.7 (1.49) & 7.3 (0.46) \\
EBIC ($\gamma = 0.0$)  & 6.5 (2.2) & 3.1 (0.7) & 3.2 (0.4) & 3.0 (0.0) & 6.6 (0.49) \\
EBIC ($\gamma = 0.5$)  & 5.1 (2.7) & 2.9 (0.54) & 3.1 (0.3) & 3.0 (0.0) & 6.6 (0.49) \\
EBIC ($\gamma = 1.0$)  & 4.5 (2.87) & 2.9 (0.54) & 3.0 (0.45) & 3.0 (0.0) & 6.4 (0.49) \\
stability ($q = 1$)  & 0.3 (0.46) & 0.7 (0.46) & 1.0 (0.0) & 1.0 (0.0) & 1.0 (0.0) \\
stability ($q = 4$)  & 0.2 (0.4) & 2.5 (0.67) & 2.9 (0.3) & 3.0 (0.0) & 3.5 (0.5) \\
stability ($q = 6$)  & 0.0 (0.0) & 2.7 (0.9) & 3.0 (0.0) & 3.1 (0.3) & 6.0 (0.0) \\
       \bottomrule
    \end{tabular}
 \end{table*}

 % fully updated
    \begin{table*}[h]
 \center
  \caption{Low-dimensional setting, $d = 8$ and $n \in \{10, 50, 100, 1000, 100000\}$. Evaluation results with noise on regression coefficients $\eta = 0.5$.}
  \label{tab:simDataNoise05_lowDimensional}
  \begin{tabular}{llllll}
  \toprule 
      \multicolumn{6}{c}{F1-Scores} \\
      \midrule
 & 10 & 50 & 100 & 1000 & 100000 \\
\midrule
proposed$^*$  & 0.47 (0.18) & 0.93 (0.09) & 0.97 (0.07) & 1.0 (0.0) & 1.0 (0.0) \\
horseshoe$^*$  & 0.63 (0.19) & 0.88 (0.13) & 0.97 (0.06) & 1.0 (0.0) & 1.0 (0.0) \\
GibbsBvs & 0.49 (0.09) & 0.87 (0.14) & 0.96 (0.07) & 0.78 (0.09) & - \\
EMVS & 0.21 (0.26) & 0.05 (0.15) & 0.0 (0.0) & 0.98 (0.06) & 1.0 (0.0) \\
SSLASSO & 0.35 (0.24) & 0.9 (0.16) & 0.97 (0.07) & 0.88 (0.1) & 0.6 (0.0) \\
AIC & 0.56 (0.13) & 0.83 (0.09) & 0.9 (0.08) & 0.64 (0.08) & 0.55 (0.02) \\
EBIC ($\gamma = 0.0$)  & 0.56 (0.13) & 0.92 (0.09) & 0.96 (0.07) & 0.77 (0.14) & 0.6 (0.0) \\
EBIC ($\gamma = 0.5$)  & 0.56 (0.14) & 0.91 (0.12) & 0.95 (0.08) & 0.82 (0.14) & 0.6 (0.0) \\
EBIC ($\gamma = 1.0$)  & 0.55 (0.14) & 0.91 (0.12) & 0.97 (0.07) & 0.9 (0.1) & 0.6 (0.0) \\
stability ($q = 1$)  & 0.1 (0.2) & 0.35 (0.23) & 0.5 (0.0) & 0.5 (0.0) & 0.5 (0.0) \\
stability ($q = 4$)  & 0.0 (0.0) & 0.91 (0.16) & 0.98 (0.06) & 0.97 (0.06) & 0.9 (0.07) \\
stability ($q = 6$)  & 0.0 (0.0) & 0.9 (0.16) & 0.99 (0.04) & 0.91 (0.07) & 0.67 (0.0) \\
\midrule
      \multicolumn{6}{c}{Average number of selected variables} \\
          \midrule
 & 10 & 50 & 100 & 1000 & 100000 \\
\midrule
proposed$^*$  & 3.2 (2.23) & 3.0 (0.63) & 3.0 (0.45) & 3.0 (0.0) & 3.0 (0.0) \\
horseshoe$^*$  & 5.1 (2.62) & 3.8 (1.54) & 3.2 (0.4) & 3.0 (0.0) & 3.0 (0.0) \\
GibbsBvs & 5.6 (2.84) & 3.6 (1.62) & 3.3 (0.46) & 4.8 (0.75) & - \\
EMVS & 1.1 (2.02) & 0.1 (0.3) & 0.0 (0.0) & 2.9 (0.3) & 3.0 (0.0) \\
SSLASSO & 1.4 (1.02) & 2.7 (0.78) & 3.0 (0.45) & 3.9 (0.83) & 7.0 (0.0) \\
AIC & 6.5 (2.2) & 4.1 (1.14) & 3.7 (0.64) & 6.5 (1.02) & 7.9 (0.3) \\
EBIC ($\gamma = 0.0$)  & 6.5 (2.2) & 3.1 (0.7) & 3.3 (0.46) & 5.0 (1.26) & 7.0 (0.0) \\
EBIC ($\gamma = 0.5$)  & 5.7 (2.37) & 2.9 (0.54) & 3.1 (0.54) & 4.5 (1.28) & 7.0 (0.0) \\
EBIC ($\gamma = 1.0$)  & 4.3 (2.65) & 2.9 (0.54) & 3.0 (0.45) & 3.8 (0.87) & 7.0 (0.0) \\
stability ($q = 1$)  & 0.3 (0.46) & 0.7 (0.46) & 1.0 (0.0) & 1.0 (0.0) & 1.0 (0.0) \\
stability ($q = 4$)  & 0.0 (0.0) & 2.6 (0.66) & 2.9 (0.3) & 3.2 (0.4) & 3.7 (0.46) \\
stability ($q = 6$)  & 0.0 (0.0) & 2.7 (0.78) & 3.1 (0.3) & 3.6 (0.49) & 6.0 (0.0) \\
       \bottomrule
    \end{tabular}
 \end{table*}

 % fully updated
  \begin{table*}[h]
 \center
  \caption{Low-dimensional setting, $d = 8$ and $n \in \{10, 50, 100, 1000, 100000\}$. Evaluation results with no noise on regression coefficients. Comparison of the proposed method and horseshoe for different $\delta$.}
  \label{tab:simDataNoise0_lowDimensional_deltaAnalysis}
  \begin{tabular}{llllll}
  \toprule 
      \multicolumn{6}{c}{F1-Scores} \\
      \midrule
 & 10 & 50 & 100 & 1000 & 100000 \\
\midrule
proposed ($\delta$ = 0.8)  & 0.47 (0.18) & 0.88 (0.1) & 1.0 (0.0) & 1.0 (0.0) & 1.0 (0.0) \\
proposed ($\delta$ = 0.5)  & 0.47 (0.18) & 0.88 (0.1) & 1.0 (0.0) & 1.0 (0.0) & 1.0 (0.0) \\
proposed ($\delta$ = 0.05)  & 0.47 (0.19) & 0.89 (0.09) & 1.0 (0.0) & 1.0 (0.0) & 1.0 (0.0) \\
proposed ($\delta$ = 0.01)  & 0.47 (0.19) & 0.89 (0.09) & 1.0 (0.0) & 1.0 (0.0) & 1.0 (0.0) \\
proposed ($\delta$ = 0.001)  & 0.48 (0.19) & 0.87 (0.09) & 0.98 (0.08) & 1.0 (0.0) & 1.0 (0.0) \\
proposed ($\delta$ = 0.0)  & 0.47 (0.19) & 0.87 (0.09) & 0.99 (0.04) & 1.0 (0.0) & 1.0 (0.0) \\
proposed$^*$  & 0.48 (0.18) & 0.88 (0.1) & 1.0 (0.0) & 1.0 (0.0) & 1.0 (0.0) \\
\midrule
horseshoe ($\delta$ = 0.8)  & 0.67 (0.14) & 0.91 (0.1) & 0.99 (0.04) & 1.0 (0.0) & 1.0 (0.0) \\
horseshoe ($\delta$ = 0.5)  & 0.7 (0.14) & 0.94 (0.08) & 0.95 (0.09) & 1.0 (0.0) & 1.0 (0.0) \\
horseshoe ($\delta$ = 0.05)  & 0.55 (0.05) & 0.6 (0.04) & 0.63 (0.09) & 0.71 (0.13) & 1.0 (0.0) \\
horseshoe ($\delta$ = 0.01)  & 0.53 (0.04) & 0.56 (0.02) & 0.56 (0.02) & 0.59 (0.04) & 0.84 (0.13) \\
horseshoe ($\delta$ = 0.001)  & 0.55 (0.0) & 0.55 (0.02) & 0.55 (0.0) & 0.55 (0.0) & 0.56 (0.02) \\
horseshoe ($\delta$ = 0.0)  & 0.55 (0.0) & 0.55 (0.0) & 0.55 (0.0) & 0.55 (0.0) & 0.55 (0.0) \\
horseshoe$^*$  & 0.59 (0.1) & 0.93 (0.14) & 0.95 (0.09) & 1.0 (0.0) & 1.0 (0.0) \\
\midrule
      \multicolumn{6}{c}{Average number of selected variables} \\
          \midrule
 & 10 & 50 & 100 & 1000 & 100000 \\
\midrule
proposed ($\delta$ = 0.8)  & 1.9 (1.87) & 2.4 (0.49) & 3.0 (0.0) & 3.0 (0.0) & 3.0 (0.0) \\
proposed ($\delta$ = 0.5)  & 1.9 (1.87) & 2.4 (0.49) & 3.0 (0.0) & 3.0 (0.0) & 3.0 (0.0) \\
proposed ($\delta$ = 0.05)  & 2.3 (2.65) & 2.6 (0.66) & 3.0 (0.0) & 3.0 (0.0) & 3.0 (0.0) \\
proposed ($\delta$ = 0.01)  & 2.3 (2.65) & 2.6 (0.66) & 3.0 (0.0) & 3.0 (0.0) & 3.0 (0.0) \\
proposed ($\delta$ = 0.001)  & 3.0 (3.1) & 2.5 (0.67) & 3.2 (0.6) & 3.0 (0.0) & 3.0 (0.0) \\
proposed ($\delta$ = 0.0)  & 2.3 (2.65) & 2.5 (0.67) & 3.1 (0.3) & 3.0 (0.0) & 3.0 (0.0) \\
proposed$^*$  & 2.6 (2.58) & 2.4 (0.49) & 3.0 (0.0) & 3.0 (0.0) & 3.0 (0.0) \\
\midrule
horseshoe ($\delta$ = 0.8)  & 2.6 (0.92) & 2.7 (0.64) & 3.1 (0.3) & 3.0 (0.0) & 3.0 (0.0) \\
horseshoe ($\delta$ = 0.5)  & 3.6 (0.92) & 3.2 (0.6) & 3.4 (0.66) & 3.0 (0.0) & 3.0 (0.0) \\
horseshoe ($\delta$ = 0.05)  & 7.6 (0.49) & 7.1 (0.7) & 6.6 (1.11) & 5.7 (1.35) & 3.0 (0.0) \\
horseshoe ($\delta$ = 0.01)  & 7.9 (0.3) & 7.7 (0.46) & 7.7 (0.46) & 7.2 (0.75) & 4.3 (1.27) \\
horseshoe ($\delta$ = 0.001)  & 8.0 (0.0) & 7.9 (0.3) & 8.0 (0.0) & 8.0 (0.0) & 7.7 (0.46) \\
horseshoe ($\delta$ = 0.0)  & 8.0 (0.0) & 8.0 (0.0) & 8.0 (0.0) & 8.0 (0.0) & 8.0 (0.0) \\
horseshoe$^*$  & 4.6 (2.2) & 3.7 (1.49) & 3.4 (0.66) & 3.0 (0.0) & 3.0 (0.0) \\
     \bottomrule
    \end{tabular}
 \end{table*}

 % fully updated
   \begin{table*}[h]
 \center
  \caption{Low-dimensional setting, $d = 8$ and $n \in \{10, 50, 100, 1000, 100000\}$. Evaluation results with noise on regression coefficients $\eta = 0.2$. Comparison of the proposed method and horseshoe for different $\delta$.}
  \label{tab:simDataNoise02_lowDimensional_deltaAnalysis}
  \begin{tabular}{llllll}
  \toprule 
      \multicolumn{6}{c}{F1-Scores} \\
      \midrule
 & 10 & 50 & 100 & 1000 & 100000 \\
\midrule
proposed ($\delta$ = 0.8)  & 0.39 (0.22) & 0.87 (0.19) & 0.97 (0.07) & 1.0 (0.0) & 1.0 (0.0) \\
proposed ($\delta$ = 0.5)  & 0.4 (0.16) & 0.95 (0.08) & 0.97 (0.07) & 1.0 (0.0) & 1.0 (0.0) \\
proposed ($\delta$ = 0.05)  & 0.39 (0.22) & 0.95 (0.08) & 0.99 (0.04) & 1.0 (0.0) & 0.68 (0.04) \\
proposed ($\delta$ = 0.01)  & 0.4 (0.22) & 0.95 (0.08) & 0.99 (0.04) & 0.99 (0.04) & 0.62 (0.03) \\
proposed ($\delta$ = 0.001)  & 0.38 (0.21) & 0.95 (0.08) & 0.99 (0.04) & 0.99 (0.04) & 0.62 (0.03) \\
proposed ($\delta$ = 0.0)  & 0.37 (0.2) & 0.95 (0.08) & 0.99 (0.04) & 0.99 (0.04) & 0.62 (0.03) \\
proposed$^*$  & 0.41 (0.16) & 0.95 (0.08) & 0.97 (0.07) & 1.0 (0.0) & 1.0 (0.0) \\
\midrule
horseshoe ($\delta$ = 0.8)  & 0.61 (0.29) & 0.95 (0.08) & 0.99 (0.04) & 1.0 (0.0) & 1.0 (0.0) \\
horseshoe ($\delta$ = 0.5)  & 0.59 (0.29) & 0.92 (0.08) & 0.99 (0.04) & 1.0 (0.0) & 1.0 (0.0) \\
horseshoe ($\delta$ = 0.05)  & 0.59 (0.11) & 0.61 (0.1) & 0.62 (0.06) & 0.65 (0.07) & 0.66 (0.04) \\
horseshoe ($\delta$ = 0.01)  & 0.55 (0.0) & 0.56 (0.04) & 0.56 (0.02) & 0.57 (0.03) & 0.58 (0.03) \\
horseshoe ($\delta$ = 0.001)  & 0.55 (0.0) & 0.55 (0.0) & 0.55 (0.0) & 0.55 (0.0) & 0.55 (0.0) \\
horseshoe ($\delta$ = 0.0)  & 0.55 (0.0) & 0.55 (0.0) & 0.55 (0.0) & 0.55 (0.0) & 0.55 (0.0) \\
horseshoe$^*$  & 0.6 (0.21) & 0.88 (0.13) & 0.99 (0.04) & 1.0 (0.0) & 1.0 (0.0) \\
\midrule
      \multicolumn{6}{c}{Average number of selected variables} \\
          \midrule
 & 10 & 50 & 100 & 1000 & 100000 \\
\midrule
proposed ($\delta$ = 0.8)  & 2.5 (2.16) & 2.6 (0.92) & 3.0 (0.45) & 3.0 (0.0) & 3.0 (0.0) \\
proposed ($\delta$ = 0.5)  & 2.6 (2.06) & 2.9 (0.54) & 3.0 (0.45) & 3.0 (0.0) & 3.0 (0.0) \\
proposed ($\delta$ = 0.05)  & 2.5 (2.25) & 2.9 (0.54) & 3.1 (0.3) & 3.0 (0.0) & 5.9 (0.54) \\
proposed ($\delta$ = 0.01)  & 2.8 (2.6) & 2.9 (0.54) & 3.1 (0.3) & 3.1 (0.3) & 6.7 (0.46) \\
proposed ($\delta$ = 0.001)  & 2.6 (2.46) & 2.9 (0.54) & 3.1 (0.3) & 3.1 (0.3) & 6.7 (0.46) \\
proposed ($\delta$ = 0.0)  & 2.7 (2.61) & 2.9 (0.54) & 3.1 (0.3) & 3.1 (0.3) & 6.7 (0.46) \\
proposed$^*$  & 2.5 (2.06) & 2.9 (0.54) & 3.0 (0.45) & 3.0 (0.0) & 3.0 (0.0) \\
\midrule
horseshoe ($\delta$ = 0.8)  & 3.1 (1.04) & 2.9 (0.54) & 3.1 (0.3) & 3.0 (0.0) & 3.0 (0.0) \\
horseshoe ($\delta$ = 0.5)  & 3.8 (0.98) & 3.3 (0.64) & 3.1 (0.3) & 3.0 (0.0) & 3.0 (0.0) \\
horseshoe ($\delta$ = 0.05)  & 7.1 (1.22) & 7.0 (1.34) & 6.7 (0.9) & 6.4 (1.02) & 6.1 (0.54) \\
horseshoe ($\delta$ = 0.01)  & 8.0 (0.0) & 7.8 (0.6) & 7.8 (0.4) & 7.6 (0.49) & 7.4 (0.49) \\
horseshoe ($\delta$ = 0.001)  & 8.0 (0.0) & 8.0 (0.0) & 8.0 (0.0) & 8.0 (0.0) & 8.0 (0.0) \\
horseshoe ($\delta$ = 0.0)  & 8.0 (0.0) & 8.0 (0.0) & 8.0 (0.0) & 8.0 (0.0) & 8.0 (0.0) \\
horseshoe$^*$  & 4.8 (2.23) & 3.8 (1.54) & 3.1 (0.3) & 3.0 (0.0) & 3.0 (0.0) \\
     \bottomrule
    \end{tabular}
 \end{table*}

 % fully updated
    \begin{table*}[h]
 \center
  \caption{Low-dimensional setting, $d = 8$ and $n \in \{10, 50, 100, 1000, 100000\}$. Evaluation results with noise on regression coefficients $\eta = 0.5$. Comparison of the proposed method and horseshoe for different $\delta$.}
  \label{tab:simDataNoise05_lowDimensional_deltaAnalysis}
  \begin{tabular}{llllll}
  \toprule 
      \multicolumn{6}{c}{F1-Scores} \\
      \midrule
 & 10 & 50 & 100 & 1000 & 100000 \\
\midrule
proposed ($\delta$ = 0.8)  & 0.36 (0.25) & 0.85 (0.19) & 0.97 (0.07) & 1.0 (0.0) & 1.0 (0.0) \\
proposed ($\delta$ = 0.5)  & 0.45 (0.18) & 0.93 (0.09) & 0.97 (0.07) & 1.0 (0.0) & 1.0 (0.0) \\
proposed ($\delta$ = 0.05)  & 0.34 (0.24) & 0.93 (0.09) & 0.95 (0.08) & 0.8 (0.11) & 0.6 (0.0) \\
proposed ($\delta$ = 0.01)  & 0.34 (0.24) & 0.93 (0.09) & 0.95 (0.08) & 0.8 (0.11) & 0.59 (0.02) \\
proposed ($\delta$ = 0.001)  & 0.34 (0.24) & 0.93 (0.09) & 0.95 (0.08) & 0.8 (0.11) & 0.59 (0.02) \\
proposed ($\delta$ = 0.0)  & 0.39 (0.22) & 0.93 (0.09) & 0.95 (0.08) & 0.8 (0.11) & 0.59 (0.02) \\
proposed$^*$  & 0.47 (0.18) & 0.93 (0.09) & 0.97 (0.07) & 1.0 (0.0) & 1.0 (0.0) \\
\midrule
horseshoe ($\delta$ = 0.8)  & 0.57 (0.28) & 0.93 (0.09) & 0.99 (0.04) & 1.0 (0.0) & 1.0 (0.0) \\
horseshoe ($\delta$ = 0.5)  & 0.6 (0.28) & 0.91 (0.08) & 0.97 (0.06) & 1.0 (0.0) & 1.0 (0.0) \\
horseshoe ($\delta$ = 0.05)  & 0.59 (0.1) & 0.59 (0.06) & 0.61 (0.06) & 0.6 (0.05) & 0.6 (0.0) \\
horseshoe ($\delta$ = 0.01)  & 0.56 (0.04) & 0.56 (0.02) & 0.56 (0.02) & 0.56 (0.02) & 0.55 (0.0) \\
horseshoe ($\delta$ = 0.001)  & 0.55 (0.0) & 0.55 (0.0) & 0.55 (0.0) & 0.55 (0.0) & 0.55 (0.0) \\
horseshoe ($\delta$ = 0.0)  & 0.55 (0.0) & 0.55 (0.0) & 0.55 (0.0) & 0.55 (0.0) & 0.55 (0.0) \\
horseshoe$^*$  & 0.63 (0.19) & 0.88 (0.13) & 0.97 (0.06) & 1.0 (0.0) & 1.0 (0.0) \\
\midrule
      \multicolumn{6}{c}{Average number of selected variables} \\
          \midrule
 & 10 & 50 & 100 & 1000 & 100000 \\
\midrule
proposed ($\delta$ = 0.8)  & 2.4 (2.11) & 2.5 (0.92) & 3.0 (0.45) & 3.0 (0.0) & 3.0 (0.0) \\
proposed ($\delta$ = 0.5)  & 3.2 (2.36) & 3.0 (0.63) & 3.0 (0.45) & 3.0 (0.0) & 3.0 (0.0) \\
proposed ($\delta$ = 0.05)  & 2.5 (2.42) & 3.0 (0.63) & 3.1 (0.54) & 4.6 (0.92) & 7.0 (0.0) \\
proposed ($\delta$ = 0.01)  & 2.5 (2.42) & 3.0 (0.63) & 3.1 (0.54) & 4.6 (0.92) & 7.1 (0.3) \\
proposed ($\delta$ = 0.001)  & 2.5 (2.42) & 3.0 (0.63) & 3.1 (0.54) & 4.6 (0.92) & 7.1 (0.3) \\
proposed ($\delta$ = 0.0)  & 2.6 (2.33) & 3.0 (0.63) & 3.1 (0.54) & 4.6 (0.92) & 7.1 (0.3) \\
proposed$^*$  & 3.2 (2.23) & 3.0 (0.63) & 3.0 (0.45) & 3.0 (0.0) & 3.0 (0.0) \\
\midrule
horseshoe ($\delta$ = 0.8)  & 3.1 (1.22) & 3.0 (0.63) & 3.1 (0.3) & 3.0 (0.0) & 3.0 (0.0) \\
horseshoe ($\delta$ = 0.5)  & 3.6 (1.02) & 3.4 (0.66) & 3.2 (0.4) & 3.0 (0.0) & 3.0 (0.0) \\
horseshoe ($\delta$ = 0.05)  & 7.3 (1.27) & 7.2 (0.98) & 6.9 (0.83) & 7.1 (0.83) & 7.0 (0.0) \\
horseshoe ($\delta$ = 0.01)  & 7.8 (0.6) & 7.7 (0.46) & 7.8 (0.4) & 7.8 (0.4) & 8.0 (0.0) \\
horseshoe ($\delta$ = 0.001)  & 8.0 (0.0) & 8.0 (0.0) & 8.0 (0.0) & 8.0 (0.0) & 8.0 (0.0) \\
horseshoe ($\delta$ = 0.0)  & 8.0 (0.0) & 8.0 (0.0) & 8.0 (0.0) & 8.0 (0.0) & 8.0 (0.0) \\
horseshoe$^*$  & 5.1 (2.62) & 3.8 (1.54) & 3.2 (0.4) & 3.0 (0.0) & 3.0 (0.0) \\
     \bottomrule
    \end{tabular}
 \end{table*}

 % fully updated
   \begin{table*}[h]
 \center
  \caption{High-dimensional setting, $d = 1000$ and $n \in \{100, 1000\}$. Evaluation results with no noise on regression coefficients.}
  \label{tab:simDataNoise0_highDimensional}
  \begin{tabular}{lll}
  \toprule 
      \multicolumn{3}{c}{F1-Scores} \\
      \midrule
 & 100 & 1000 \\
\midrule
proposed$^*$  & 0.93 (0.09) & 1.0 (0.0) \\
horseshoe$^*$  & 0.45 (0.23) & 0.86 (0.23) \\
EMVS & 0.94 (0.09) & 1.0 (0.0) \\
SSLASSO & 0.99 (0.04) & 1.0 (0.0) \\
AIC & 0.06 (0.0) & 0.01 (0.0) \\
EBIC ($\gamma = 0.0$)  & 0.06 (0.0) & 0.01 (0.0) \\
EBIC ($\gamma = 0.5$)  & 0.06 (0.0) & 0.01 (0.0) \\
EBIC ($\gamma = 1.0$)  & 0.06 (0.0) & 0.01 (0.0) \\
stability ($q = 1$)  & 0.25 (0.25) & 0.5 (0.0) \\
stability ($q = 50$)  & 1.0 (0.0) & 1.0 (0.0) \\
stability ($q = 100$)  & 0.88 (0.1) & 1.0 (0.0) \\
\midrule
      \multicolumn{3}{c}{Average number of selected variables} \\
          \midrule
 & 100 & 1000 \\
\midrule
proposed$^*$  & 2.8 (0.6) & 3.0 (0.0) \\
horseshoe$^*$  & 29.7 (53.92) & 5.2 (5.02) \\
EMVS & 2.7 (0.46) & 3.0 (0.0) \\
SSLASSO & 3.1 (0.3) & 3.0 (0.0) \\
AIC & 99.0 (0.0) & 999.0 (0.0) \\
EBIC ($\gamma = 0.0$)  & 99.0 (0.0) & 999.0 (0.0) \\
EBIC ($\gamma = 0.5$)  & 99.0 (0.0) & 999.0 (0.0) \\
EBIC ($\gamma = 1.0$)  & 99.0 (0.0) & 999.0 (0.0) \\
stability ($q = 1$)  & 0.5 (0.5) & 1.0 (0.0) \\
stability ($q = 50$)  & 3.0 (0.0) & 3.0 (0.0) \\
stability ($q = 100$)  & 2.4 (0.49) & 3.0 (0.0) \\
     \bottomrule
    \end{tabular}
 \end{table*}

   \begin{table*}[h]
 \center
  \caption{High-dimensional setting, $d = 1000$ and $n \in \{100, 1000\}$. Evaluation results with noise on regression coefficients $\eta = 0.2$.}
  \label{tab:simDataNoise02_highDimensional}
  \begin{tabular}{lll}
  \toprule 
      \multicolumn{3}{c}{F1-Scores} \\
      \midrule
 & 100 & 1000 \\
\midrule
proposed$^*$  & 0.84 (0.08) & 1.0 (0.0) \\
horseshoe$^*$  & 0.51 (0.21) & 0.63 (0.17) \\
EMVS & 0.86 (0.09) & 0.98 (0.06) \\
SSLASSO & 0.94 (0.09) & 0.99 (0.04) \\
AIC & 0.06 (0.0) & 0.01 (0.0) \\
EBIC ($\gamma = 0.0$)  & 0.06 (0.0) & 0.01 (0.0) \\
EBIC ($\gamma = 0.5$)  & 0.06 (0.0) & 0.01 (0.0) \\
EBIC ($\gamma = 1.0$)  & 0.06 (0.0) & 0.01 (0.0) \\
stability ($q = 1$)  & 0.45 (0.15) & 0.4 (0.2) \\
stability ($q = 50$)  & 0.94 (0.09) & 0.96 (0.07) \\
stability ($q = 100$)  & 0.88 (0.1) & 0.97 (0.06) \\
\midrule
      \multicolumn{3}{c}{Average number of selected variables} \\
          \midrule
 & 100 & 1000 \\
\midrule
proposed$^*$  & 2.2 (0.4) & 3.0 (0.0) \\
horseshoe$^*$  & 11.8 (8.78) & 7.7 (4.58) \\
EMVS & 2.3 (0.46) & 2.9 (0.3) \\
SSLASSO & 2.7 (0.46) & 3.1 (0.3) \\
AIC & 99.0 (0.0) & 999.0 (0.0) \\
EBIC ($\gamma = 0.0$)  & 99.0 (0.0) & 999.0 (0.0) \\
EBIC ($\gamma = 0.5$)  & 99.0 (0.0) & 999.0 (0.0) \\
EBIC ($\gamma = 1.0$)  & 99.0 (0.0) & 999.0 (0.0) \\
stability ($q = 1$)  & 0.9 (0.3) & 0.8 (0.4) \\
stability ($q = 50$)  & 2.7 (0.46) & 3.3 (0.46) \\
stability ($q = 100$)  & 2.4 (0.49) & 3.2 (0.4) \\
     \bottomrule
    \end{tabular}
 \end{table*}

 % fully updated
   \begin{table*}[h]
 \center
  \caption{High-dimensional setting, $d = 1000$ and $n \in \{100, 1000\}$. Evaluation results with no noise on regression coefficients. Comparison of the proposed method and horseshoe for different $\delta$.}
  \label{tab:simDataNoise0_highDimensional_deltaAnalysis}
  \begin{tabular}{lll}
  \toprule 
      \multicolumn{3}{c}{F1-Scores} \\
      \midrule
 & 100 & 1000 \\
\midrule
proposed ($\delta$ = 0.8)  & 0.5 (0.0) & 0.82 (0.06) \\
proposed ($\delta$ = 0.5)  & 0.53 (0.09) & 1.0 (0.0) \\
proposed ($\delta$ = 0.05)  & 0.93 (0.09) & 1.0 (0.0) \\
proposed ($\delta$ = 0.01)  & 0.93 (0.09) & 1.0 (0.0) \\
proposed ($\delta$ = 0.001)  & 0.93 (0.09) & 1.0 (0.0) \\
proposed ($\delta$ = 0.0)  & 0.93 (0.09) & 1.0 (0.0) \\
proposed$^*$  & 0.93 (0.09) & 1.0 (0.0) \\
\midrule
horseshoe ($\delta$ = 0.8)  & 0.92 (0.1) & 1.0 (0.0) \\
horseshoe ($\delta$ = 0.5)  & 0.95 (0.08) & 1.0 (0.0) \\
horseshoe ($\delta$ = 0.05)  & 0.58 (0.2) & 0.94 (0.11) \\
horseshoe ($\delta$ = 0.01)  & 0.18 (0.06) & 0.3 (0.06) \\
horseshoe ($\delta$ = 0.001)  & 0.01 (0.0) & 0.01 (0.0) \\
horseshoe ($\delta$ = 0.0)  & 0.01 (0.0) & 0.01 (0.0) \\
horseshoe$^*$  & 0.45 (0.23) & 0.86 (0.23) \\
\midrule
      \multicolumn{3}{c}{Average number of selected variables} \\
          \midrule
 & 100 & 1000 \\
\midrule
proposed ($\delta$ = 0.8)  & 1.0 (0.0) & 2.1 (0.3) \\
proposed ($\delta$ = 0.5)  & 1.1 (0.3) & 3.0 (0.0) \\
proposed ($\delta$ = 0.05)  & 2.8 (0.6) & 3.0 (0.0) \\
proposed ($\delta$ = 0.01)  & 2.8 (0.6) & 3.0 (0.0) \\
proposed ($\delta$ = 0.001)  & 2.8 (0.6) & 3.0 (0.0) \\
proposed ($\delta$ = 0.0)  & 2.8 (0.6) & 3.0 (0.0) \\
proposed$^*$  & 2.8 (0.6) & 3.0 (0.0) \\
\midrule
horseshoe ($\delta$ = 0.8)  & 2.6 (0.49) & 3.0 (0.0) \\
horseshoe ($\delta$ = 0.5)  & 2.9 (0.54) & 3.0 (0.0) \\
horseshoe ($\delta$ = 0.05)  & 9.6 (7.53) & 3.5 (0.92) \\
horseshoe ($\delta$ = 0.01)  & 43.9 (48.84) & 17.7 (4.1) \\
horseshoe ($\delta$ = 0.001)  & 495.1 (119.73) & 404.2 (22.27) \\
horseshoe ($\delta$ = 0.0)  & 1000.0 (0.0) & 1000.0 (0.0) \\
horseshoe$^*$  & 29.7 (53.92) & 5.2 (5.02) \\
     \bottomrule
    \end{tabular}
 \end{table*}

 % fully updated
    \begin{table*}[h]
 \center
  \caption{High-dimensional setting, $d = 1000$ and $n \in \{100, 1000\}$. Evaluation results with noise on regression coefficients $\eta = 0.2$. Comparison of the proposed method and horseshoe for different $\delta$.}
  \label{tab:simDataNoise02_highDimensional_deltaAnalysis}
  \begin{tabular}{lll}
  \toprule 
      \multicolumn{3}{c}{F1-Scores} \\
      \midrule
 & 100 & 1000 \\
\midrule
proposed ($\delta$ = 0.8)  & 0.45 (0.15) & 0.82 (0.06) \\
proposed ($\delta$ = 0.5)  & 0.5 (0.0) & 0.98 (0.06) \\
proposed ($\delta$ = 0.05)  & 0.84 (0.08) & 0.97 (0.06) \\
proposed ($\delta$ = 0.01)  & 0.84 (0.08) & 0.97 (0.06) \\
proposed ($\delta$ = 0.001)  & 0.84 (0.08) & 0.97 (0.06) \\
proposed ($\delta$ = 0.0)  & 0.84 (0.08) & 0.97 (0.06) \\
proposed$^*$  & 0.84 (0.08) & 1.0 (0.0) \\
\midrule
horseshoe ($\delta$ = 0.8)  & 0.86 (0.09) & 1.0 (0.0) \\
horseshoe ($\delta$ = 0.5)  & 0.92 (0.1) & 1.0 (0.0) \\
horseshoe ($\delta$ = 0.05)  & 0.69 (0.12) & 0.77 (0.13) \\
horseshoe ($\delta$ = 0.01)  & 0.2 (0.03) & 0.3 (0.05) \\
horseshoe ($\delta$ = 0.001)  & 0.01 (0.0) & 0.01 (0.0) \\
horseshoe ($\delta$ = 0.0)  & 0.01 (0.0) & 0.01 (0.0) \\
horseshoe$^*$  & 0.51 (0.21) & 0.63 (0.17) \\
\midrule
      \multicolumn{3}{c}{Average number of selected variables} \\
          \midrule
 & 100 & 1000 \\
\midrule
proposed ($\delta$ = 0.8)  & 0.9 (0.3) & 2.1 (0.3) \\
proposed ($\delta$ = 0.5)  & 1.0 (0.0) & 2.9 (0.3) \\
proposed ($\delta$ = 0.05)  & 2.2 (0.4) & 3.2 (0.4) \\
proposed ($\delta$ = 0.01)  & 2.2 (0.4) & 3.2 (0.4) \\
proposed ($\delta$ = 0.001)  & 2.2 (0.4) & 3.2 (0.4) \\
proposed ($\delta$ = 0.0)  & 2.2 (0.4) & 3.2 (0.4) \\
proposed$^*$  & 2.2 (0.4) & 3.0 (0.0) \\
\midrule
horseshoe ($\delta$ = 0.8)  & 2.3 (0.46) & 3.0 (0.0) \\
horseshoe ($\delta$ = 0.5)  & 2.6 (0.49) & 3.0 (0.0) \\
horseshoe ($\delta$ = 0.05)  & 5.6 (1.62) & 5.0 (1.18) \\
horseshoe ($\delta$ = 0.01)  & 28.8 (8.2) & 18.0 (3.9) \\
horseshoe ($\delta$ = 0.001)  & 495.1 (58.04) & 417.0 (20.53) \\
horseshoe ($\delta$ = 0.0)  & 1000.0 (0.0) & 1000.0 (0.0) \\
horseshoe$^*$  & 11.8 (8.78) & 7.7 (4.58) \\
     \bottomrule
    \end{tabular}
 \end{table*}
 
  \begin{table*}[h]
 \center
  \caption{High-dimensional setting, $d = 1000$ and $n \in \{100, 1000\}$. Evaluation results with no noise on regression coefficients. For different $\delta$, comparison of the proposed method and product moment matching priors (MOM) from \citep{johnson2012bayesian}.}
  \label{tab:simDataNoise0_highDimensional_deltaAnalysisVS_MOM}
  \begin{tabular}{lll}
  \toprule 
      \multicolumn{3}{c}{F1-Scores} \\
      \midrule
 & 100 & 1000 \\
\midrule
proposed ($\delta = 0.8$)  & 0.5 (0.0) & 0.82 (0.06) \\
proposed ($\delta = 0.5$)  & 0.53 (0.09) & 1.0 (0.0) \\
proposed ($\delta = 0.05$)  & 0.93 (0.09) & 1.0 (0.0) \\
proposed ($\delta = 0.01$)  & 0.93 (0.09) & 1.0 (0.0) \\
proposed ($\delta = 0.001$)  & 0.93 (0.09) & 1.0 (0.0) \\
\midrule
MOM ($\delta = 0.8$)  & 0.9 (0.1) & 1.0 (0.0) \\
MOM ($\delta = 0.5$)  & 0.89 (0.09) & 1.0 (0.0) \\
MOM ($\delta = 0.05$)  & 1.0 (0.0) & 1.0 (0.0) \\
MOM ($\delta = 0.01$)  & 0.0 (0.0) & 0.87 (0.04) \\
MOM ($\delta = 0.001$)  & 0.0 (0.0) & 0.05 (0.15) \\
\midrule
      \multicolumn{3}{c}{Average number of selected variables} \\
          \midrule
 & 100 & 1000 \\
\midrule
proposed ($\delta = 0.8$)  & 1.0 (0.0) & 2.1 (0.3) \\
proposed ($\delta = 0.5$)  & 1.1 (0.3) & 3.0 (0.0) \\
proposed ($\delta = 0.05$)  & 2.8 (0.6) & 3.0 (0.0) \\
proposed ($\delta = 0.01$)  & 2.8 (0.6) & 3.0 (0.0) \\
proposed ($\delta = 0.001$)  & 2.8 (0.6) & 3.0 (0.0) \\
\midrule
MOM ($\delta = 0.8$)  & 2.5 (0.5) & 3.0 (0.0) \\
MOM ($\delta = 0.5$)  & 2.6 (0.66) & 3.0 (0.0) \\
MOM ($\delta = 0.05$)  & 3.0 (0.0) & 3.0 (0.0) \\
MOM ($\delta = 0.01$)  & 0.0 (0.0) & 3.9 (0.3) \\
MOM ($\delta = 0.001$)  & 0.0 (0.0) & 0.1 (0.3) \\
     \bottomrule
    \end{tabular}
 \end{table*}

   \begin{table*}[h]
 \center
  \caption{High-dimensional setting, $d = 1000$ and $n \in \{100, 1000\}$. Evaluation results with noise on regression coefficients $\eta = 0.2$. For different $\delta$, comparison of the proposed method and product moment matching priors (MOM) from \citep{johnson2012bayesian}.}
  \label{tab:simDataNoise02_highDimensional_deltaAnalysisVS_MOM}
  \begin{tabular}{lll}
  \toprule 
      \multicolumn{3}{c}{F1-Scores} \\
      \midrule
 & 100 & 1000 \\
\midrule
proposed ($\delta = 0.8$)  & 0.45 (0.15) & 0.82 (0.06) \\
proposed ($\delta = 0.5$)  & 0.5 (0.0) & 0.98 (0.06) \\
proposed ($\delta = 0.05$)  & 0.84 (0.08) & 0.97 (0.06) \\
proposed ($\delta = 0.01$)  & 0.84 (0.08) & 0.97 (0.06) \\
proposed ($\delta = 0.001$)  & 0.84 (0.08) & 0.97 (0.06) \\
\midrule
MOM ($\delta = 0.8$)  & 0.84 (0.08) & 0.99 (0.04) \\
MOM ($\delta = 0.5$)  & 0.84 (0.08) & 0.99 (0.04) \\
MOM ($\delta = 0.05$)  & 0.98 (0.06) & 0.96 (0.07) \\
MOM ($\delta = 0.01$)  & 0.0 (0.0) & 0.89 (0.06) \\
MOM ($\delta = 0.001$)  & 0.0 (0.0) & 0.0 (0.0) \\
\midrule
      \multicolumn{3}{c}{Average number of selected variables} \\
          \midrule
 & 100 & 1000 \\
\midrule
proposed ($\delta = 0.8$)  & 0.9 (0.3) & 2.1 (0.3) \\
proposed ($\delta = 0.5$)  & 1.0 (0.0) & 2.9 (0.3) \\
proposed ($\delta = 0.05$)  & 2.2 (0.4) & 3.2 (0.4) \\
proposed ($\delta = 0.01$)  & 2.2 (0.4) & 3.2 (0.4) \\
proposed ($\delta = 0.001$)  & 2.2 (0.4) & 3.2 (0.4) \\
\midrule
MOM ($\delta = 0.8$)  & 2.2 (0.4) & 3.1 (0.3) \\
MOM ($\delta = 0.5$)  & 2.2 (0.4) & 3.1 (0.3) \\
MOM ($\delta = 0.05$)  & 2.9 (0.3) & 3.3 (0.46) \\
MOM ($\delta = 0.01$)  & 0.0 (0.0) & 3.8 (0.4) \\
MOM ($\delta = 0.001$)  & 0.0 (0.0) & 0.0 (0.0) \\
     \bottomrule
    \end{tabular}
 \end{table*}

\section*{Appendix D: Details of real data experiments and additional results}

% In this section, we compare the results of our proposed and all baselines on three real data sets: crime data \citep{raftery1997bayesian,liang2008mixtures}, ozone data  \citep{garcia2013sampling}, and GDP growth data (SDM) \citep{sala2004determinants}.

The details of all data sets are shown in Table \ref{tab:realData_properties}; all variables are described in Tables \ref{tab:ozone_variables}, \ref{tab:crime_variables}  and Tables \ref{tab:SDM_variables1} and \ref{tab:SDM_variables2}. 
In order to make the choice of all hyper-parameters invariant to the scale, we normalize the observations to have roughly the same scale as for the synthetic data set.
In detail, we normalize the covariates to have mean 0 and variance 1, and the response variable to have mean 0 and variance 30. 
Furthermore, we log-transform the crime data as in \citep{liang2008mixtures}.

For the experiments with the real data we use 100000 MCMC-samples for the proposed method, GibbsBvs, and the horseshoe prior.\footnote{Out of which 10\% are used for burn in.}
Concerning the stability selection method, based on our findings from the simulated data, we set $q$ to the values $\{0.1 \cdot d, 0.5 \cdot d, 0.8 \cdot d\}$.

Tn Table \ref{tab:realData_SDM}, we also show the results on the GDP growth data (SDM) \citep{sala2004determinants}.

%We see that the horseshoe method and EMVS perform similar as for the simulated data.
%The horseshoe prior with thresholding, finds models with relatively many variables, whereas EMVS tends to select models with only very few variables.
%In particular, EMVS suggests that in the ozone and SDM data, none of the variables are relevant, which is quite a strong statement that is contradicting the results from all other methods.
% The stability selection method appears to select too few variables, independent of the setting of $q$. We note that for SDM, for $q = 0.8 \cdot d$, the stability selection method did not terminate correctly. 
% The results for EBIC highlight the sensitivity to the hyper-parameter $\gamma$.
%Our proposed method shows similar results to GibbsBvs, except for SDM.

\paragraph{GDP growth data (SDM)}
For SDM, our proposed model suggests that only EAST and MALFAL66 have relatively high regression coefficients, but our method also shows that the expected increase in mean-squared error is around 27\% when compared to the Bayesian averaged model that uses all variables.
In Table \ref{tab:rankingVars_SDM}, we show the inclusion probabilities of the proposed method together with the results reported in \citep{sala2004determinants}.
We see that all the top 18 variables that have been considered as significant by \citep{sala2004determinants} are also listed in the top 18 of the proposed method ($\delta = 0$).
Moreover, the results of the proposed method with $\delta = 0.5$, suggest, that among those 18 variables, only 7 variables have a probability of more than $20\%$ of having a high effect size.
 In particular, it appears that DENS65C (density of costal population) seems to have only marginal influence on economic growth. 
 
 \begin{table*}[h]
  \center
  \caption{Statistics of real data sets}
  \label{tab:realData_properties}
  \begin{tabular}{llll}
  \toprule 
 & ozone & crime & SDM \\
\midrule
$n$ &  178 & 47 & 88 \\
$d$ &  35 & 15 & 67 \\
\bottomrule
    \end{tabular}
 \end{table*}

  \begin{table*}[h]
  \center
  \caption{Response $y$ and covariates  of ozone data. The data set contains all of the variables below, including all second-order terms and interactions. This table is partly copied from Table 5 in the supplement material of \citep{garcia2013sampling}.}
  \label{tab:ozone_variables}
  \begin{tabular}{ll}
    \toprule 
  y & Daily maximum 1-hour-average ozone reading (ppm) at Upland, CA \\
 x4 & 500-millibar pressure height (m) measured at Vandenberg AFB \\
x5 & Wind speed (mph) at Los Angeles International Airport (LAX) \\
x6 & Humidity (\%) at LAX \\
x7 & Temperature (Fahrenheit degrees) measured at Sandburg, CA \\
x8 & Inversion base height (feet) at LAX \\
x9 & Pressure gradient (mm Hg) from LAX to Daggett, CA \\
x10 & Visibility (miles) measured at LAX \\
\bottomrule
    \end{tabular}
 \end{table*}
 
   \begin{table*}[h]
  \center
  \caption{Response $y$ and covariates of crime data. This table is partly copied from Table 4 in \citep{raftery1997bayesian}.}
  \label{tab:crime_variables}
  \begin{tabular}{ll}
    \toprule 
   y & crime rate \\
M & Percentage of males age 14-24 \\
So & Indicator variable for southern state  \\
Ed & Mean years of schooling  \\
Po1 & Police expenditure in 1960  \\
Po2 & Police expenditure in 1959  \\
LF & Labor force participation rate  \\
M.F & Number of males per 1,000 females  \\
Pop & State population  \\
NW & Number of nonwhites per 1,000 people  \\
U1 & Unemployment rate of urban males age 14-24  \\
U2  & Unemployment rate of urban males, age 35-39  \\
GDP & Wealth  \\
Ineq & Income inequality  \\
Prob & Probability of imprisonment  \\ 
Time & Average time served in state prisons  \\
\bottomrule
    \end{tabular}
 \end{table*}

   \begin{table*}[h]
  \center
  \caption{Response $y$ and covariates (first part) of SDM data. This table is copied from the description of R package 'BayesVarSel'.}
  \label{tab:SDM_variables1}
  \begin{tabular}{ll}
    \toprule 
\footnotesize y & \footnotesize Growth of GDP per capita at purchasing power parities between 1960 and 1996. \\
\footnotesize ABSLATIT & \footnotesize Absolute latitude. \\
\footnotesize AIRDIST & \footnotesize Logarithm of minimal distance (in km) from New York, Rotterdam, or Tokyo. \\
\footnotesize AVELF & \footnotesize Average of five different indices of ethnolinguistic fractionalization \\ % which is the probability of two random people in a country not speaking the same language. \\
\footnotesize BRIT & \footnotesize Dummy for former British colony after 1776. \\
\footnotesize BUDDHA & \footnotesize Fraction of population Buddhist in 1960. \\
\footnotesize CATH00 & \footnotesize Fraction of population Catholic in 1960. \\
\footnotesize CIV72 & \footnotesize Index of civil liberties index in 1972. \\
\footnotesize COLONY & \footnotesize Dummy for former colony. \\
\footnotesize CONFUC & \footnotesize Fraction of population Confucian. \\
\footnotesize DENS60 & \footnotesize Population per area in 1960. \\
\footnotesize DENS65C & \footnotesize Coastal (within 100 km of coastline) population per coastal area in 1965. \\
\footnotesize DENS65I & \footnotesize Interior (more than 100 km from coastline) population per interior area in 1965. \\
\footnotesize DPOP6090 & \footnotesize Average growth rate of population between 1960 and 1990. \\
\footnotesize EAST & \footnotesize Dummy for East Asian countries. \\
\footnotesize ECORG & \footnotesize Degree Capitalism index. \\
\footnotesize ENGFRAC & \footnotesize Fraction of population speaking English. \\
\footnotesize EUROPE & \footnotesize Dummy for European economies. \\
\footnotesize FERTLDC1 & \footnotesize Fertility in 1960's. \\
\footnotesize GDE1 & \footnotesize Average share public expenditures on defense as fraction of GDP between 1960 and 1965. \\
\footnotesize GDPCH60L & \footnotesize Logarithm of GDP per capita in 1960. \\
\footnotesize GEEREC1 & \footnotesize Average share public expenditures on education as fraction of GDP between 1960 and 1965. \\
\footnotesize GGCFD3 & \footnotesize Average share of expenditures on public investment as fraction of GDP between 1960 and 1965. \\
\footnotesize GOVNOM1 & \footnotesize Average share of nominal government spending to nominal GDP between 1960 and 1964. \\
\footnotesize GOVSH61 & \footnotesize Average share government spending to GDP between 1960 and 1964. \\
\footnotesize GVR61 & \footnotesize Share of expenditures on government consumption to GDP in 1961. \\
\footnotesize H60 & \footnotesize Enrollment rates in higher education. \\
\footnotesize HERF00 & \footnotesize Religion measure. \\ 
\footnotesize HINDU00 & \footnotesize Fraction of the population Hindu in 1960. \\
\footnotesize IPRICE1 & \footnotesize Average investment price level between 1960 and 1964 on purchasing power parity basis. \\
\footnotesize LAAM & \footnotesize Dummy for Latin American countries. \\
\footnotesize LANDAREA & \footnotesize Area in km. \\
\footnotesize LANDLOCK & \footnotesize Dummy for landlocked countries. \\
\bottomrule
    \end{tabular}
 \end{table*}
 
    \begin{table*}[h]
  \center
  \caption{Covariates (second part) of SDM data. This table is copied from the description of R package 'BayesVarSel'.}
  \label{tab:SDM_variables2}
  \begin{tabular}{ll}
    \toprule 
\footnotesize LHCPC & \footnotesize Log of hydrocarbon deposits in 1993. \\
\footnotesize LIFE060 & \footnotesize Life expectancy in 1960. \\
\footnotesize LT100CR & \footnotesize Proportion of country's land area within 100 km of ocean or ocean-navigable river. \\
\footnotesize MALFAL66 & \footnotesize Index of malaria prevalence in 1966. \\
\footnotesize MINING & \footnotesize Fraction of GDP in mining. \\
\footnotesize MUSLIM00 & \footnotesize Fraction of population Muslim in 1960. \\
\footnotesize NEWSTATE & \footnotesize National independence. \\ % : 0 if before 1914; 1 if between 1914 and 1945; 2 if between 1946 and 1989; and 3 if after 1989. \\
\footnotesize OIL & \footnotesize Dummy for oil-producing country. \\
\footnotesize OPENDEC1 & \footnotesize Ratio of exports plus imports to GDP, averaged over 1965 to 1974. \\
\footnotesize ORTH00 & \footnotesize Fraction of population Orthodox in 1960. \\
\footnotesize OTHFRAC & \footnotesize Fraction of population speaking foreign language. \\
\footnotesize P60 & \footnotesize Enrollment rate in primary education in 1960. \\
\footnotesize PI6090 & \footnotesize Average inflation rate between 1960 and 1990. \\
\footnotesize SQPI6090 & \footnotesize Square of average inflation rate between 1960 and 1990. \\
\footnotesize PRIGHTS & \footnotesize Political rights index. \\
\footnotesize POP1560 & \footnotesize Fraction of population younger than 15 years in 1960. \\ 
\footnotesize POP60 & \footnotesize Population in 1960 \\
\footnotesize POP6560 & \footnotesize Fraction of population older than 65 years in 1960. \\
\footnotesize PRIEXP70 & \footnotesize Fraction of primary exports in total exports in 1970. \\
\footnotesize PROT00 & \footnotesize Fraction of population Protestant in 1960. \\
\footnotesize RERD & \footnotesize Real exchange rate distortions. \\
\footnotesize REVCOUP & \footnotesize Number of revolutions and military coups. \\
\footnotesize SAFRICA & \footnotesize Dummy for Sub-Saharan African countries. \\
\footnotesize SCOUT & \footnotesize Measure of outward orientation. \\
\footnotesize SIZE60 & \footnotesize Logarithm of aggregate GDP in 1960. \\
\footnotesize SOCIALIST & \footnotesize Dummy for countries under Socialist rule for considerable time during 1950 to 1995. \\
\footnotesize SPAIN & \footnotesize Dummy variable for former Spanish colonies. \\
\footnotesize TOT1DEC1 & \footnotesize Growth of terms of trade in the 1960's. \\
\footnotesize TOTIND & \footnotesize Terms of trade ranking \\
\footnotesize TROPICAR & \footnotesize Proportion of country's land area within geographical tropics. \\
\footnotesize TROPPOP & \footnotesize Proportion of country's population living in geographical tropics. \\
\footnotesize WARTIME & \footnotesize Fraction of time spent in war between 1960 and 1990. \\
\footnotesize WARTORN & \footnotesize Indicator for countries that participated in external war between 1960 and 1990. \\
\footnotesize YRSOPEN & \footnotesize Number of years economy has been open between 1950 and 1994. \\
\footnotesize ZTROPICS & \footnotesize Fraction tropical climate zone. \\
\bottomrule
    \end{tabular}
 \end{table*}
 
  % updated
    \begin{table*}[h]
 \center
  \caption{Selected variables for the SDM data. For proposed method and horseshoe method, we denote by "MSE inc" the expected increase in mean squared error compared to choosing the full model.}
  \label{tab:realData_SDM}
  \footnotesize
  \begin{tabular}{ll}
  \toprule 
 method & selected variables \\
\midrule
  \footnotesize proposed ($\delta$  =  0.8, MSE inc = 110.29\%) & \footnotesize EAST \\
  \footnotesize proposed ($\delta$  =  0.5, MSE inc = 27.07\%) & \footnotesize EAST, MALFAL66 \\
  \footnotesize proposed ($\delta$  =  0.05, MSE inc = 27.25\%) & \footnotesize EAST, MALFAL66 \\
  \footnotesize proposed ($\delta$  =  0.01, MSE inc = 27.2\%) & \footnotesize  EAST, MALFAL66 \\
  \footnotesize proposed ($\delta$  =  0.001, MSE inc = 27.14\%) & \footnotesize EAST, MALFAL66 \\
  \footnotesize proposed ($\delta$  =  0.0, MSE inc = 27.22\%) & \footnotesize EAST, MALFAL66 \\
  \midrule
  \footnotesize horseshoe ($\delta$  =  0.8, MSE inc = 69.31\%) & \footnotesize EAST, GDPCH60L, IPRICE1, P60 \\
  \footnotesize horseshoe ($\delta$  =  0.5, MSE inc = 20.16\%) & \footnotesize CONFUC, EAST, GDPCH60L, IPRICE1, LIFE060, P60, TROPICAR \\
  \footnotesize horseshoe ($\delta$  =  0.05, MSE inc = 0.0\%) & \footnotesize all except DENS65I, DPOP6090, ECORG, EUROPE, HERF00, \\
  & \footnotesize LANDAREA, LANDLOCK, OIL, ORTH00, PI6090, SQPI6090, \\
  & \footnotesize POP6560, SIZE60, TOT1DEC1, TOTIND, WARTIME, WARTORN \\
  \footnotesize horseshoe ($\delta$  =  0.01, MSE inc = 0.0\%) & \footnotesize all except DENS65I, ECORG, LANDAREA, SQPI6090, WARTIME \\ 
  \footnotesize horseshoe ($\delta$  =  0.001, MSE inc = 0.0\%) & \footnotesize  all \\
  \footnotesize horseshoe ($\delta$  =  0.0, MSE inc = 0.0\%) & \footnotesize all  \\
  \midrule
    \midrule
  \footnotesize GibbsBvs & \footnotesize DENS65C, EAST, GDPCH60L, IPRICE1, P60, TROPICAR  \\
\footnotesize EMVS & \footnotesize none \\
\footnotesize SSLASSO & \footnotesize EAST, P60, TROPICAR \\
\footnotesize MOM ($\delta = 0.8$) & \footnotesize EAST, MALFAL66 \\
\footnotesize MOM  ($\delta = 0.5$) & \footnotesize  EAST, MALFAL66 \\
\footnotesize MOM ($\delta = 0.05$) & \footnotesize BUDDHA, CONFUC, EAST, GVR61, IPRICE1, P60, SAFRICA, \\
& \footnotesize TROPICAR \\
\footnotesize AIC & \footnotesize AVELF, BUDDHA, CIV72, CONFUC, DENS65C, EAST, GDPCH60L, \\ 
& \footnotesize  GGCFD3,  GOVNOM1, GVR61, HINDU00, IPRICE1, MALFAL66,  \\
& \footnotesize  MINING, MUSLIM00,  OPENDEC1, OTHFRAC, P60, POP60, RERD,  \\
& \footnotesize  REVCOUP, SAFRICA, SPAIN, TROPICAR, TROPPOP, YRSOPEN \\
\footnotesize EBIC ($\gamma = 0$) & \footnotesize CONFUC, EAST, MALFAL66, P60, TROPPOP, YRSOPEN \\
\footnotesize EBIC ($\gamma = 0.5$) & \footnotesize EAST, TROPPOP, YRSOPEN  \\
\footnotesize EBIC ($\gamma = 1.0$) & \footnotesize EAST, TROPPOP, YRSOPEN  \\
\footnotesize stability ($q = 0.1\cdot d$) & \footnotesize EAST, YRSOPEN \\
\footnotesize stability ($q = 0.5\cdot d$) & \footnotesize none \\
\footnotesize stability ($q = 0.8\cdot d$) & \footnotesize - (did not terminate) \\
       \bottomrule
    \end{tabular}
 \end{table*}
 
    \begin{table*}[h]
 \center
  \caption{Top 10 selected models using the proposed method with $\delta = 0.0$ and $\delta = 0.5$ for the SDM data.}
  \label{tab:rankingModels_SDM}
  \begin{tabular}{ll}
  \toprule 
         model & probability \\
\midrule
\multicolumn{2}{c}{$\delta = 0.5$} \\
\midrule
EAST, MALFAL66 & 0.12 \\
EAST, P60, TROPICAR & 0.035 \\
EAST, P60 & 0.034 \\
EAST, MALFAL66, P60 & 0.026 \\
EAST, TROPICAR & 0.019 \\
EAST, LIFE060 & 0.016 \\
EAST, GDPCH60L, LIFE060, MALFAL66 & 0.012 \\
EAST, IPRICE1, P60, TROPICAR & 0.011 \\
EAST, IPRICE1, P60 & 0.011 \\
EAST, GDPCH60L, IPRICE1, LIFE060 & 0.01 \\
\midrule
      \multicolumn{2}{c}{$\delta = 0.0$} \\
\midrule
EAST, MALFAL66 & 0.055 \\
DENS65C, EAST, GDPCH60L, IPRICE1, P60, TROPICAR & 0.02 \\
EAST, MALFAL66, P60 & 0.015 \\
EAST, P60, TROPICAR & 0.012 \\
EAST, MALFAL66, P60, SPAIN & 0.007 \\
EAST, MALFAL66, SPAIN & 0.007 \\
EAST, GVR61, MALFAL66 & 0.006 \\
EAST, GDPCH60L, LIFE060, MALFAL66 & 0.006 \\
EAST, LIFE060, MALFAL66 & 0.006 \\
EAST, MALFAL66, YRSOPEN & 0.006 \\
       \bottomrule
    \end{tabular}
 \end{table*}

\begin{table*}[h]
 \center
  \caption{Top 30 variable inclusion probabilities using the proposed method with $\delta = 0.0$ and $\delta = 0.5$ for the SDM data. For reference, we also show the results that were reported in \citep{sala2004determinants} using the BACE method.
  Variables in bold mark the 18 variables that were considered as significant in \citep{sala2004determinants}.}
  \label{tab:rankingVars_SDM}
    \begin{tabular}{lll}
  \begin{tabular}{ll}
  \toprule
variable & $\delta = 0.5$  \\
\midrule
\bf EAST & 0.892 \\
\bf P60 & 0.485 \\
\bf MALFAL66 & 0.341 \\
\bf GDPCH60L & 0.34 \\
\bf IPRICE1 & 0.328 \\
\bf TROPICAR & 0.285 \\
\bf LIFE060 & 0.253 \\
\bf CONFUC & 0.113 \\
\bf YRSOPEN & 0.085 \\
\bf SAFRICA & 0.079 \\
RERD & 0.068 \\
\bf DENS65C & 0.063 \\
\bf GVR61 & 0.055 \\
\bf LAAM & 0.045 \\
TROPPOP & 0.045 \\
\bf AVELF & 0.044 \\
\bf MUSLIM00 & 0.043 \\
\bf BUDDHA & 0.042 \\
\bf MINING & 0.04 \\
OTHFRAC & 0.034 \\
\bf SPAIN & 0.032 \\
OPENDEC1 & 0.031 \\
ABSLATIT & 0.029 \\
PRIEXP70 & 0.028 \\
H60 & 0.027 \\
GOVSH61 & 0.024 \\
DENS60 & 0.022 \\
FERTLDC1 & 0.022 \\
POP1560 & 0.018 \\
PROT00 & 0.018 \\
       \bottomrule
    \end{tabular}
    & 
      \begin{tabular}{ll}
  \toprule
variable & $\delta = 0.0$ \\
\midrule
\bf EAST & 0.855 \\
\bf P60 & 0.623 \\
\bf IPRICE1 & 0.523 \\
\bf GDPCH60L & 0.475 \\
\bf MALFAL66 & 0.434 \\
\bf TROPICAR & 0.411 \\
\bf DENS65C & 0.248 \\
\bf LIFE060 & 0.231 \\
\bf CONFUC & 0.189 \\
\bf YRSOPEN & 0.14 \\
\bf SAFRICA & 0.136 \\
\bf LAAM & 0.128 \\
\bf SPAIN & 0.126 \\
\bf GVR61 & 0.109 \\
\bf MINING & 0.097 \\
\bf MUSLIM00 & 0.093 \\
\bf BUDDHA & 0.093 \\
\bf AVELF & 0.089 \\
RERD & 0.086 \\
TROPPOP & 0.071 \\
OPENDEC1 & 0.067 \\
OTHFRAC & 0.063 \\
PRIEXP70 & 0.062 \\
H60 & 0.061 \\
GOVSH61 & 0.058 \\
DENS60 & 0.055 \\
PRIGHTS & 0.053 \\
ABSLATIT & 0.052 \\
PROT00 & 0.048 \\
POP1560 & 0.046 \\
       \bottomrule
    \end{tabular} 
  &    
      \begin{tabular}{ll}
  \toprule
variable & BACE \\
\midrule
\bf EAST & 0.823 \\
\bf P60 & 0.774 \\
\bf IPRICE1 & 0.774 \\
\bf GDPCH60L & 0.685 \\
\bf TROPICAR & 0.563 \\
\bf DENS65C & 0.428 \\
\bf MALFAL66 & 0.252 \\
\bf LIFE060 & 0.209 \\
\bf CUNFUC & 0.206 \\
\bf SAFRICA & 0.154 \\
\bf LAAM & 0.149 \\
\bf MINING & 0.124 \\
\bf SPAIN & 0.123 \\
\bf YRSOPEN & 0.119 \\
\bf MUSLIM00 & 0.114 \\
\bf BUDDHA & 0.108 \\
\bf AVELF & 0.105 \\
\bf GVR61& 0.104 \\
DENS60 & 0.086 \\
RERD & 0.082 \\
OTHFRAC & 0.080 \\
OPENDEC1 & 0.076 \\
PRIGHTS & 0.066 \\
GOVSH61 & 0.063 \\
H60 & 0.061 \\
TROPPOP & 0.058 \\
PRIEXP70 & 0.053 \\
GGCFD3 & 0.048 \\
PROT00 & 0.046 \\
HINDU00 & 0.045 \\
       \bottomrule
    \end{tabular} 
        \end{tabular}
 \end{table*}

\bibliographystyle{plainnat}
\bibliography{paperReferences_new.bib}
\end{document}